\documentclass[11pt]{article}
\usepackage{graphicx,afterpage,times}
\usepackage{lscape}
\usepackage{dsfont}
\usepackage{bm}
\usepackage{wrapfig}
\usepackage{soul}
\usepackage{anyfontsize}
\usepackage{subfig}
\usepackage{algorithm2e}
\usepackage{chngcntr}
\RestyleAlgo{ruled}
\usepackage{enumitem}
\usepackage{cancel}
\usepackage[dvipsnames]{xcolor}
\usepackage{empheq}\usepackage{mathtools}
\usepackage[top=0.5in, bottom=0.8in, left=0.5in, right=0.5in]{geometry}
\usepackage{mathrsfs,color}
\usepackage[normalem]{ulem}
\usepackage{amsfonts}\usepackage{multirow}
\usepackage{amssymb}\usepackage{caption,comment}
\usepackage{amsmath}\usepackage{float}
\usepackage{relsize}
\usepackage{amssymb,amsthm}
\usepackage[toc,page,header,title]{appendix}
\usepackage{graphicx,afterpage,cancel}
\usepackage{epstopdf}
\usepackage{mdframed}
\newcommand{\rmd}{\mathrm{d}}

\usepackage[pdftex,colorlinks=true]{hyperref}
\hypersetup{CJKbookmarks,%
bookmarksnumbered,%
colorlinks,%
linkcolor=black,%
citecolor=blue,%
plainpages=false,%
pdfstartview=FitH}
\numberwithin{equation}{section}
\numberwithin{figure}{section}

\newcommand\tabcaption{\def\@captype{table}\caption}

\newtheoremstyle{theorem}%
  {3pt}
  {3pt}
  {}
  {}
  {\bfseries\color{red}}
  {\textcolor{red}{.}}
  {.5em}
  {}
\theoremstyle{theorem}
\newtheorem{thm}{Theorem}[section]

\newtheoremstyle{lemma}%
  {3pt}
  {3pt}
  {}
  {}
  {\bfseries\color{blue}}
  {\textcolor{blue}{.}}
  {.5em}
  {\thmname{#1} A\thmnumber{#2}\thmnote{ (#3)}}
\theoremstyle{lemma}

\newtheoremstyle{definition}%
  {3pt}
  {3pt}
  {}
  {}
  {\bfseries\color{green}}
  {\textcolor{green}{.}}
  {.5em}
  {\thmname{#1} A\thmnumber{#2}\thmnote{ (#3)}}
\theoremstyle{definition}

\newcommand\mybar{\kern1pt\rule[-\dp\strutbox]{1pt}{\baselineskip}\kern1pt}

\definecolor{orange}{RGB}{255,127,0}

\newcommand{\blue}[1]{\textcolor{blue}{#1}}

\def\d{{\, \rm d}}

\newcommand{\ml}{\boldsymbol{\Lambda}}
\newcommand{\ms}{\boldsymbol{\Sigma}}
\newcommand{\mr}[1]{\mathbf{R}_{\text{#1}}}
\newcommand{\mrt}[1]{\mathbf{R}_{\text{#1}}(t)}

\newcommand{\vx}{\mathbf{x}}
\newcommand{\vy}{\mathbf{y}}
\newcommand{\ve}{\boldsymbol{\epsilon}}
\newcommand{\vz}{\mathbf{z}}

\newcommand{\vh}{\mathbf{h}}
\newcommand{\vf}{\mathbf{f}}
\newcommand{\vxt}{\mathbf{x}(t)}
\newcommand{\vyt}{\mathbf{y}(t)}
\newcommand{\vzt}{\mathbf{z}(t)}

\newcommand{\vm}[1]{\boldsymbol{\mu}_{\text{#1}}}
\newcommand{\vmt}[1]{\boldsymbol{\mu}_{\text{#1}}(t)}
\newcommand{\smooth}[1]{\overleftarrow{#1}}
\newcommand{\dt}{\Delta t}
\newcommand{\vw}{\mathbf{W}}
\newcommand{\vwt}[1]{\mathbf{W}_{#1}(t)}
\newcommand{\va}{\boldsymbol{\alpha}}
\newcommand{\ma}{\mathbf{A}}
\newcommand{\mb}{\mathbf{B}}
\newcommand{\mc}{\boldsymbol{\Gamma}}
\newcommand{\cF}{\mathcal{F}}
\newcommand{\pp}{\mathbb{P}}
\newcommand{\rr}{\mathbb{R}}
\newcommand{\tran}{\mathtt{T}}
\newcommand{\ee}[1]{\mathbb{E}\left[#1\right]}
\newcommand{\nf}{\normalfont{f}}
\newcommand{\ns}{\normalfont{s}}

\makeatletter
\renewcommand*{\@cite@ofmt}{\bfseries\hbox}

\makeatother

\allowdisplaybreaks

\title{An Adaptive Online Smoother with Closed-Form Solutions and Information-Theoretic Lag Selection for Conditional Gaussian Nonlinear Systems}
\author{Marios Andreou\textsuperscript{1, *}, Nan Chen\textsuperscript{1, \textdagger} \& Yingda Li\textsuperscript{1, \textdaggerdbl}}
\date{%
    \footnotesize
    {}\textsuperscript{1}Department of Mathematics, University of Wisconsin--Madison, 480 Lincoln Drive, Madison, WI 53706, USA\\%
    {}\textsuperscript{*}mandreou@math.wisc.edu (Corresponding Author)\\
    {}\textsuperscript{\textdagger}chennan@math.wisc.edu\\
    {}\textsuperscript{\textdaggerdbl}zjkliyingda@gmail.com\\[2ex]%
    \normalsize
    \today
}
\begin{document}

\maketitle

\begin{abstract}
    Data assimilation (DA) combines partial observations with dynamical models to improve state estimation. Filter-based DA uses only past and present data and is the prerequisite for real-time forecasts. Smoother-based DA exploits both past and future observations. It aims to fill in missing data, provide more accurate estimations, and develop high-quality datasets. However, the standard smoothing procedure requires using all historical state estimations, which is storage-demanding, especially for high-dimensional systems. This paper develops an adaptive-lag online smoother for a large class of complex dynamical systems with strong nonlinear and non-Gaussian features, which has important applications to many real-world problems. The adaptive lag allows the utilization of observations only within a nearby window, thus reducing computational complexity and storage needs. Online lag adjustment is essential for tackling turbulent systems, where temporal autocorrelation varies significantly over time due to intermittency, extreme events, and nonlinearity. Based on the uncertainty reduction in the estimated state, an information criterion is developed to systematically determine the adaptive lag. Notably, the mathematical structure of these systems facilitates the use of closed analytic formulae to calculate the online smoother and adaptive lag, avoiding empirical tunings as in ensemble-based DA methods. The adaptive online smoother is applied to studying three important scientific problems. First, it helps detect online causal relationships between state variables. Second, the advantage of reduced computational storage expenditure is illustrated via Lagrangian DA, a high-dimensional nonlinear problem. Finally, the adaptive smoother advances online parameter estimation with partial observations, emphasizing the role of the observed extreme events in accelerating convergence.
\end{abstract}

\section{Introduction} \label{sec:1}

Complex turbulent nonlinear dynamical systems (CTNDSs) have broad applications across various fields \cite{frisch1995turbulence, majda2006nonlinear, vallis2017atmospheric}. These systems are characterized by their high dimensionality and multiscale structures, with strong nonlinear interactions occurring between state variables at different spatiotemporal scales. Extreme events, intermittency, and non-Gaussian probability density functions (PDFs) are also commonly observed \cite{farazmand2019extreme, denny2009prediction, mohamad2015probabilistic}. 

State estimation in CTNDSs is essential for parameter estimation, prediction, optimal control, and generating complete datasets \cite{evensen2022data, stengel1994optimal, ruiz2013estimating}. However, the turbulent nature of the dynamics can amplify small errors in the model structure, spatiotemporal solutions, or initial conditions when relying solely on forecasts. Data assimilation (DA), which integrates observations with system dynamics, is widely used to improve state estimation \cite{kalnay2003atmospheric, lahoz2010data, majda2012filtering, evensen2009data, law2015data}. Given the inevitable uncertainty in state estimation, especially for the unobserved variables of CTNDSs, probabilistic approaches via Bayesian inference are natural choices. The model provides a prior distribution, while observations inform the likelihood. These are then combined to form the posterior distribution for state estimation.

DA can be classified into two categories, based on when the observational data are incorporated. Filtering uses observations only up to the current time \cite{majda2012filtering, evensen2009data, law2015data}. Serving as the initialization, filter-based state estimation is a prerequisite for real-time forecasts. In contrast, smoothing leverages data from the entire observation period \cite{law2015data, rauch1965maximum, chen2020efficient, sarkka2023bayesian}, including future data, which makes it highly effective for optimal state estimation in offline data postprocessing. This helps to fill in missing values, minimize bias, and create complete datasets \cite{uppala2005era}. With the extra information from future observations, smoothing often produces more accurate and less uncertain state estimates than filtering. When the system dynamics and observational mappings are linear, with additive Gaussian noise, the corresponding filtering and smoothing methods are the Kalman(--Bucy) filter and the Rauch--Tung--Striebel (RTS) smoother, respectively \cite{rauch1965maximum, kalman1961new, bucy1987filtering}, where the posterior distribution is Gaussian and can be computed using closed-form analytical solutions.

Due to the intrinsic nonlinear dynamics and non-Gaussian statistics of CTNDSs, analytic solutions for DA are rarely available. As a result, various numerical and approximate methods have been developed, including the ensemble Kalman filter/smoother and the particle (or sequential Monte Carlo) filter/smoother \cite{evensen2009data, law2015data, anderson2001ensemble, delmoral1997nonlinear, liu1998sequential, kitagawa1996monte}. These methods are widely used but often face tremendous computational costs, especially in high-dimensional systems \cite{kuo2005lifting}, which limits the number of particles or ensemble members, potentially causing biases and numerical instabilities \cite{gottwald2013mechanism, harlim2010catastrophic, snyder2008obstacles}. Empirical tuning techniques, such as noise inflation, localization, and resampling, are widely used in practice to mitigate these issues \cite{anderson2001ensemble, chen2020predicting, hol2006resampling, greybush2011balance}. However, these ad hoc tuning methods are usually quite challenging to implement systematically. Closed-form analytic solutions for DA are thus highly desirable, as they improve computational efficiency, stability, and accuracy, especially in capturing non-Gaussian features, including intermittency and extreme events. They also facilitate the theoretical analysis of the error and uncertainty during state estimation.

Instead of refining DA schemes directly, computational challenges in state estimation can be addressed by developing approximate models that yield analytic solutions for the posterior distribution. While linear approximations allow for standard methods like the Kalman filter or RTS smoother, linearizing a strongly nonlinear system often leads to biases and instabilities. An alternative is a recently developed class of nonlinear systems that includes many turbulent models in geophysics, fluids, engineering, and neuroscience \cite{liptser2001statisticsI, liptser2001statisticsII, chen2018conditional}. Despite their nonlinear dynamics and non-Gaussian statistics, the conditional distributions of the unobserved state variables given the observations, which are precisely the posterior distributions in the DA context, are Gaussian, leading to the term conditional Gaussian nonlinear systems (CGNSs). The CGNS framework allows the use of closed analytic formulae for solving these conditional distributions, helping develop efficient algorithms for filtering, smoothing, and sampling without the ad hoc tuning often needed in ensemble-based DA methods. CGNSs have been utilized as surrogate models in various applications, including DA, prediction, preconditioning, and machine learning \cite{chen2022conditional, chen2024cgnsde, chen2014predicting}.

The standard smoother-based state estimation procedure involves executing a forward pass for filtering across the entire observational period, followed by a backward pass for smoothing \cite{rauch1965maximum, chen2020efficient}. However, the standard offline smoother requires storing the filter solution for the entire duration before initiating the backward pass, which requires substantial computational storage, particularly in high-dimensional systems. Due to the wide application of smoother-based state estimation, it is of practical importance to develop a computationally efficient and accurate algorithm that has the potential to significantly reduce storage requirements.

This paper develops a forward-in-time online smoother algorithm with adaptive lags for the CGNS framework, eliminating the need for a full backward pass. The online smoother sequentially updates the current state as new observations become available. By doing so, it effectively addresses the computational storage issue. While online schemes exist for the RTS smoother and ensemble-based methods, the CGNS online smoother has several unique advantages. First, despite the intrinsic nonlinearity of the underlying dynamics, closed analytic formulae are available to compute the nonlinear online smoother, providing precise and accurate solutions which avoid numerical and sampling errors as in ensemble-based methods. Second, due to the turbulent nature of the system, observations influence the estimated state only within a short time window, which enhances computational efficiency and reduces storage needs. Third, different from fixed-lag smoothers \cite{kitagawa1996monte, cappe2010inference, olsson2008sequential}, the lag in the CGNS smoother is adaptively determined. Online lag adjustment is essential for studying turbulent systems, where temporal autocorrelation varies significantly over time due to intermittency, extreme events, and nonlinearity. A fixed lag usually either overuses storage (when overestimated) or introduces a large bias (if underestimated). In contrast, an adaptive lag implicitly optimizes the use of data and computational storage. Finally, the adaptive lag is systematically determined using an information criterion based on the uncertainty reduction in the posterior distribution. It emphasizes the importance of the posterior uncertainty and differs from some of the existing adaptive lag selection criteria that rely solely on the posterior mean \cite{poddar2022adaptive}. As closed analytic formulae are available for posterior distributions, the information gain can be computed efficiently and accurately.

The adaptive online smoother for CGNSs is employed to study three important scientific problems. First, the online update of the smoother estimate allows for quantification of the improvement in state estimation by incorporating future information. This facilitates the inference of causal relationships between the state variables. A nonlinear dyad model with strong non-Gaussian features is utilized for such a study. Second, the CGNS framework is applied to Lagrangian data assimilation, which is a high-dimensional nonlinear problem that has a significant storage requirement \cite{chen2014information, chen2015noisy, chen2024lagrangian}. The online smoother allows for the estimation of the unobserved flow states based on Lagrangian tracers. This study highlights the role of the adaptive-lag online smoother in reducing computational storage needs, when compared to its fixed-lag variant. Finally, the online smoother facilitates developing an online parameter estimation algorithm with partial observations. It helps reveal the role of the observed intermittent extreme events in advancing parameter estimation.

The remainder of this paper is organized as follows. Section \ref{sec:2} introduces the CGNS modeling framework, including the equations for the optimal nonlinear filter and offline smoother state estimation. Section \ref{sec:3} presents the adaptive online smoother. In Section \ref{sec:4}, the application of the adaptive online smoother to the three key problems is demonstrated. Section \ref{sec:5} includes the conclusion. The appendices contain the analysis and proofs, as well as additional details.

\section{The Conditional Gaussian Nonlinear System Modeling Framework} \label{sec:2}

Throughout this paper, \textbf{boldface} letters are exclusively used to denote vectors for the sake of mathematical clarity. In this regard, we use \textbf{l}owercase boldface letters to denote column vectors, while \textbf{U}ppercase boldface letters denote matrices. The only exception to this rule is $\vw$ (with some subscript or superscript), which denotes a Wiener process. Although in this work this always corresponds to a column vector, we instead use an uppercase letter due to literary tradition.

Let $t$ denote the time variable, with $t\in[0,T]$, where $T>0$ is allowed to be infinite. Let $(\Omega, \cF, \pp)$ be a complete probability space and $\{\cF_t\}_{t\in[0,T]}$ to be a filtration of sub-$\sigma$-algebras of $(\Omega, \cF)$, which we assume is augmented (i.e., complete and right-continuous). For every filtration there exists a smallest such augmented filtration refining $\{\cF_t\}_{t\in[0,T]}$ (known as its completion), so this is without loss of generality. We let $\big(\mathbf{x}(t,\omega),\mathbf{y}(t,\omega)\big)$, for $t\in[0,T]$ and $\omega\in\Omega$ (for the rest of this work we drop the event or sample space dependence for notational simplicity, but it is always implied), be a partially observable $(S,\mathcal{A})$-valued stochastic process, where $\vx$ is the observable component, while $\vy$ is the unobservable component. The theory that follows can be applied mutatis mutandis to any partially observable stochastic process that takes values over a measurable space $(S,\mathcal{A})$ where $S$ is a separable Hilbert space (finite-dimensional or not) over a complete scalar ground field and $\mathcal{A}$ is a $\sigma$-algebra of $S$, but for this work it suffices to consider complex-valued finite-dimensional processes with respect to the usual Euclidean inner product. As such, we let $S=\mathbb{C}^{k+l}$ and $\mathcal{A}=\mathcal{B}_{\mathbb{C}^{k+l}}\equiv\mathcal{B}_{\rr^{2(k+l)}}$, with $\vx$ being a $k$-dimensional vector and $\vy$ an $l$-dimensional one, where $\mathcal{B}_{\rr^{2(k+l)}}$ is the Borel $\sigma$-algebra of $\rr^{2(k+l)}$, since $\mathrm{dim}_{\rr}(\mathbb{C}^{k+l})=2(k+l)$. We assume that
\begin{equation*}
    (\vx,\vy)=\big((x_1(t),\ldots,x_k(t), y_1(t),\ldots,y_l(t)),\cF_t\big), \ t\in[0,T],
\end{equation*}
meaning the partially observable random process is (jointly) adapted to the filtration $\{\cF_t\}_{t\in[0,T]}$, i.e., for all times $t\in[0,T]$ the random vector defined by $\big(\vx(t,\cdot)^\tran,\vy(t,\cdot)^\tran\big)^\tran:\Omega\to S$ is an $(\cF_t;\mathcal{A})$-measurable function. Specifically, this implies that the natural filtration of $\cF$ with respect to $\{\vx(s)\}_{s\leq t}$, which is the sub-$\sigma$-algebra generated by the observable processes for times $s\leq t$ and is defined by
\begin{equation*}
    \cF_t^\vx:=\sigma\big(\{ \mathbf{x}(s)\}_{s\leq t}\big)=\big\{\mathbf{x}(s)^{-1}[A]=\mathbf{x}(s,\cdot)^{-1}[A]: A\in\mathcal{A}, \ s\leq t\big\},
\end{equation*}
satisfies $\cF_t^\vx\subseteq \cF_t$, since $\vx$ is adapted to the filtration $\{\cF_t\}_{t\in[0,T]}$ by construction and by definition $\cF_t^\vx$ is the smallest such filtration. We call this natural filtration the observable $\sigma$-algebra (at time $t$) for the remainder of this work.

It is known that, at each time instant $t\in[0,T]$, the optimal estimate in the mean-square sense of some measurable function $\mathbf{h}(t,\vx,\vy)$ on the basis of the observations up to time $T'\in[t,T]$, $\{\mathbf{x}(s)\}_{s\leq T'}$, is exactly its conditional expectation conditioned on the observable $\sigma$-algebra at time $T'$, $\mathbb{E}\big[\vh(t,\vx,\vy)|\cF_{T'}^\vx\big]$. This is known as the a-posteriori mean and this assertion of optimality rests on the tacit assumption that $\mathbb{E}\big[\left\|\vh(\vx,\vy)\right\|_2^2\big]$ is finite, where $\left\|\cdot\right\|_2$ denotes the usual Euclidean norm over $\mathbb{C}^{\mathrm{dim}(\mathbf{h})}$ \cite{liptser2001statisticsI, liptser2001statisticsII}. Usually $\vh$ is a function of the unobserved process $\vy$ and in this work we exclusively use $\vh=\vy$ as to recover the optimal filter and smoother conditional statistics of the hidden process when conditioning on the observations.

The goal of optimal state estimation is to characterize the posterior states using a system of stochastic differential equations (SDEs), known as the optimal nonlinear filter or smoother equations depending on whether we condition on current observations or the entire observational period, respectively. In general, without making specific assumptions about the structure of the processes $\vh$ and $\vx$, determining $\ee{\vh|\cF_{T'}^\vx}$ is challenging. CGNSs resolve this by having the key advantage of conditional Gaussian posterior distributions, which can be written down using closed analytic formulae. The linear state-observation system filtered by the classical Kalman--Bucy filter \cite{kalman1961new, bucy1987filtering} is the simplest example of a CGNS \cite{liptser2001statisticsII, chen2016nonlinear}. Despite the conditional Gaussianity, these coupled systems remain highly nonlinear, with the associated marginal and joint distributions being highly non-Gaussian, which allows the systems to capture many realistic turbulent features.

\subsection{Conditionally Gaussian Processes} \label{sec:2.1}

In its most general form, a conditional Gaussian system of processes consists of two diffusion-type processes defined by the following nonlinear system of stochastic differentials given in Itô form \cite{liptser2001statisticsII, chen2018conditional}:
\begin{align}
    \rmd\vxt &= \big(\ml^\vx(t,\vx)\vyt+\vf^\vx(t,\vx)\big)\rmd t+\ms_1^\vx(t,\vx)\rmd \vwt{1}+\ms_2^\vx(t,\vx)\rmd \vwt{2}, \label{eq:condgauss1}\\
    \rmd\vyt &= \big(\ml^\vy(t,\vx)\vyt+\vf^\vy(t,\vx)\big)\rmd t+\ms_1^\vy(t,\vx)\rmd \vwt{1}+\ms_2^\vy(t,\vx)\rmd \vwt{2}, \label{eq:condgauss2}
\end{align}
where
\begin{equation*}
    \vw_1=\big((W_{11}(t),\ldots,W_{1d}(t)),\cF_t\big) \quad\text{and}\quad \vw_2=\big((W_{21}(t),\ldots,W_{2r}(t)),\cF_t\big),
\end{equation*}
are two mutually independent complex-valued Wiener processes (i.e., both their real and imaginary parts are mutually independent real-valued Wiener processes with standardized covariances) and almost every path of $\vx$ and $\vy$ is in $C^0\big([0,T];\mathbb{C}^k\big)$ and $C^0\big([0,T];\mathbb{C}^l\big)$, respectively. The elements of the vector- and matrix-valued functions of multiplicative factors ($\ml^\vx,\ml^\vy$), forcings ($\vf^\vx,\vf^\vy$), and noise feedbacks ($\ms_1^\vx,\ms_2^\vx,\ms_1^\vy,\ms_2^\vy$) are assumed to be nonanticipative (adapted) functionals over the measurable time-function cylinder
\begin{equation*}
    \big(C^{0,k}_T,\mathscr{B}^k_T\big):=\left([0,T]\times C^0\big([0,T];\mathbb{C}^k\big), \mathcal{B}\big([0,T]\big)\otimes \mathcal{B}\big(C^0\big([0,T];\mathbb{C}^k\big)\big)\right),
\end{equation*}
where $\otimes$ denotes the tensor-product $\sigma$-algebra on the underlying product space, i.e.,
\begin{equation*}
    \mathscr{B}^k_T=\mathcal{B}\big([0,T]\big)\otimes\mathcal{B}\big(C^0\big([0,T];\mathbb{C}^k\big)\big)=\sigma\left(\big\{A\times B: A\in\mathcal{B}\big([0,T]\big), B\in\mathcal{B}\big(C^0\big([0,T];\mathbb{C}^k\big)\big)\big\}\right),
\end{equation*}
with $\mathcal{B}\big(C^0\big([0,T];\mathbb{C}^k\big)\big)$ being the $\sigma$-algebra generated by the topology of compact convergence on the space of continuous functions from $[0,T]$ to $\mathbb{C}^k$, $C^0\big([0,T];\mathbb{C}^k\big)$. It is important to emphasize here the fact that in a CGNS, the unobservable component $\vy$ enters the dynamics in a conditionally linear manner, whereas the observable process $\vx$ can enter into the coefficients of both equations in any measurably nonlinear way.

Many CTNDSs fit into the CGNS modeling framework. Some well-known classes of these systems are physics-constrained nonlinear stochastic models (for example the noisy versions of the Lorenz models, low-order models of Charney--DeVore flows, and a paradigm model for topographic mean flow interaction) \cite{majda2012physics, harlim2014ensemble}, stochastically coupled reaction-diffusion models used in neuroscience and ecology (for example the stochastically coupled FitzHugh--Nagumo models and the stochastically coupled SIR epidemic models), and spatiotemporally multiscale models for turbulence, fluids, and geophysical flows (for example the Boussinesq equations with noise and the stochastically forced rotating shallow-water equation) \cite{chen2018conditional}. This modeling framework has also been exploited to develop realistic low-order stochastic models for the Madden--Julian oscillation (MJO) and Arctic sea ice \cite{chen2014predicting, chen2022efficient}.

In addition to modeling many natural phenomena, the CGNS framework and its closed analytic DA formulae have been applied to study many theoretical and practical problems. It has been utilized to develop a nonlinear Lagrangian data assimilation algorithm, allowing rigorous analysis to study model error and uncertainty, as well as recovery of turbulent ocean flows with noisy observations from Lagrangian tracers \cite{chen2014information, chen2015noisy, chen2024lagrangian}. The analytically solvable DA scheme has also been applied to the prediction and state estimation of the non-Gaussian intermittent time series of the MJO and the monsoon intraseasonal variabilities, in addition to the filtering of the stochastic skeleton model for the MJO \cite{chen2014predicting, chen2016filtering, chen2016nonlinear}. Notably, the efficient DA procedure also helps develop a rapid algorithm to solve the high-dimensional Fokker--Planck equation \cite{chen2017beating, chen2018efficient}. Worth highlighting is that the ideas of the CGNS modeling framework and the associated DA procedures have also been applied to develop cheap exactly-solvable forecast models in dynamic stochastic superresolution of sparsely observed turbulent systems \cite{branicki2013dynamic, keating2012new}, build stochastic superparameterization for geophysical turbulence \cite{majda2014new}, and design efficient multiscale DA schemes via blended particle filters for high-dimensional chaotic systems \cite{majda2014blended}.

A set of sufficient regularity conditions needs to be assumed a-priori so that the main results of the CGNS framework can be established. We enforce the same set of assumptions as in the work of Andreou \& Chen \cite{andreou2024martingale}, which we outline in Appendix \ref{sec:app1}, with the rationale behind the adoption of each one being provided there (and in the references therein). These conditions are sufficient to show that the posterior distributions of the unobserved variables, when conditioning on the observational data, are Gaussian, as stated in the following theorem; this is exactly why \eqref{eq:condgauss1}--\eqref{eq:condgauss2} is called a CGNS. For the following, we abuse notation for clarity and instead write $\big(\cdot\big|\vx(s),s\leq t\big)$ to indicate the fact that we are conditioning on the observable $\sigma$-algebra at time $t$, $\big(\cdot\big|\cF_t^\vx\big)$.

\begin{thm}[\textbf{Conditional Gaussianity}] \label{thm:condgaussianity}
    Let $\big(\vxt,\vyt\big)$ satisfy \eqref{eq:condgauss1}--\eqref{eq:condgauss2} and assume that the regularity conditions \textbf{\blue{(1)}}--\textbf{\blue{(8)}} in Appendix \ref{sec:app1} hold. Additionally, assume that the initial conditional distribution $\pp\big(\vy(0)\leq \boldsymbol{\alpha}_0\big|\vx(0)\big)$\footnote{The event $\big\{\vy(s)\leq \boldsymbol{\alpha}=(\alpha_{1},\ldots,\alpha_{l})^\tran\big\}$ is to be understood coordinate-wise: $\big\{j=1,\ldots,l: \big(\mathrm{Re}(y_j(s))\leq \mathrm{Re}(\alpha_{j}),\mathrm{Im}(y_j(s))\leq \mathrm{Im}(\alpha_{j})\big)\big\}$.} is ($\pp$-almost surely) Gaussian, $\mathcal{N}_l\big(\vm{f}(0),\mr{f}(0)\big)$\footnote{Whenever we refer to a complex-valued multivariate Gaussian or normal distribution in this work, we mean a circularly-symmetric one (i.e., with a zero pseudo-covariance or relation matrix) \cite{lapidoth2017foundation}.}, and mutually independent from the Wiener processes $\vw_1$ and $\vw_2$, where
    \begin{equation*}
        \vm{f}(0):=\mathbb{E}\big[\vy(0)\big|\vx(0)\big] \quad\text{and}\quad \mr{f}(0):=\mathbb{E}\big[(\vy(0)-\vm{f}(0))(\vy(0)-\vm{f}(0))^\dagger\big|\vx(0)\big],
    \end{equation*}
    with $\cdot^\dagger$ denoting the Hermitian transpose operator. Furthermore, assume $\pp\big(\mathrm{tr}(\mr{f}(0))<+\infty\big)=1$, where $\mathrm{tr}(\cdot)$ denotes the trace operator, meaning that the initial estimation mean-square error between $\vy(0)$ and $\vm{f}(0)$ is almost surely finite. Then, for any $t_j$ such that $0\leq t_1 < t_2 < \cdots < t_n\leq t$, with $t\in[0,T]$, and $\va_1,\ldots,\va_n\in\mathbb{C}^l$, the conditional distribution
    \begin{equation*}
        \pp\big(\vy(t_1)\leq \va_1,\ldots,\vy(t_n)\leq \va_n\big|\vx(s),s\leq t\big),
    \end{equation*}
    is ($\pp$-almost surely) Gaussian.
\end{thm}

\begin{proof}[\textbf{\underline{Proof}}] This is Theorem 12.6 in Liptser \& Shiryaev \cite{liptser2001statisticsII}, which is the multi-dimensional analog of Theorem 11.1. Thorough details are also provided in Kolodziej \cite{kolodziej1980state}. For the analogous result in the case of discrete time (i.e., where the CGNS consists of stochastic difference equations and $\vx$ is observed at discrete times instants), see Theorem 13.3 of Liptser \& Shiryaev \cite{liptser2001statisticsII}, with the respective sufficient assumptions given in Subchapter 13.2.1.

The proof, regardless of continuous- or discrete-time, uses the conditional characteristic function method and a conditional version of the law-uniqueness theorem \cite{kolodziej1980state, yuan2016some}.
\end{proof}

\subsection{Analytically Solvable Filter and Smoother Posterior Distributions} \label{sec:2.2}

With Theorem \ref{thm:condgaussianity} established, it is then possible to yield the optimal nonlinear filter state estimation equations as showed in the following theorem, where the subscript ``$\,\text{\normalfont{f}}\,$" is used to denote the filter conditional Gaussian statistics, which appropriately stands for filter. The filter conditional Gaussian statistics are also known as the filter posterior mean and filter posterior covariance under the Bayesian inference dynamics framework.

\begin{thm}[\textbf{Optimal Nonlinear Filter State Estimation Equations}] \label{thm:filtering}
    Let the assumptions of Theorem \ref{thm:condgaussianity} and the additional regularity conditions \textbf{\blue{(9)}}, \textbf{\blue{(11)}}, and \textbf{\blue{(12)}}, which are outlined in Appendix \ref{sec:app1}, to hold. Then, for any $t\in[0,T]$, the $\cF^\vx_t$-measurable Gaussian statistics of the Gaussian conditional distribution
    \begin{equation*}
        \pp\big(\vyt\leq \va\big|\vx(s),s\leq t\big) \overset{\d}{\sim}\mathcal{N}_l\big(\vmt{\nf},\mrt{\nf}\big),
    \end{equation*}
    defined as
    \begin{equation*}
        \vmt{\nf}:=\mathbb{E}\big[\vyt\big|\vx(s),s\leq t\big] \quad\text{and}\quad \mrt{\nf}:=\mathbb{E}\big[(\vyt-\vmt{\nf})(\vyt-\vmt{\nf})^\dagger\big|\vx(s),s\leq t\big],
    \end{equation*}
    are the unique continuous solutions of the system of optimal nonlinear filter equations:
    \begin{align}
    \d \vmt{\nf}&=(\ml^\vy\vm{\nf}+\vf^\vy)\d t+(\ms^\vy\circ \ms^\vx+\mr{\nf}(\ml^\vx)^\dagger)(\ms^\vx\circ\ms^\vx)^{-1}(\d \vxt -(\ml^\vx\vm{\nf}+\vf^\vx)\d t), \label{eq:filter1}\\
    \d \mrt{\nf}&=\big(\ml^\vy\mr{\nf}+\mr{\nf}(\ml^\vy)^\dagger+\ms^\vy\circ\ms^\vy-(\ms^\vy\circ \ms^\vx+\mr{\nf}(\ml^\vx)^\dagger)(\ms^\vx\circ\ms^\vx)^{-1}(\ms^\vx\circ \ms^\vy+\ml^\vx\mr{\nf})\big)\d t\label{eq:filter2},
    \end{align}
    with initial conditions $\vm{\nf}(0)=\ee{\vy(0)\big|\vx(0)}$ and $\mr{\nf}(0)=\ee{(\vy(0)-\vm{\nf}(0))(\vy(0)-\vm{\nf}(0))^\dagger\big|\vx(0)}$, where the noise interactions through the Gramians (with respect to rows) are defined as
    \begin{gather*}
        (\ms^\vx\circ \ms^\vx)(t,\vx):=\ms_1^\vx(t,\vx)\ms_1^\vx(t,\vx)^\dagger+\ms_2^\vx(t,\vx)\ms_2^\vx(t,\vx)^\dagger, \\
        (\ms^\vy\circ \ms^\vy)(t,\vx):=\ms_1^\vy(t,\vx)\ms_1^\vy(t,\vx)^\dagger+\ms_2^\vy(t,\vx)\ms_2^\vy(t,\vx)^\dagger, \\
        (\ms^\vx\circ \ms^\vy)(t,\vx):=\ms_1^\vx(t,\vx)\ms_1^\vy(t,\vx)^\dagger+\ms_2^\vx(t,\vx)\ms_2^\vy(t,\vx)^\dagger,\quad (\ms^\vy\circ \ms^\vx)(t,\vx):=(\ms^\vx\circ \ms^\vy)(t,\vx)^\dagger.
    \end{gather*}
    Furthermore, if the initial covariance matrix $\mr{\nf}(0)$ is positive-definite ($\pp$-almost surely),  then all the matrices $\mrt{\nf}$, for $t\in[0,T]$, remain positive-definite ($\pp$-almost surely).
\end{thm}
\begin{proof}[\textbf{\underline{Proof}}] This is Theorem 12.7 in Liptser \& Shiryaev \cite{liptser2001statisticsII}, which is the multi-dimensional analog of Theorems 12.1 and 12.3. Thorough details are also provided in Kolodziej \cite{kolodziej1980state}. For the analogous result in the case of discrete time, see Theorem 13.4 of Liptser \& Shiryaev \cite{liptser2001statisticsII}, which outlines the corresponding optimal recursive nonlinear filter difference equations, with the respective sufficient assumptions given in Subchapter 13.2.1. For a martingale-free proof see Theorem 2 in Andreou \& Chen \cite{andreou2024martingale} (although it additionally requires assumption \textbf{\blue{(10)}} found in Appendix \ref{sec:app1}, as to pass from the discrete-time filter to the continuous one via a formal limit).
\end{proof}

The form of the filter mean equation can be intuitively explained within the Bayesian inference framework for DA. The first two terms on the right-hand side of \eqref{eq:filter1}, namely $\ml^\vy\vm{f}+\vf^\vy$,  represent the forecast mean, derived from the process of the unobserved variables in \eqref{eq:condgauss2}. The remaining terms account for correcting the prior mean state, incorporating information from partial observations. The matrix factor in front of the innovation, or pre-fit measurement residual $\d \vxt -(\ml^\vx\vm{\nf}+\vf^\vx)\d t$, is analogous to the Kalman gain in classical Kalman filter theory: ${(\ms^\vy\circ \ms^\vx+\mr{\nf}(\ml^\vx)^\dagger)(\ms^\vx\circ\ms^\vx)^{-1}}$. This factor determines the weight of the observations when updating the model-predicted state. Even in cases where the variables $\vy$ of the unobserved process do not explicitly depend on $\vx$, such as in the Lagrangian DA, the observational process \eqref{eq:condgauss1} still couples the observed and unobserved components. This coupling allows the observations to impact state estimation and correct the forecast. Finally, observe that the equation driving the evolution of the covariance tensor $\mr{f}$ is a random Riccati equation \cite{kandil2003matrix, bishop2019stability}, since the coefficients depend on the observable random variables $\vx$.

As already stated, the CGNS framework also enjoys closed analytic formulae for the recovery of the optimal smoother state estimation. The smoother posterior distribution at time $t\in[0,T]$ exploits the observational information in the entire period $[0,T]$ and therefore it allows for a more accurate and less uncertain estimated state compared to the filter solution. Solving the optimal nonlinear smoother equations requires applying a forward pass (filtering) from $t=0$ to $t=T$, which is then followed by a backward pass (smoothing) from $t=T$ to $t=0$ \cite{law2015data, rauch1965maximum, chen2020efficient, sarkka2023bayesian}. In the theorem that follows, the optimal nonlinear smoother equations, which run backwards in time, are showed, with the subscript ``$\,\text{\normalfont{s}}\,$" denoting the smoother conditional Gaussian statistics, which appropriately stands for smoother. The smoother conditional Gaussian statistics are also known as the smoother posterior mean and smoother posterior covariance under the Bayesian inference dynamics framework.

\begin{thm}[\textbf{Optimal Nonlinear Smoother State Estimation Backward Equations}] \label{thm:smoothing}
    Let the assumptions of Theorem \ref{thm:condgaussianity} and the additional regularity conditions \textbf{\blue{(9)}}, \textbf{\blue{(11)}}, and \textbf{\blue{(12)}}, which are outlined in Appendix \ref{sec:app1}, to hold. In addition, assume
    \begin{equation} \label{eq:positivedefinite}
        \pp\left( \underset{t\in[0,T]}{\mathrm{inf}}\{ \mathrm{det}(\mr{\nf}(t))\}>0\right)=1,
    \end{equation}
    and define the following auxiliary matrices:
    \begin{align}
        \ma(t,\vx)&:=\ml^\vy(t,\vx)-(\ms^\vy\circ \ms^\vx)(t,\vx)(\ms^\vx\circ \ms^\vx)^{-1}(t,\vx)\ml^\vx(t,\vx)\in\mathbb{C}^{l \times l}, \label{eq:auxiliarymata}\\
        \mb(t,\vx) &:=  (\ms^\vy\circ \ms^\vy)(t,\vx)-(\ms^\vy\circ \ms^\vx)(t,\vx)(\ms^\vx\circ \ms^\vx)^{-1}(t,\vx)(\ms^\vx\circ \ms^\vy)(t,\vx)\in\mathbb{C}^{l \times l}. \label{eq:auxiliarymatb}
    \end{align}
    Then, for any $T\geq t\geq 0$ ($t$ running backward), the $\cF_T^\vx$-measurable Gaussian statistics of the smoother posterior distribution
    \begin{equation*}
        \pp\big(\vyt\big|\vx(s), s\leq T\big)\overset{\d}{\sim}\mathcal{N}_l\big(\vm{\ns}(t),\mr{\ns}(t)\big),
    \end{equation*}
    defined as
    \begin{equation*}
        \vmt{\ns}:=\mathbb{E}\big[\vyt\big|\vx(s),s\leq T\big] \quad\text{and}\quad \mrt{\ns}:=\mathbb{E}\big[(\vyt-\vmt{\ns})(\vyt-\vmt{\ns})^\dagger\big|\vx(s),s\leq T\big],
    \end{equation*}
    are the unique continuous solutions to the system of optimal nonlinear smoother backward differential equations:
    \begin{align}
        \smooth{\d \vm{\ns}}&=-\left(\ml^\vy\vm{\ns}+\vf^\vy-\mb\mr{\nf}^{-1}(\vm{\nf}-\vm{\ns})\right)\d t+(\ms^\vy\circ \ms^\vx)(\ms^\vx\circ \ms^\vx)^{-1}\big( \smooth{\d\vx}+(\ml^\vx\vm{\ns}+\vf^\vx)\d t \big), \label{eq:revbackinter1} \\
         \smooth{\d \mr{\ns}}&= -\big((\ma+\mb\mr{\nf}^{-1})\mr{\ns}+\mr{\ns}(\ma+\mb\mr{\nf}^{-1})^\dagger-\mb\big)\d t. \label{eq:revbackinter2}
    \end{align}
    The backward-arrow notation in \eqref{eq:revbackinter1}--\eqref{eq:revbackinter2} is to be understood as:
    \begin{equation*}
        \smooth{\d \vm{\ns}}:=\lim_{\Delta t\to 0} \big(\vm{\ns}(t)-\vm{\ns}(t+\Delta t)\big), \quad\smooth{\d \mr{\ns}}:=\lim_{\Delta t\to 0} \big(\mr{\ns}(t)-\mr{\ns}(t+\Delta t)\big), \quad \smooth{\d \vx}:=\lim_{\Delta t\to 0} \big(\vx(t)-\vx(t+\Delta t)\big).
    \end{equation*}
    In other words, the notation $\smooth{\frac{\d\cdot}{\d t}}$ corresponds to the negative of the usual derivative, which means the system in \eqref{eq:revbackinter1}--\eqref{eq:revbackinter2} is to be solved backward over $[0,T]$. The ``starting" values for the smoother posterior statistics, $\big(\vm{\ns}(T),\mr{\ns}(T)\big)$, are the same as those of the corresponding filter estimates at the endpoint $t=T$, $\big(\vm{\nf}(T),\mr{\nf}(T)\big)$.
\end{thm}
\begin{proof}[\textbf{\underline{Proof}}] This is Theorem 12.10 in Liptser \& Shiryaev \cite{liptser2001statisticsII}. For the analogous result in the case of discrete time, see Theorem 13.12, which outlines the corresponding optimal recursive nonlinear smoother backward difference equations, with the respective sufficient assumptions given in Subchapters 13.2.1 and 13.3.8. For a martingale-free proof see Theorem 3 in Andreou \& Chen \cite{andreou2024martingale} (although it additionally requires assumptions \textbf{\blue{(10)}} and \textbf{\blue{(13)}} found in Appendix \ref{sec:app1}, as to pass from the discrete-time smoother to the continuous one via a formal limit).
\end{proof}

\section{Optimal Online Smoother} \label{sec:3}

As aforementioned, the standard smoothing procedure involves a forward pass using a filtering method followed by a backward pass to obtain the optimal smoother state estimation \cite{law2015data, rauch1965maximum, chen2020efficient, sarkka2023bayesian}. Recalculating the smoother statistics from scratch whenever new observations become available is computationally and storage-intensive. Therefore, an online version of the optimal nonlinear smoother, requiring only forward-in-time updates, is highly desirable, as it can update estimated states sequentially with new data. Fortunately, the CGNS framework provides exact, closed-form solutions for deriving an online smoother. These facilitate efficiency and the understanding of the mathematical and numerical properties of the CGNS.

\subsection{Optimal Online Forward-In-Time Discrete Smoother} \label{sec:3.1}

We begin by adopting a time discretization scheme, determined by the rate at which observational data arrive, similar to numerical integration. Although our framework allows for irregular observation intervals, for simplicity, we assume the data arrive at a regular, uniform rate, which is common in practice. Thus, observations are available sequentially at times $t_0=0<t_1<t_2<\cdots<t_n<\cdots<+\infty$, where $\dt_j\equiv\dt=t_{j+1}-t_j,$ for all $j$. We also assume that $\dt$ is sufficiently small to ensure stability and consistency of the discrete-time schemes for numerical integration of the CGNS equations, as well as of the optimal nonlinear filter and smoother equations, thus ensuring their convergence \cite{andreou2024martingale}.

In what follows, the superscript notation $\cdot^{\, j}$ is used to denote the discrete approximation to the continuous form of the respective vector or matrix functional when evaluated on $t_j$, for example $\ml^{\vx,j}:=\ml^\vx(t_j,\vx(t_j))$. We write the CGNS of equations \eqref{eq:condgauss1}--\eqref{eq:condgauss2} in a discrete fashion using the Euler--Maruyama scheme \cite{kloeden1992numerical, gardiner2009stochastic}. Then the time discretization of \eqref{eq:condgauss1}--\eqref{eq:condgauss2} is simply given by
\begin{align}
    \vx^{j+1} &= \vx^{j}+\left(\ml^{\vx,j}\vy^j+\vf^{\vx,j}\right)\dt+\ms_1^{\vx,j}\sqrt{\dt} \ve_1^j+\ms_2^{\vx,j}\sqrt{\dt} \ve_2^j, \label{eq:discretecondgauss1}\\
    \vy^{j+1} &= \vy^{j}+\left(\ml^{\vy,j}\vy^j+\vf^{\vy,j}\right)\dt+\ms_1^{\vy,j}\sqrt{\dt} \ve_1^j+\ms_2^{\vy,j}\sqrt{\dt} \ve_2^j, \label{eq:discretecondgauss2}
\end{align}
where $\ve_1^j$ and $\ve_2^j$ are mutually independent complex standard Gaussian random noises. The explicit nature and order of the temporal discretization used here is sufficient for the CGNS framework we are working with, meaning implicit or higher-order discretization schemes are not needed \cite{andreou2024martingale}.

Given a series of realizations for the observable $\vx$, $\big\{\vx^0,\vx^1,\ldots,\vx^n\big\}$, where $\vx^j$ was obtained $\dt$ time units after $\vx^{j-1}$ for $j=1,\ldots,n$, we let $\vm{s}^{j,n}$ and $\mr{s}^{j,n}$ denote the discrete smoother posterior mean and the discrete smoother posterior covariance, respectively, when evaluated at time $t=t_j$ and conditioned on this realization of $\vx$ up to time $t=t_n$, where $0\leq j\leq n$ (see Theorem \ref{thm:smoothing}). In other words, we define
\begin{gather*}
    \vm{\normalfont{s}}^{j,n}:=\mathbb{E}\big[\vy^j\big|\vx^s,s\leq n\big],\quad \mr{\normalfont{s}}^{j,n}:=\mathrm{Cov}\big(\vy^j,\vy^j\big|\vx^s,s\leq n\big)=\mathbb{E}\big[(\vy^j-\vm{\normalfont{s}}^{j,n})(\vy^j-\vm{\normalfont{s}}^{j,n})^\dagger\big|\vx^s,s\leq n\big].
\end{gather*}
Notice how we explicitly note the dependence of the smoother state estimates on the length of the observational period. As is known by the discrete-time counterpart of Theorem \ref{thm:condgaussianity} (see Theorem 13.3 in Liptser \& Shiryaev \cite{liptser2001statisticsII}), the smoother posterior distribution $\pp\big(\vy^j\big|\vx^s, s=0,\ldots,n\big)$ is ($\pp$-almost surely) conditionally Gaussian, $\mathcal{N}_l\big(\vm{\normalfont{s}}^{j,n},\mr{\normalfont{s}}^{j,n}\big)$, and it is possible to show that the conditional mean $\vm{\normalfont{s}}^{j,n}$ and conditional covariance $\mr{\normalfont{s}}^{j,n}$ of the discrete smoother at time step $t_j$ when conditioning up to the $n$-th observation satisfy the following recursive backward difference equations for $n\in\mathbb{N}$ and $j=0,1,\ldots,n-1$ (under the regularity conditions \textbf{\blue{(1)}}--\textbf{\blue{(13)}} in Appendix \ref{sec:app1} and all other assumptions thus far) \cite{andreou2024martingale}:
\begin{align}
    \vm{\normalfont{s}}^{j,n}&=\vm{\nf}^j+\mathbf{E}^j\big(\vm{\normalfont{s}}^{j+1,n}-(\mathbf{I}_{l\times l}+\ml^{\vy,j}\dt)\vm{\nf}^j-\vf^{\vy,j}\dt\big)+\mathbf{F}^j\big(\vx^{j+1}-\vx^{j}-(\ml^{\vx,j}\vm{\nf}^j+\vf^{\vx,j})\dt\big), \label{eq:discretesmoother1} \\
    \mr{\normalfont{s}}^{j,n}&=\mr{f}^j+\mathbf{E}^j\big(\mr{\normalfont{s}}^{j+1,n}(\mathbf{E}^j)^\dagger-(\mathbf{I}_{l\times l}+\ml^{\vy,j}\dt)\mr{f}^j\big)-\mathbf{F}^j\ml^{\vx,j}\mr{f}^j\dt, \label{eq:discretesmoother2}
\end{align}
where $\mathbf{I}_{l\times l}$ is the $l\times l$ identity matrix and the auxiliary matrices appearing in \eqref{eq:discretesmoother1}--\eqref{eq:discretesmoother2} are given up to leading-order $O(\dt)$, since we are using the Euler--Maruyama scheme to temporally discretize the CGNS of equations, as
\begin{align}
    \mathbf{E}^j&:=\mathbf{I}_{l\times l}+\big((\ms^\vy\circ\ms^\vx)^j\big((\ms^\vx\circ\ms^\vx)^{j}\big)^{-1}\mathbf{G}^{\vx,j}-\mathbf{G}^{\vy,j}\big)\dt+O(\dt^2)\in\mathbb{C}^{l\times l}, \label{eq:discretesmootherauxmat1}\\
    \begin{split}
        \mathbf{F}^j&:=-\mr{\nf}^j\Big((\mathbf{K}^j)^{\dagger}+\big((\mathbf{G}^{\vx,j})^\dagger\mathbf{K}^j\mr{\nf}^j(\mathbf{K}^j)^{\dagger}-(\mr{\nf}^j)^{-1}(\mathbf{H}^{j})^\dagger\mr{\nf}^j(\mathbf{K}^j)^\dagger+(\ml^{\vy,j})^{\dagger}(\mathbf{K}^{j})^\dagger\big)\dt\\
        &\hspace{1.5cm} -(\ml^{\vx,j})^\dagger\big(\big((\ms^\vx\circ\ms^\vx)^{j}\big)^{-1}+\mathbf{K}^j\mr{\nf}^j(\mathbf{K}^j)^\dagger\dt\big)\Big)+O(\dt^2)\in\mathbb{C}^{l\times k},
    \end{split} \label{eq:discretesmootherauxmat2}
\end{align}
where
\begin{equation}
\begin{gathered}
    \mathbf{G}^{\vx,j}:=\ml^{\vx,j}+(\ms^\vx\circ\ms^\vy)^j(\mr{\nf}^j)^{-1}\in\mathbb{C}^{k\times l}, \quad \mathbf{G}^{\vy,j}:=\ml^{\vy,j}+(\ms^\vy\circ\ms^\vy)^j(\mr{\nf}^j)^{-1}\in\mathbb{C}^{l\times l}, \\
    \mathbf{H}^j:=(\mr{\nf})^{-1}\big(\ml^{\vy,j}\mr{f}^j+\mr{f}^j(\ml^{\vy,j})^{\dagger}+(\ms^\vy\circ\ms^\vy)^j\big)\in\mathbb{C}^{l\times l}, \quad
    \mathbf{K}^j:=\big((\ms^\vx\circ\ms^\vx)^{j}\big)^{-1}\mathbf{G}^{\vx,j}\in\mathbb{C}^{k\times l}.
\end{gathered} \label{eq:auxiliarymatrices}
\end{equation}
We also note here that in the absence of noise cross-interaction, i.e., $\ms^\vx\circ\ms^\vy\equiv \mathbf{0}_{k\times l}$, \eqref{eq:discretesmootherauxmat1} and \eqref{eq:discretesmootherauxmat2} simplify significantly, where up to leading-order $O(\dt)$ their expressions reduce down to \cite{andreou2024martingale}:
\begin{align}
    \mathbf{E}^j:=\mathbf{I}_{l\times l}-\mathbf{G}^{\vy,j}\dt+O(\dt^2),\quad 
    \mathbf{F}^j:= \mathbf{G}^{\vy,j}\mr{f}^j(\ml^{\vx,j})^{\dagger}\big((\ms^\vx\circ\ms^\vx)^{j}\big)^{-1}\dt+O(\dt^2). \label{eq:simplediscretesmootherauxmat}
\end{align}

Using \eqref{eq:discretesmoother1}--\eqref{eq:discretesmoother2}, the following theorem outlines the procedure for obtaining an optimal online forward-in-time discrete smoother state estimate for the posterior Gaussian statistics. The proof of this result is provided in Appendix \ref{sec:app2}.

\begin{thm}[\textbf{Optimal Online Forward-In-Time Discrete Smoother}] \label{thm:onlinesmoother}
Let $\big(\vxt,\vyt\big)$ satisfy \eqref{eq:condgauss1}--\eqref{eq:condgauss2} and assume the validity of the assumptions in Theorem \ref{thm:smoothing} and all regularity conditions outlined in Appendix \ref{sec:app1}, \textbf{\blue{(1)}}--\textbf{\blue{(13)}}. Suppose now the observational data for the observed process, $\vx^0,\vx^1,\vx^2,\ldots$, are given sequentially. When a new observation, denoted by $\vx^n$ for $n\in\mathbb{N}$, becomes available, it is utilized to update all the existing optimal state estimates at time instants $t_j$ for $0\leq j\leq n-1$ and it then provides a new state estimate at $j=n$. The discrete smoother posterior distribution $\pp\big(\vy^j\big|\vx^s, 0\leq s\leq n\big)$ is conditionally Gaussian,
\begin{equation} \label{eq:onlinegaussianity}
    \pp\big(\vy^j\big|\vx^s, 0\leq s\leq n\big) \overset{\d}{\sim}\mathcal{N}_l\big(\vm{s}^{j,n},\mr{s}^{j,n}\big),
\end{equation}
and the conditional mean $\vm{s}^{j,n}$ and conditional covariance $\mr{s}^{j,n}$ for $n\in\mathbb{N}$ and $0\leq j \leq n-1$ satisfy the following recursive backward difference equations:
\begin{align}
    \vm{s}^{j,n}&=\vm{s}^{j,n-1}+\mathbf{D}^{j,n-2}\big(\vm{s}^{n-1,n}-\vm{f}^{n-1}\big), \label{eq:onlinerecursive1}\\
    \mr{s}^{j,n}&=\mr{s}^{j,n-1}+\mathbf{D}^{j,n-2}\big(\mr{s}^{n-1,n}-\mr{f}^{n-1}\big)(\mathbf{D}^{j,n-2})^\dagger,\label{eq:onlinerecursive2}
\end{align}
where the update matrix $\mathbf{D}^{j,n-1}$, or $\mathbf{D}^{j,n-2}$ after the trivial reindexing $n-1\leadsto n-2$ in the following (without loss of generality), is defined in a forward-in-time fashion as
\begin{alignat}{3}
    &\mathbf{D}^{n,n-1}&&:=\mathbf{I}_{l\times l}, \hspace{1cm} && \text{ for } n\in\mathbb{N}, \nonumber \\
    &\mathbf{D}^{n-1,n-1}&&:=\mathbf{E}^{n-1}, \hspace{1cm} && \text{ for } n\in\mathbb{N}, \label{eq:updatematrix1}\\
    &\mathbf{D}^{j,n-1}&&:=\mathbf{D}^{j,n-2}\mathbf{E}^{n-1}, \hspace{1cm} && \text{ for } n\in\mathbb{N}_{\geq 2}, \text{ and } j=0,1,\ldots,n-2, \nonumber
\end{alignat}
where $\mathbf{E}^j$ is given up to leading-order in \eqref{eq:discretesmootherauxmat1}. For $n\in\mathbb{N}$ we have
\begin{align}
    \vm{s}^{n-1,n}&=\mathbf{E}^{n-1}\vm{f}^{n}+\mathbf{b}^{n-1}, \label{eq:onlineauxiliary1}\\
    \mr{s}^{n-1,n}&=\mathbf{E}^{n-1}\mr{f}^{n}(\mathbf{E}^{n-1})^\dagger+\mathbf{P}^{n-1}_n, \label{eq:onlineauxiliary2}
\end{align}
with the $\mathbf{b}^{n-1}$ and $\mathbf{P}^{n-1}_n$ auxiliary residual terms being defined by
\begin{align}
    \begin{split}
    \mathbf{b}^{n-1}&:=\vm{f}^{n-1}-\mathbf{E}^{n-1}\big((\mathbf{I}_{l\times l}+\ml^{\vy,n-1}\dt)\vm{\nf}^{n-1}+\vf^{\vy,n-1}\dt\big)\\
    &\hspace{1.5cm}+\mathbf{F}^{n-1}\big(\vx^{n}-\vx^{n-1}-(\ml^{\vx,n-1}\vm{\nf}^{n-1}+\vf^{\vx,n-1})\dt\big),
    \end{split} \label{eq:onlineauxiliary3}\\
    \mathbf{P}^{n-1}_{n}&:= \mr{f}^{n-1}-\mathbf{E}^{n-1}(\mathbf{I}_{l\times l}+\ml^{\vy,n-1}\dt)\mr{f}^{n-1}-\mathbf{F}^{n-1}\ml^{\vx,n-1}\mr{f}^{n-1}\dt, \label{eq:onlineauxiliary4}
\end{align}
for $\mathbf{E}^n$ and $\mathbf{F}^n$ given up to leading-order in \eqref{eq:discretesmootherauxmat1} and \eqref{eq:discretesmootherauxmat2}, respectively. For $j=n$, we have by definition that $\vm{s}^{n,n}=\vm{f}^n$ and $\mr{s}^{n,n}=\mr{f}^n$, since the smoother and filter posterior Gaussian statistics coincide at the end point per Theorem \ref{thm:smoothing}.
\end{thm}

Based on this theorem, the algorithm associated with the online forward-in-time discrete smoother update equations \eqref{eq:onlinerecursive1}--\eqref{eq:onlinerecursive2} is outlined in Algorithm \ref{algo:onlinesmoother}. There, we use a compact expression for the update matrix $\mathbf{D}^{j,n-2}$,
\begin{equation} \label{eq:updatematrixcompact}
    \mathbf{D}^{j,n-2}=\overset{\mathlarger{\curvearrowright}}{\prod^{n-2}_{i=j}} \mathbf{E}^i=\mathbf{E}^j\mathbf{E}^{j+1}\cdots \mathbf{E}^{n-2},
\end{equation}
which is equivalent to \eqref{eq:updatematrix1}, with the details of this equivalence being shown in Appendix \ref{sec:app2}. (The curved arrow pointing to the right above the product symbol in \eqref{eq:updatematrixcompact} indicates the order or direction with which we expand said product.)

\begin{algorithm}
\caption{\textbf{Optimal Online Forward-In-Time Discrete Smoother}} \label{algo:onlinesmoother}
\KwData{$\vx^0$,  $\vm{f}^0=\vm{f}(0)$, $\mr{f}^0=\mr{f}(0)$, $\dt$}
\KwResult{Discrete Smoother Gaussian Statistics $\left\{\vm{s}^{j,n}\right\}_{0\leq j\leq n}$ and $\left\{\mr{s}^{j,n}\right\}_{0\leq j\leq n}$ for $n\in\mathbb{N}$}
\For{$n\in\mathbb{N}$}{
    Receive new observation $\vx^n$ after time $\dt$\;
    Compute $\vm{s}^{n,n}=\vm{f}^n$ through \eqref{eq:filter1}: \\ \hspace{0.2cm}$\vm{f}^n \gets \vm{f}^{n-1}+(\ml^{\vy,n-1}\vm{f}^{n-1}+\vf^{\vy,n-1})\dt+(\mr{f}^{n-1}(\ml^{\vx,n-1})^{\dagger}+(\ms^\vy\circ\ms^\vx)^{n-1})$ \\
    \hspace{1.2cm}$\times\big((\ms^\vx\circ \ms^\vx)^{n-1}\big)^{-1}\big(\vx^{n}-\vx^{n-1}-(\ml^{\vx,n-1}\vm{f}^{n-1}+\vf^{\vx,n-1})\dt\big)$\;
    Compute $\mr{s}^{n,n}=\mr{f}^n$ through \eqref{eq:filter2}: \\ \hspace{0.2cm}$\mr{f}^n \gets \mr{f}^{n-1}+\big(\ml^{\vy,n-1}\mr{f}^{n-1}+\mr{f}^{n-1}(\ml^{\vy,n-1})^{\dagger}+(\ms^\vy\circ\ms^\vy)^{n-1}$ \\
    \hspace{1.3cm}$-(\mr{f}^{n-1}(\ml^{\vx,n-1})^{\dagger}+(\ms^\vy\circ\ms^\vx)^{n-1})\big((\ms^\vx\circ\ms^\vx)^{n-1}\big)^{-1}(\ml^{\vx,n-1}\mr{f}^{n-1}+(\ms^\vx\circ\ms^\vy)^{n-1})\big)\dt$\;
    Compute $\mathbf{E}^{n-1}$, $\mathbf{b}^{n-1}$, and $\mathbf{P}^{n-1}_n$ through \eqref{eq:discretesmootherauxmat1}, \eqref{eq:onlineauxiliary3}, and \eqref{eq:onlineauxiliary4}, respectively: \\
    \hspace{0.2cm}$\mathbf{E}^{n-1}\gets \mathbf{I}_{l\times l}+\big((\ms^\vy\circ\ms^\vx)^{n-1}\big((\ms^\vx\circ\ms^\vx)^{n-1}\big)^{-1}\mathbf{G}^{\vx,n-1}-\mathbf{G}^{\vy,n-1}\big)\dt$\;
    \hspace{0.2cm}$\mathbf{b}^{n-1}\gets \vm{f}^{n-1}-\mathbf{E}^{n-1}\big((\mathbf{I}_{l\times l}+\ml^{\vy,n-1}\dt)\vm{\nf}^{n-1}+\vf^{\vy,n-1}\dt\big)$ \\ \hspace{1.6cm}$+\mathbf{F}^{n-1}\big(\vx^{n}-\vx^{n-1}-(\ml^{\vx,n-1}\vm{\nf}^{n-1}+\vf^{\vx,n-1})\dt\big)$\;
    \hspace{0.2cm}$\mathbf{P}_n^{n-1}\gets \mr{f}^{n-1}-\mathbf{E}^{n-1}(\mathbf{I}_{l\times l}+\ml^{\vy,n-1}\dt)\mr{f}^{n-1}-\mathbf{F}^{n-1}\ml^{\vx,n-1}\mr{f}^{n-1}\dt$\;

    Compute $\vm{s}^{n-1,n}$ and $\mr{s}^{n-1,n}$ through \eqref{eq:onlineauxiliary1} and \eqref{eq:onlineauxiliary2}, respectively: \\
    \hspace{0.2cm}$ \vm{s}^{n-1,n}=\mathbf{E}^{n-1}\vm{f}^{n}+\mathbf{b}^{n-1}$\;
    \hspace{0.2cm}$ \mr{s}^{n-1,n}=\mathbf{E}^{n-1}\mr{f}^{n}(\mathbf{E}^{n-1})^\dagger+\mathbf{P}^{n-1}_n$\;

    \For{$j=0:n-1$}{
        $\displaystyle\mathbf{D}^{j,n-2}\gets\overset{\mathlarger{\curvearrowright}}{\prod^{n-2}_{i=j}} \mathbf{E}^i$\;
        $\vm{s}^{j,n} \gets \vm{s}^{j,n-1}+\mathbf{D}^{j,n-2}\big(\vm{s}^{n-1,n}-\vm{f}^{n-1}\big)$\;
        $\mr{s}^{j,n}=\mr{s}^{j,n-1}+\mathbf{D}^{j,n-2}\big(\mr{s}^{n-1,n}-\mr{f}^{n-1}\big)(\mathbf{D}^{j,n-2})^\dagger$\;
}
}
\end{algorithm}

\subsection{Understanding How the Online Forward-In-Time Discrete Smoother Works} \label{sec:3.2}

Figure \ref{fig:onlinesmoother} presents a schematic diagram that intuitively explains how the online smoother operates using forward filtering and backward smoothing estimates. Although the figure focuses on the online discrete smoother mean for brevity, the same logic (as well as subsequent analysis) applies to the smoother covariance. In the diagram, the first superscript index corresponds to rows, while the second corresponds to columns. The discussion centers on the last red dashed box (last column), representing the discrete smoother estimates at the current observation (the $n$-th observation) over all time instants $j=0,\ldots,n$. Each column (fixed second index) illustrates the backward smoothing process for a given batch realization $\{\vx^0,\ldots,\vx^n\}$, involving a forward pass to compute the filter estimate (Theorem \ref{thm:filtering}) and a backward pass to calculate the smoother estimate (Theorem \ref{thm:smoothing}), indicated by the red arrows. Each row (fixed first index) shows how the discrete online smoother estimates at a fixed time $t_j$ are updated as new observations arrive sequentially, using the online forward-in-time smoother algorithm. Information flows into $\vm{s}^{j,n}$ from the filter estimates $\vm{f}^{j}$ and $\vm{f}^{n}$. The former, $\vm{f}^{j}$, influences it through forward online-smoothing in its row, transitioning towards $\vm{s}^{j,n-1}$ and to the next columns using the online forward-in-time smoother \eqref{eq:onlinerecursive1}. (Through $\vm{s}^{j,n-1}$, the information in $\vm{f}^{n-1}$ is also captured via the classical backward smoother in its column.) The latter, $\vm{f}^{n}$, affects it via the backward smoother \eqref{eq:discretesmoother1} in its column, implicitly through $\vm{s}^{n-1,n}$ in \eqref{eq:onlinerecursive1}. These deductions, stemming from Figure \ref{fig:onlinesmoother}, are explicitly and mathematically expressed by the fact that the amount of updated information incurred at time step $t_j$, due to the new observation $\vx^n$, is proportional to the update at $t_{n-1}$ via the matrix $\mathbf{D}^{j,n-2}$, since
\begin{equation*}
    \vm{s}^{j,n}-\vm{s}^{j,n-1}=\mathbf{D}^{j,n-2}\big(\vm{s}^{n-1,n}-\vm{f}^{n-1}\big) \quad\text{and}\quad \vm{f}^{n-1}=\vm{s}^{n-1,n-1}.
\end{equation*}
Essentially, $\vm{s}^{n-1,n}-\vm{f}^{n-1}=\vm{s}^{n-1,n}-\vm{s}^{n-1,n-1}$ amounts to the innovation in the mean stemming from the new observational data, which is then being weighted by the optimal gain tensor $\mathbf{D}^{j,n-2}$.

\begin{figure}[ht]
\begin{center}
\includegraphics[width=0.7\textwidth]{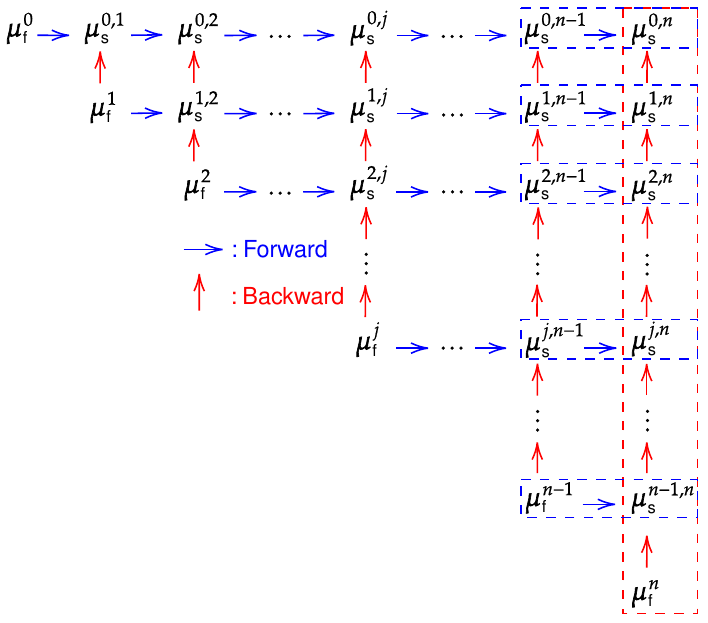}
\end{center}
\caption{Schematic diagram of how the online discrete smoother update works. Only the smoother mean $\vm{s}^{\cdot,\cdot}$ is depicted; the same diagram applies to the smoother covariance $\mr{s}^{\cdot,\cdot}$ as well, without loss of generality.}
\label{fig:onlinesmoother}
\end{figure}

\subsubsection{How the Online Forward-In-Time Discrete Smoother's Mechanisms Guide the Development of an Adaptive-Lag Algorithm} \label{sec:3.2.1}

It becomes apparent from \eqref{eq:onlinerecursive1}--\eqref{eq:onlinerecursive2} that to understand how the online smoother update works, as well as to figure out how information dissipates through it, we need to study the mathematical properties of the update tensor given in \eqref{eq:updatematrix1}, or compactly in \eqref{eq:updatematrixcompact}, which in turn implies the thorough study of the auxiliary operator
\begin{equation*}
    \mathbf{E}^{j}=\mathbf{I}_{l\times l}-\big(\mathbf{G}^{\vy,j}-(\ms^\vy\circ\ms^\vx)^{j}\big((\ms^\vx\circ\ms^\vx)^{j}\big)^{-1}\mathbf{G}^{\vx,j}\big)\dt+O(\dt^2),
\end{equation*}
for $j=0,1,\ldots,n-2$. Such an analysis will also assist in the formulation of an adaptive lag variant for the online discrete smoother, that implicitly optimizes the use of data for the reduction of the computational overhead. 

Based on the definitions of \eqref{eq:auxiliarymata}--\eqref{eq:auxiliarymatb}, it is easy to see that
\begin{equation*}
    \mathbf{E}^{j}=\mathbf{I}_{l\times l}-\big(\ma^j+\mb^j(\mr{f}^j)^{-1}\big)\dt+O(\dt^2)=\mathbf{I}_{l\times l}-\big(\ma^j\mr{f}^j+\mb^j\big)(\mr{f}^j)^{-1}\dt+O(\dt^2).
\end{equation*}
As such, it is immediate by this form of $\mathbf{E}^j$, that if $\dt$ is sufficiently small and
\begin{equation} \label{eq:ejspectralradiuscondition}
    \mathbf{0}_{l\times l} \prec  \mathbf{G}^{\vy,j}-(\ms^\vy\circ\ms^\vx)^j\big((\ms^\vx\circ\ms^\vx)^{j}\big)^{-1}\mathbf{G}^{\vx,j}=\ma^j+\mb^j(\mr{f}^j)^{-1},
\end{equation}
where ``$\prec$" is to be understood in the Loewner partial ordering sense over the convex cone of nonnegative-definite matrices, then by definition the spectral radius of $\mathbf{E}^j$ is less than $1$ uniformly over $j$:
\begin{equation} \label{eq:ejspectralradius}
    \varrho(\mathbf{E}^j)<1, \ \forall j=0,1,\ldots,n-1.
\end{equation}
Notably, \eqref{eq:ejspectralradiuscondition} is a necessary condition to establish mean-square stability of the smoother Gaussian statistics \cite{khasminskii2012stochastic, neckel2013random, lu2021mathematical, arnold2014random}, as can be clearly seen by the damping in the backward random differential equation that the smoother covariance satisfies in \eqref{eq:revbackinter2} of Theorem \ref{thm:smoothing}, with its discrete counterpart given in \eqref{eq:discretesmoother2}, and where the same damping coefficient also appears in the smoother mean's linear evolution equation, since after a few algebraic manipulations we have
\begin{equation*}
    \smooth{\d \vm{\ns}}=\big(-(\ma+\mb\mr{\nf}^{-1})\vm{\ns}+\vf^\vy-\mb\mr{\nf}^{-1}\vm{\nf}\big)\d t+(\ms^\vy\circ \ms^\vx)(\ms^\vx\circ \ms^\vx)^{-1}\big( \smooth{\d\vx}+\vf^\vx\d t \big),
\end{equation*}
with its discrete counterpart given in \eqref{eq:discretesmoother1}. Probing into \eqref{eq:ejspectralradius}, note that the filter covariance matrix, $\mr{f}^j$, reflects uncertainties in both the unobserved and observed dynamics, as it is conditioned on the current observable $\sigma$-algebra. If the observed variables have much less uncertainty than the unobserved process, the eigenvalues of $(\ms^\vy\circ\ms^\vy)^j$ will be much larger than those of $\mr{f}^j$, as suggested by discretizing the random Riccati equation in \eqref{eq:filter2}, since the quadratic term is weighted by $\mc:=(\ml^\vx)^\dagger(\ms^\vx\circ\ms^\vx)^{-1}\ml^\vx$ and then subtracted from the nonnegative-definite stochastic forcing $\ms^\vy\circ\ms^\vy$ (in the update of the forecast uncertainty). Consequently, by this reasoning and \eqref{eq:discretesmootherauxmat1}, the spectrum of $\mathbf{E}^j$ should concentrate near the origin and within the unit disk in the complex plane, thus satisfying \eqref{eq:ejspectralradius}. Recalling that the control on the spectrum of $\mathbf{E}^j$ is determined by the exponential mean-square stability of the smoother posterior statistics, in this scenario observations become highly informative; they provide sufficient information to update the estimated state over a short time interval, with an exponentially-fast decaying impact at distant past time instants. 

This is further reflected in the spectral radius of the update matrix, $\mathbf{D}^{j,n-2}$, which explicitly determines the influence of the observations on the discrete smoother updates during the recursive backward difference equations in \eqref{eq:onlinerecursive1}--\eqref{eq:onlinerecursive2}. Numerical examples, such as the nonlinear dyad-interaction model in Section \ref{sec:4.1}, show that even with intermittent instabilities and noise cross-interactions it exhibits an overall exponential decay backwards in time (i.e., in $j$ for each $n$), with an accelerating rate under the observance of extreme events in the observed time series. Therefore, any potent adaptive-lag strategy accompanying the forward-in-time online smoother in \eqref{eq:onlinerecursive1}--\eqref{eq:onlinerecursive2} should: 
\begin{itemize}
    \item[(a)] Yield larger lag values during the generation of extreme events, where observations significantly inform and contribute to the smoother update, thus quickly resolving the intermittent instabilities. During this observational period (i.e., for these $n$), the spectral radius of $\mathbf{D}^{j,n-2}$ remains significant (close to $1$) for time instants (i.e., for $j$) which cover the initiation period of the extreme event.

    \item[(b)] Produce smaller lags at periods of large signal-to-noise ratios in the observed variables, where the prior state estimates from the preceding observations are sufficiently skillful (e.g., during the demise of extreme events). During this observational period (i.e., for these $n$), the spectral radius of $\mathbf{D}^{j,n-2}$ showcases remarkable rapid exponential decay at time instants (i.e., for $j$) prior to and during the extreme event.
\end{itemize}
See Panels (g)--(i) of Figure \ref{fig:Dyad_Interaction_Fig_2} and Panels (a)--(c) of Figure \ref{fig:LDA_Ice_Floes_Fig_3} for numerical examples.

In general, except in trivial cases where the unobserved state space is one-dimensional (e.g., see Section \ref{sec:3.2.2}), condition \eqref{eq:ejspectralradius} on the constituents of the product defining the update matrix $\mathbf{D}^{j,n-2}$ in \eqref{eq:updatematrixcompact} does not guarantee that the update tensor $\mathbf{D}^{j,n-2}$ itself will also have a subunit spectral radius for $j$ sufficiently far back from $n-1$. (Such a condition is necessary for establishing rigorous mathematical results of convergence for \eqref{eq:onlinerecursive1}--\eqref{eq:onlinerecursive2}, and reflects the decaying impact region of $\vx^n$ on $\vy^j$'s online smoother state estimation.) This is because, generally, the spectral radius of a product is not dominated by the product of the associated spectral radii unless rather restrictive conditions are assumed, with sufficient conditions including simultaneous triangularization of the factors over $\mathbb{C}$ \cite{radjavi2000simultaneous} (e.g., the factors commute, as a consequence of Gelfand's theorem) or the factors being radial matrices \cite{goldberg1974matrices}. Consequently, predicting how each observation affects state estimation through the online smoother updates can be complex and highly dependent on the specific structure of the CGNS. This variability underscores the need for a simple numerical procedure to assess the impact region of each new observation, as described by the recursive backward difference equations for the discrete smoother posterior Gaussian statistics in \eqref{eq:onlinerecursive1}--\eqref{eq:onlinerecursive2}, which in turn aids in constructing an effective adaptive lag strategy for the online smoother. Nevertheless, even in the ideal setting where control on the spectrum of $\mathbf{D}^{j,n-2}$ can be established, an online smoother which defines its adaptive lag based on the eigenvalues of the update matrix incurs computational costs in the order of $O(l^3)$, which for high-dimensional systems is inexcusable. This is why, in Section \ref{sec:3.4}, we develop an information theory-based adaptive-lag approach to the online discrete smoother which efficiently approximates the observational impact region, while implicitly and effectively capturing the essence of the exact method defined by the spectral properties of the update matrix $\mathbf{D}^{j,n-2}$.

As the system dimension and $n$ increase, the computational overhead for updating the online smoother becomes significant, particularly due to storing the auxiliary matrices needed for calculating the update matrix $\mathbf{D}^{j,n-2}$, specifically the $\mathbf{E}^j$ matrices, as well as the filter and online smoother covariance matrices. Additionally, constant recalculation of the covariance tensor can lead to $O(l^3)$ computations per update, especially when $k\approx l$. Given that the impact of new observations on past states decays exponentially (linked to the aforementioned spectral properties of $\mathbf{D}^{j,n-2}$ and $\mathbf{E}^j$), it is advantageous to develop an adaptive-lag (or fixed-lag) online smoother \cite{kitagawa1996monte, cappe2010inference, olsson2008sequential}, where the impact region of or lag at each observation is adaptively computed (or finite and predetermined) based on the dynamical and statistical properties of the CGNS.

\subsubsection{Behavior of the Online Forward-In-Time Discrete Smoother for a Two-Dimensional Linear System with Additive Noise} \label{sec:3.2.2}

To make things more concrete, we show that in the case of a simple two-dimensional linear Gaussian model, the value of $\mathbf{E}^j\equiv E^j$ (which in this case is just a scalar), is necessarily in $(-1,1)$, thus concretely establishing \eqref{eq:ejspectralradius} and by extension, due to the one-dimensional observable and unobserved state spaces, leads to $\mathbf{D}^{j,n-2}\equiv D^{j,n-2}\in(-1,1)$ for this specific example. Consider the following two-dimensional linear Gaussian model with additive noise:
\begin{align*}
    \d x&=\big(\lambda_xy+f_x(t,x)\big)+\sigma_x\d W_1,\\
    \d y&=\big(\lambda_yy+f_y(t,x)\big)+\sigma_y\d W_2,
\end{align*}
where $\lambda_x\in\rr$, $\lambda_y<0$, $\sigma_x>0$, and $\sigma_y>0$ are constants and $f_x(t,x)$ is such that the $x$-dynamics are stable. As before, $x$ and $y$ are the observed and unobserved variables, respectively. Note that we assume the absence of noise cross-interaction, without loss of generality. Here, due to the linearity of the system, $\lambda_y$ needs to be negative in order to guarantee the existence of the statistical equilibrium state of the coupled system when $y$ is assumed to be the unobservable (this also becomes apparent from \eqref{eq:equilibriumcoveqn} later on) \cite{gardiner2009stochastic, khasminskii2012stochastic, liu2019stochastic}. Since $\lambda_x$, $\lambda_y$ and the noise feedbacks are constant, the equilibrium solution for the filter covariance, $R_{\text{f}}^{\text{eq}}$, can thus be solved via the following steady state Riccati equation (see \eqref{eq:filter2}) \cite{chen2014information, chen2015noisy, gardiner2009stochastic}:
\begin{equation} \label{eq:equilibriumcoveqn}
    2\lambda_yR_{\text{f}}^{\text{eq}}+\sigma_y^2=\frac{(R_{\text{f}}^{\text{eq}}\lambda_x)^2}{\sigma_x^2},
\end{equation}
which after solving for $R_{\text{f}}^{\text{eq}}$ gives,
\begin{equation}  \label{eq:equilibriumcov}
    R_{\text{f}}^{\text{eq}}=\frac{\lambda_y\sigma_x^2+\sigma_x\sqrt{\Psi}}{\lambda_x^2},
\end{equation}
where $\Psi=\lambda_y^2\sigma_x^2+\lambda_x^2\sigma_y^2$. Observe now that by plugging-in \eqref{eq:equilibriumcov} into the expression found in \eqref{eq:ejspectralradiuscondition},
\begin{equation*}
    \mathbf{G}^\vy-(\ms^\vy\circ\ms^\vx)(\ms^\vx\circ\ms^\vx)^{-1}\mathbf{G}^\vx\equiv G^y,
\end{equation*}
which is stationary in time and simplifies down to $G^y$ because there is no noise cross-interaction, then we retrieve
\begin{align*}
    G^y&=\lambda_y+\frac{\sigma_y^2}{R_{\text{f}}^{\text{eq}}}=\frac{1}{R_{\text{f}}^{\text{eq}}}\left(\lambda_y\frac{\lambda_y\sigma_x^2+\sigma_x\sqrt{\Psi}}{\lambda_x^2}+\sigma_y^2\right)\\
    &=\frac{1}{R_{\text{f}}^{\text{eq}}}\left(\frac{\lambda_y^2\sigma_x^2+\lambda_y\sigma_x\sqrt{\Psi}+\lambda_x^2\sigma_y^2}{\lambda_x^2}\right)=\frac{\sqrt{\Psi}}{\lambda_x^2R_{\text{f}}^{\text{eq}}} \big(\sqrt{\Psi}+\lambda_y\sigma_x\big)>0,
\end{align*}
where we have used the definition of $\Psi$, the fact that $\sqrt{\Psi}+\lambda_y\sigma_x=\sqrt{\lambda_y^2\sigma_x^2+\lambda_x^2\sigma_y^2}+\lambda_y\sigma_x>0$, and that $R_{\text{f}}^{\text{eq}}>0$. This validates condition \eqref{eq:ejspectralradiuscondition}, for sufficiently small step size $\dt$, and so the auxiliary constant $E^j$ is always within the interval $(-1,1)$, meaning \eqref{eq:ejspectralradius} is satisfied. By extension, since everything is just a scalar, this implies $\mathbf{D}^{j,n-2}\equiv D^{j,n-2}\in(-1,1)$ because of \eqref{eq:updatematrixcompact} (see Appendix \ref{sec:app2}).

\subsection{Fixed-Lag Strategy for the Online Smoother} \label{sec:3.3}
The fixed-lag smoother is the simplest method for conserving computational and storage resources, as it assumes a uniformly predetermined impact region for each new observation. Setting an a-priori impact region incurs no additional computational cost. However, this approach has significant drawbacks for most turbulent and complex dynamical systems: it may assign an unnecessarily long fixed lag, wasting storage with diminishing returns in state estimation, or a short fixed lag that introduces substantial biases in the estimated state. Nevertheless, it serves as a building block for the adaptive-lag online smoother.

The fixed-lag methodology is defined as follows: we assume a predetermined impact region size for each new observation, denoted as lag $L\in\{1,\ldots,n-1\}$ in time steps (the edge cases $L=0$ and $L=n$ are discussed later). All subsequent measurements can be expressed in simulation time units by multiplying by $\dt$. We also assume that $L$ is uniform across $n$, the number of observations. The fixed-lag online smoother Gaussian statistics are then given by the following formulae:
\begin{align}
    \vm{s}^{j,n}\equiv \vm{s}^{j,n}(L)&:=\begin{cases}
        \vm{s}^{j,n-1}, & \text{ for } 0\leq j\leq n-1-L\\
        \vm{s}^{j,n-1}+\mathbf{D}^{j,n-2}\big(\vm{s}^{n-1,n}-\vm{f}^{n-1}\big), &\text{ for } n-L\leq j\leq n-1,
    \end{cases} \label{eq:fixedlag1}\\
    \mr{s}^{j,n}\equiv \mr{s}^{j,n}(L)&:=\begin{cases}
        \mr{s}^{j,n-1}, & \text{ for } 0\leq j\leq n-1-L\\
        \mr{s}^{j,n-1}+\mathbf{D}^{j,n-2}\big(\mr{s}^{n-1,n}-\mr{f}^{n-1}\big)(\mathbf{D}^{j,n-2})^\dagger, &\text{ for } n-L\leq j\leq n-1,
    \end{cases} \label{eq:fixedlag2}
\end{align}
with the details of the online discrete smoother being outlined in Algorithm \ref{algo:onlinesmoother} of Section \ref{sec:3.1}. For $j=n$, as already discussed, we always have $\vm{s}^{n,n}=\vm{f}^n$ and $\mr{s}^{n,n}=\mr{f}^n$. As such, we have for $\forall n\in\mathbb{N}$ and $j=0,\ldots,n-1$, that $\vm{s}^{j,n}=\vm{s}^{j,n}(L)$ and $\mr{s}^{j,n}=\mr{s}^{j,n}(L)$, where this explicitly denotes the dependence on the predetermined fixed lag value of $L$.

We mention here that the edge cases of $L=0$ and $L=n$ actually recover the optimal filter and (offline) smoother state estimations, respectively, under the consideration that the current $n$ observations define a complete time series for the observable variables. It is not hard to see why this is the case. Observe that when $L=0$, then we only work with the first branch of the definitions for the fixed-lag smoother Gaussian statistics in \eqref{eq:fixedlag1} and \eqref{eq:fixedlag2}, and as such with each new observation we just carry over the estimated state from the corresponding time instant $t_j$ at the previous iteration, i.e., from the last observation. As such, since $\vm{s}^{n,n}=\vm{f}^n$ and $\mr{s}^{n,n}=\mr{f}^n$, for every $n\in\mathbb{N}_0$, then it is immediate that when $L=0$, with each new observation we just carry-on with the forward-pass and calculate the optimal filter Gaussian statistics at the new observation and then carry over all the previous state estimates. But since at each point we are calculating just the filter state, then we are essentially doing a forward-pass as the observations come in, and as such we only recover the filter Gaussian statistics. On the other hand, when $L=n$, then we only work with the second branch of the definitions for the fixed-lag smoother Gaussian statistics in \eqref{eq:fixedlag1} and \eqref{eq:fixedlag2}, which means that with each new observation we continue the forward-pass and obtain the new filter estimate at the end-point ($j=n$), but using this new state estimation we do a full backward-pass update and obtain the smoother estimates at each and every time instant. As such, we are essentially doing a full forward- and backward-pass with the arrival of each observation, which means we are recovering with each new observational data the offline smoother posterior Gaussian statistics, fully, for this complete observational period.

\subsection{Online Smoother with an Adaptive Lag Determined Using Information Theory} \label{sec:3.4}

The adaptive-lag method for defining optimal online smoother state estimates seeks to dynamically determine the impact region of each new observation while minimizing storage costs. To achieve a consistent strategy for calculating the required lag for each observation, we leverage fundamental tools from information theory and the conditional Gaussian structure of the CGNS framework. This approach enables the optimally estimated state from the online forward-in-time discrete smoother to more effectively capture intermittency, extreme events, and nonlinear dynamics that influence the transient behavior of complex and turbulent systems. In the following subsections, we first introduce key concepts from information theory, which contribute to the development of the adaptive-lag strategy for the online smoother.

\subsubsection{Information Theory Fundamentals: Relative Entropy and the Signal--Dispersion Decomposition} \label{sec:3.4.1}

Due to the turbulent dynamics, new observational data influence the estimated states through the online smoother update only within a finite time interval, with this impact exhibiting exponential decay over time. Therefore, it is natural to quantify the information gain from applying the optimal online forward-in-time discrete smoother (posterior), with varying lengths of lags, compared to simply carrying forward the previous state estimate (prior). By doing so, the optimal lag can be discovered. We utilize empirical information theory to develop a measure that quantifies this information gain \cite{cover2005elements}, reflecting the additional information in the posterior distribution beyond the prior \cite{majda2012lessons, branicki2012quantifying, giannakis2012information}. In the information-theoretic framework, this is naturally achieved through their respective PDFs, by using a certain type of distance function that measures the statistical discrepancy between them. A general class of such functions is known as $f$-divergences \cite{csiszar2004information, liese2006divergences}.

While the general framework of $f$-divergences is rather versatile, there is a certain choice of $f$ which takes great advantage of the form of the CGNS framework, and especially of the conditional Gaussianity of the posterior distributions. Specifically, by choosing $f(z)=z\log z$, then we recover the so called Kullback--Leibler divergence or relative entropy \cite{kullback1951information, kullback1997information}. In this case, we denote the divergence with $\mathcal{P}$, and is simply given as \cite{majda2010quantifying, kleeman2011information}
\begin{equation} \label{eq:relativeentropy}
    \mathcal{P}(p,q)=\int_{\mathbb{C}^n} p(\mathbf{u})\log\left(\frac{p(\mathbf{u})}{q(\mathbf{u})}\right)\d\mathbf{u},
\end{equation}
where the integration with respect to the state vector is to be understood in the Lebesgue sense\footnote{We implicitly assume in \eqref{eq:relativeentropy} that $p$ and $q$ are the PDFs of two probability distributions, $\pp$ and $\mathbb{Q}$, respectively, with $\pp\ll\mathbb{Q}$ (i.e., $\pp$ is dominated by $\mathbb{Q}$), so the integration is taken over the support of $q$, where both $\pp$ and $\mathbb{Q}$ are dominated by the Lebesgue measure.}. For our framework, we read $\mathcal{P}(p,q)$ as the information gain of $p$ from $q$, as in the information that has been gained by using the true density $p$ in lieu of the approximation $q$. The relative entropy enjoys the lucrative properties of non-negativity and of invariance under general nonlinear changes of state variables \cite{cai2002mathematical}. These properties, among others, allow the relative entropy to be widely utilized in assessing the information gain and quantifying model error, model sensitivity, predictability, and statistical response for stochastic CTNDSs \cite{branicki2014quantifying, kleeman2002measuring, delsole2004predictability, delsole2005predictability, branstator2010two, liu2016predictability, andreou2024statistical}.

The relative entropy ensures that the extreme events related to tail probabilities are not underestimated. The ratio of the two PDFs within the logarithm quantifies the gap in the tail probability, allowing for an unbiased characterization of statistical differences. Many CTNDSs modeled through the CGNS framework generate intermittent, rare, and extreme events. Thus, using a logarithmic generator function effectively resolves tail probability events, aiding in quantifying information gain for the adaptive lag criterion, uncertainty reduction of posterior solutions, and causal inference. In addition, relative entropy benefits from this logarithmic structure, providing a simple, closed-form formula when both distributions are Gaussian. This is ideal for the online smoother update, where both the prior and posterior distributions are Gaussian, and differ only in their statistics. Consequently, this facilitates the efficient calculation of the adaptive lag, thus further reducing computational costs and storage requirements. Specifically, when both $\pp\overset{\d}{\sim}\mathcal{N}_N\big(\vm{p},\mr{p}\big)$ and $\mathbb{Q}\overset{\d}{\sim}\mathcal{N}_N\big(\vm{q},\mr{q}\big)$, where $N=\mathrm{dim}(\mathbf{u})$ is the dimension of the phase space, then the relative entropy adopts an explicit and simple formula \cite{majda2006nonlinear, kleeman2002measuring}:
\begin{equation} \label{eq:signaldispersion}
    \mathcal{P}(p,q)=\frac{1}{2}(\vm{p}-\vm{q})^\dagger\mr{q}^{-1}(\vm{p}-\vm{q})+\frac{1}{2}\left(\mathrm{tr}(\mr{p}\mr{q}^{-1})-N-\log(\mathrm{det}(\mr{p}\mr{q}^{-1}))\right),
\end{equation}
where $\mathrm{tr}(\cdot)$ and $\mathrm{det}(\cdot)$ are the trace and determinant of a matrix, respectively. The first quadratic form term in \eqref{eq:signaldispersion} is called the signal and measures the information gain in the mean weighted by the model or approximation covariance, whereas the second term is called the dispersion and involves only the covariance ratio $\mr{p}\mr{q}^{-1}$. This is why the expression in \eqref{eq:signaldispersion} is known as the signal--dispersion decomposition of the relative entropy for Gaussian variables.

\subsubsection{Development of the Adaptive-Lag Strategy for the Online Smoother} \label{sec:3.4.2}

Here we outline the information-theoretic approach for defining the adaptive lag at the acquirement of a new observation. Let us assume that we have currently observed the $n$-th measurement, $\vx^n$, which is used to update all the existing online smoother Gaussian statistics at time instants $0\leq j\leq n-1$, and then provide a new state estimate at $j=n$. As shown in Theorem \ref{thm:onlinesmoother}, that update happens via \eqref{eq:onlinerecursive1}--\eqref{eq:onlinerecursive2}. We would like to define an adaptive lag at the $n$-observation, $L_n\in\{1,\ldots,n-1\}$, such that only ``the most informative or impactful" of these updates are actually carried out, in other words,
\begin{align}
    \vm{s}^{j,n}\equiv \vm{s}^{j,n}(L_n)&:=\begin{cases}
        \vm{s}^{j,n-1}, & \text{ for } 0\leq j\leq n-1-L_n\\
        \vm{s}^{j,n-1}+\mathbf{D}^{j,n-2}\big(\vm{s}^{n-1,n}-\vm{f}^{n-1}\big), &\text{ for } n-L_n\leq j\leq n-1,
    \end{cases} \label{eq:adaptivelagmean}\\
    \mr{s}^{j,n}\equiv \mr{s}^{j,n}(L_n)&:=\begin{cases}
        \mr{s}^{j,n-1}, & \text{ for } 0\leq j\leq n-1-L_n\\
        \mr{s}^{j,n-1}+\mathbf{D}^{j,n-2}\big(\mr{s}^{n-1,n}-\mr{f}^{n-1}\big)(\mathbf{D}^{j,n-2})^\dagger, &\text{ for } n-L_n\leq j\leq n-1,
    \end{cases} \label{eq:adaptivelagcov}
\end{align}
with the details of the online discrete smoother being outlined in Algorithm \ref{algo:onlinesmoother} of Section \ref{sec:3.1}. For $j=n$, as already discussed, we always have $\vm{s}^{n,n}=\vm{f}^n$ and $\mr{s}^{n,n}=\mr{f}^n$.

To define the adaptive lag $L_n$ using the information theory principles from Section \ref{sec:3.4.1}, we quantify the information gain from updating the state estimate at the $j$-th time step with the $n$-th observation, compared to using the estimate from the previous $(n-1)$-st observation.
In the CGNS framework, two possible Gaussian distributions arise for the online smoother state estimation. The first, $p^{j,n}_{\text{updated}}$, is the optimal posterior distribution obtained through the online forward-in-time smoother, discussed in Section \ref{sec:3.1}. The second, $p^{j,n}_{\text{lagged}}$, is the prior distribution, where no update is made and the previous state estimate at time $t_j$ is carried over instead. The posterior or the ``updated'' distribution has Gaussian statistics given by
\begin{equation*}
    \vm{s}^{j,n-1}+\mathbf{D}^{j,n-2}\big(\vm{s}^{n-1,n}-\vm{f}^{n-1}\big) \quad\text{and}\quad
    \mr{s}^{j,n-1}+\mathbf{D}^{j,n-2}\big(\mr{s}^{n-1,n}-\mr{f}^{n-1}\big)(\mathbf{D}^{j,n-2})^\dagger,
\end{equation*}
for its posterior mean and covariance tensor, respectively, as indicated in \eqref{eq:onlinerecursive1}--\eqref{eq:onlinerecursive2}. On the other hand, the prior or ``lagged" distribution has the corresponding Gaussian statistics of
\begin{equation*}
    \vm{s}^{j,n-1} \quad\text{and}\quad
    \mr{s}^{j,n-1},
\end{equation*}
i.e., no update (or innovation, under the Kalman gain terminology) is being made or procured due to the newly acquired observation. This is precisely described by the first branch in \eqref{eq:adaptivelagmean} and \eqref{eq:adaptivelagcov}.

With the posterior or ``updated" conditional Gaussian distribution denoting the true distribution, while the prior or ``lagged" conditional Gaussian distribution denotes the approximation due to the adaptive lag (by not carrying out the online smoother update), we can then naturally take advantage of the signal--dispersion decomposition formula for the relative entropy given in \eqref{eq:signaldispersion}. As such, we have that the information gain or uncertainty reduction incurred by carrying out the update through the online smoother at time $t_j$ by using the newly obtained $n$-th observation, is equal to
\begin{align}
\begin{split}
    \mathcal{P}\left(p^{j,n}_{\text{updated}}, p^{j,n}_{\text{lagged}}\right)=&\ \frac{1}{2}\big(\vm{s}^{n-1,n}-\vm{f}^{n-1}\big)^\dagger(\mathbf{D}^{j,n-2})^\dagger(\mr{s}^{j,n-1})^{-1}\mathbf{D}^{j,n-2}\big(\vm{s}^{n-1,n}-\vm{f}^{n-1}\big)\\
    &+\frac{1}{2}\big(\mathrm{tr}(\mathbf{Q}^{j,n})-l-\mathrm{ln}(\mathrm{det}(\mathbf{Q}^{j,n}))\big),
\end{split} \label{eq:lackinfoupdate}
\end{align}
where $\mathbf{Q}^{j,n}$ is defined as a covariance ratio matrix
\begin{align}
\begin{split}
    \mathbf{Q}^{j,n}&:=\big(\mr{s}^{j,n-1}+\mathbf{D}^{j,n-2}\left(\mr{s}^{n-1,n}-\mr{f}^{n-1}\right)(\mathbf{D}^{j,n-2})^\dagger\big)(\mr{s}^{j,n-1})^{-1}\\
    &=\mathbf{I}_{l\times l}+\big(\mathbf{D}^{j,n-2}\left(\mr{s}^{n-1,n}-\mr{f}^{n-1}\right)(\mathbf{D}^{j,n-2})^\dagger\big)(\mr{s}^{j,n-1})^{-1}.
\end{split} \label{eq:covarianceratio}
\end{align}
The first term in \eqref{eq:lackinfoupdate} corresponds to the signal, while the latter to the dispersion, terms already introduced in Section \ref{sec:3.4.1}. With \eqref{eq:lackinfoupdate}--\eqref{eq:covarianceratio} in-hand, we can now outline the procedure for determining $L_n$. Having a predefined upper lag bound, hereby denoted by $b\in\mathbb{N}$, we calculate
\begin{equation*}
    \mathcal{P}\left(p^{j,n}_{\text{updated}}, p^{j,n}_{\text{lagged}}\right) \text{ for } R_n\leq j\leq n-1,
\end{equation*}
where $R_n:=\max\{n-1-b, 1\}$, thus totaling
\begin{equation*}
    n-R_n=\begin{cases}
        b+1, & \text{ if } n-1-b\geq 1 \\
        n-1, & \text{ if } n-1-b<1,
    \end{cases}
\end{equation*}
values. Two possible approaches to define the adaptive lag, both yielding similar results after adjusting hyperparameters, are: we can either use the sequence of relative entropies $\left\{ \mathcal{P}\left(p^{j,n}_{\text{updated}}, p^{j,n}_{\text{lagged}}\right)\right\}_{R_n\leq j\leq n-1}$, or apply the local standard deviation (LSDev) $\{\sigma^{j,n}\}_{R_n\leq j\leq n-1}$ to the original entropy sequence and use that instead. The LSDev of the entropy sequence serves as a proxy for the first derivative of the information gain (when estimated via a finite difference scheme), efficiently identifying regions of stagnation or rapid change in it. Specifically $\sigma^{j,n}$, for a moving window span size of $w\equiv 1 (\mathrm{mod}\ 2)$, is the standard deviation of the neighborhood centered at the corresponding input value (or pixel in image filtering) over $j$, extending $(w-1)/2$ positions to the left and to the right, for each $n$. This can potentially allow for more effective adaptive lag selection. Both methods behave similarly in most numerical experiments, including all case studies in Section \ref{sec:4}, especially when the information gain sequence behaves exponentially with respect to $j$ for each $n$. In such cases, both the relative entropy and its LSDev exhibit the same growth, making the approaches roughly equivalent. For simplicity, in what follows, we present the adaptive lag procedure using the LSDev. However, the original relative entropy sequence of information gains could also be used directly. Naturally, when using $\{\sigma^{j,n}\}_{R_n\leq j\leq n-1}$, we additionally have to consider the dependence of the procedure on the window size being used to define the moving LSDev. Taking this into consideration, for a moving window span size of $w\in\mathbb{N}_{\geq 3}$, with $w\equiv 1 (\mathrm{mod}\ 2)$, which gives $\{\sigma^{j,n}\}_{R_n\leq j\leq n-1}=\{\sigma^{j,n}(w)\}_{R_n\leq j\leq n-1}$, and a tolerance parameter $\delta\in(0,1)$, the definition of the adaptive lag at the $n$-th observation in this framework is given as
\begin{equation} \label{eq:adaptlaginfodef1}
    L_n\equiv L_n(b,w,\delta):=n-1-\mathrm{max}\big\{R_n\leq j\leq n-1:\sigma^{j,n}(w)<\delta\big\}\Big(\in\{0,1,\ldots,b\}\Big),
\end{equation}
where if no such maximizer exists, then we simply set
\begin{equation} \label{eq:adaptlaginfodef2}
    L_n:=\min\{n-1, b\}.
\end{equation}
(Throughout all numerical case studies in this work, we set $w=7$ which is MATLAB's default value for \texttt{stdfilt()}.) Notice that for larger values of $\delta$ we allow for more relaxed or smaller lags, while for smaller tolerance values we push the adaptive lag values towards their predefined upper bound, $b$. Per our methodology, and as noted in \eqref{eq:adaptlaginfodef1}, we have that $\forall n\in\mathbb{N}$, $L_n\in\{0,1,\ldots,b\}$.

As aforementioned, it is possible to instead use the sequence of information gains directly to define the adaptive lag,
\begin{equation} \label{eq:adaptlaginfodeforiginalseq}
    L_n\equiv L_n(b,\delta):=n-1-\mathrm{max}\Big\{R_n\leq j\leq n-1:\mathcal{P}\left(p^{j,n}_{\text{updated}}, p^{j,n}_{\text{lagged}}\right)<\delta\Big\},
\end{equation}
where again if no such maximiser exists we simply use \eqref{eq:adaptlaginfodef2}. For most cases this approach leads to the same results after a slight modification to the tolerance parameter being used (and for appropriate values of $w$ in \eqref{eq:adaptlaginfodef1}; see the numerical results in Figure \ref{fig:LDA_Ice_Floes_Fig_3} of Section \ref{sec:4.2}, where both $\sigma^{j,n}$ and $\mathcal{P}\left(p^{j,n}_{\text{updated}}, p^{j,n}_{\text{lagged}}\right)$ behave as $C^n_1e^{C^n_2j}$ for $R_n\leq j\leq n-1$ and $C^n_1,C^n_2>0$). In such settings, if we use $\delta=O(10^{-d})$ with the sequence of local standard deviations, for $d\in\mathbb{N}$, then we can get similar results when using the original sequence of relative entropies by using $\delta=O(10^{-d+s/2})$ where $s\in\mathbb{N}$ is defined through the time step as $\dt=O(10^{-s})$, since the former is approximately the first derivative of the latter with respect to the time step, as already mentioned. Furthermore, the simplicity and flexibility of the adaptive lag framework in \eqref{eq:adaptlaginfodef1} or \eqref{eq:adaptlaginfodeforiginalseq} further allows for the inclusion of a penalty for large lag values. This is achieved by modifying the criterion (e.g., by using a monotonic polynomial of time as the penalization) and then choosing the largest lag for which this penalized information gain (or its penalized local standard deviation) remains small under the tolerance $\delta$, while prioritizing smaller lags.

It is important to note here that in practice, i.e., from a computational standpoint, the adaptive lag in \eqref{eq:adaptlaginfodef1} and \eqref{eq:adaptlaginfodeforiginalseq} is calculated by a backtracking search approach, where $j$ starts from $n-1$ and moves backwards towards $R_n$ until the condition in \eqref{eq:adaptlaginfodef1} or \eqref{eq:adaptlaginfodeforiginalseq} is met. Therefore, it is not required to calculate the full $\left\{ \mathcal{P}\left(p^{j,n}_{\text{updated}}, p^{j,n}_{\text{lagged}}\right)\right\}_{R_n\leq j\leq n-1}$ or $\{\sigma^{j,n}\}_{R_n\leq j\leq n-1}$ sequences and then identify the corresponding ($\mathrm{arg}$)$\max$.

\section{Applications of the Adaptive-Lag Online Smoother} \label{sec:4}
This section demonstrates the application of the online smoother in addressing three key scientific challenges: (a) state estimation and causality detection, especially in the presence of intermittency and extreme events, (b) evaluating computational and storage efficiency in a high-dimensional Lagrangian data assimilation application, and (c) developing an online parameter estimation algorithm with partial observations.

The simulations for all the numerical experiments in this work were carried out on an AMD Ryzen\textsuperscript{TM} 7 5800H mobile CPU of 8 cores, 16 threads, with an average clock speed of 3.2 GHz (turbo boosts up to 4.4 GHz), with a TDP of 45W. In terms of memory, a two-memory-module configuration was used, with 2x16GB DDR4 RAM sticks at 3200MT/s.

\subsection{Detecting Causal Evidence Using the Adaptive-Lag Online Smoother} \label{sec:4.1}

Causal inference helps identify the mechanisms that drive the dynamics of a system. It advances more accurate predictions, better control, and improved decision-making under uncertainty \cite{sun2015causal, hlavavckova2007causality}. It also reveals the role of extreme events and intermittency in modulating the dynamical evolution of the underlying system. By clarifying cause-effect relationships, causal inference enhances our ability to model, simulate, and intervene in complex systems. This subsection demonstrates how the adaptive online smoother and the associated information gain, using a simple nonlinear dyad model, can reveal causal links by showing how time-delayed information in one process improves state estimation in another. The model generates intermittency and extreme events, and the smoother identifies how long temporal information is needed for causal detection during different phases.

The dyad model considered here is the following stochastic model with quadratic nonlinearities, multiplicative and cross-interacting noise, and physics-constraints (energy conservation in the quadratic nonlinear terms):
\begin{align}
    \d u(t)&=\big(-d_uu(t)+\gamma u(t)v(t)+F_u\big)\d t+\sigma_u\d W_u(t), \label{eq:dyad1} \\
    \d v(t)&=\big(-d_vv(t)-\gamma u(t)^2+F_v\big)\d t+\sigma_{vu}u(t)\d W_u(t)+\sigma_v\d W_v(t), \label{eq:dyad2}
\end{align}
where all parameters are constants, with $d_u,d_v,\sigma_u,\sigma_v>0$, and $\gamma,\sigma_{vu},F_u,F_v\in\rr$, while $W_u$ and $W_v$ are two mutually independent Wiener processes. In \eqref{eq:dyad1}--\eqref{eq:dyad2}, The variable $u$ represents an observed or resolved mode in the turbulent signal, interacting with the unresolved mode $v$ through quadratic nonlinearities. As $v$ influences the damping of $u$, variations in $v$ cause the time series of $u$ to display intermittency and extreme events. The system \eqref{eq:dyad1}--\eqref{eq:dyad2} fits exactly the CGNS framework in \eqref{eq:condgauss1}--\eqref{eq:condgauss2}, with $\vx=x:=u$ and $\vy=y:=v$. The following parameter values are used for this numerical case study:
\begin{gather*}
    d_u = 0.5,\quad \gamma = 3, \quad F_u = 1, \quad \sigma_u=0.6, \\
    d_v = 0.5, \quad F_v = 0.3, \quad \sigma_{vu}=0.8, \quad \sigma_v=1.
\end{gather*}
The parameters are assigned to induce intermittent and extreme events in the observable variable $u$ through the stochastic damping generated by the unobserved variable $v$, with the goal of producing highly non-Gaussian PDFs for both observable and unobserved dynamics. These values are also chosen to address potential issues of observability \cite{majda2012filtering, ogata2010modern, dorf2017modern}, which is linked to causal detection using the adaptive online smoother. The coupled system in \eqref{eq:dyad1}--\eqref{eq:dyad2} loses practical observability when the observed process $u$ provides no information about the unobserved variable $v$. Intuitively, this occurs during phases where $u \approx 0$, rendering $v$ ineffective in contributing to the dynamics of $u$. A numerical integration time step of $\dt=0.005$ is used, which also defines the uniform rate of obtaining the observations for the online smoother, with a total simulation time of $T = 60$ units.

Figure \ref{fig:Dyad_Interaction_Fig_1} presents various plots comparing the filter and offline smoother posterior estimated states at different levels (this is equivalent to an online smoother which uses $L_n=n$ at each $t_n=n\dt$, where $n$ is the current number of observations; see Section \ref{sec:3.3}). Panel (a) displays the trajectory of the observed variable $u$, while Panel (b) displays that for the unobserved variable $v$. For the latter, we also include the posterior means calculated using Theorems \ref{thm:filtering} and \ref{thm:smoothing}, along with the first two standard deviations away from the mean, represented by accordingly colored shaded regions. The posterior estimated states overall follow the true values with relatively low uncertainty. As expected, the smoother provides a more accurate estimation during the uncertain quiescent periods of the unobserved variable. During the extreme events which are characterized by a high signal-to-noise ratio, both methods lead to smaller uncertainty. Panel (c) plots the information gain of the filter and smoother posterior distributions relative to the prior statistical attractor (or equilibrium distribution) of $v$ on a logarithmically scaled y-axis. Since both the equilibrium (computed numerically using many long simulated trajectories; see Panel (f)), $p_{\text{att}}(v)$, and posterior distributions (at each time instant) are Gaussian, the information gain in Panel (c) is calculated using the signal--dispersion decomposition formula of relative entropy in \eqref{eq:signaldispersion}. The temporal evolution of relative entropy reveals that the information gain from the smoother posterior Gaussian distribution is almost everywhere larger than that from the filter solution, as is expected. Notably, the information gain corresponding to the onset of the intermittent phases of $u$ shoots up due to the strong signal-to-noise ratios. Note that the information gain using the smoother is more significant than the filter at these phases and the more uncertain quiescent phases due to additional observational information. Panel (d) shows the time-averaged PDF of $u$ generated by its signal over $t\in[20,60]$, alongside the corresponding Gaussian fit, i.e., the normal distribution with the same mean and standard deviation as the truth over the associated period. Due to the intermittent extreme events, we can discern that the induced density is highly non-Gaussian, with an apparent positive skewness and a heavy tail. Panel (e) is the same as (d) but for $v$, which also includes the densities corresponding to the filter and smoother solutions. We can again see that the density corresponding to $v$ displays some non-Gaussian features. Furthermore, as is expected, the PDF of the smoother solution is better equipped to approximate that of the truth compared to the filter one by better resolving the tail behavior and first few ordered moments.

\begin{figure}[!ht]
\centering
    \includegraphics[width=\textwidth]{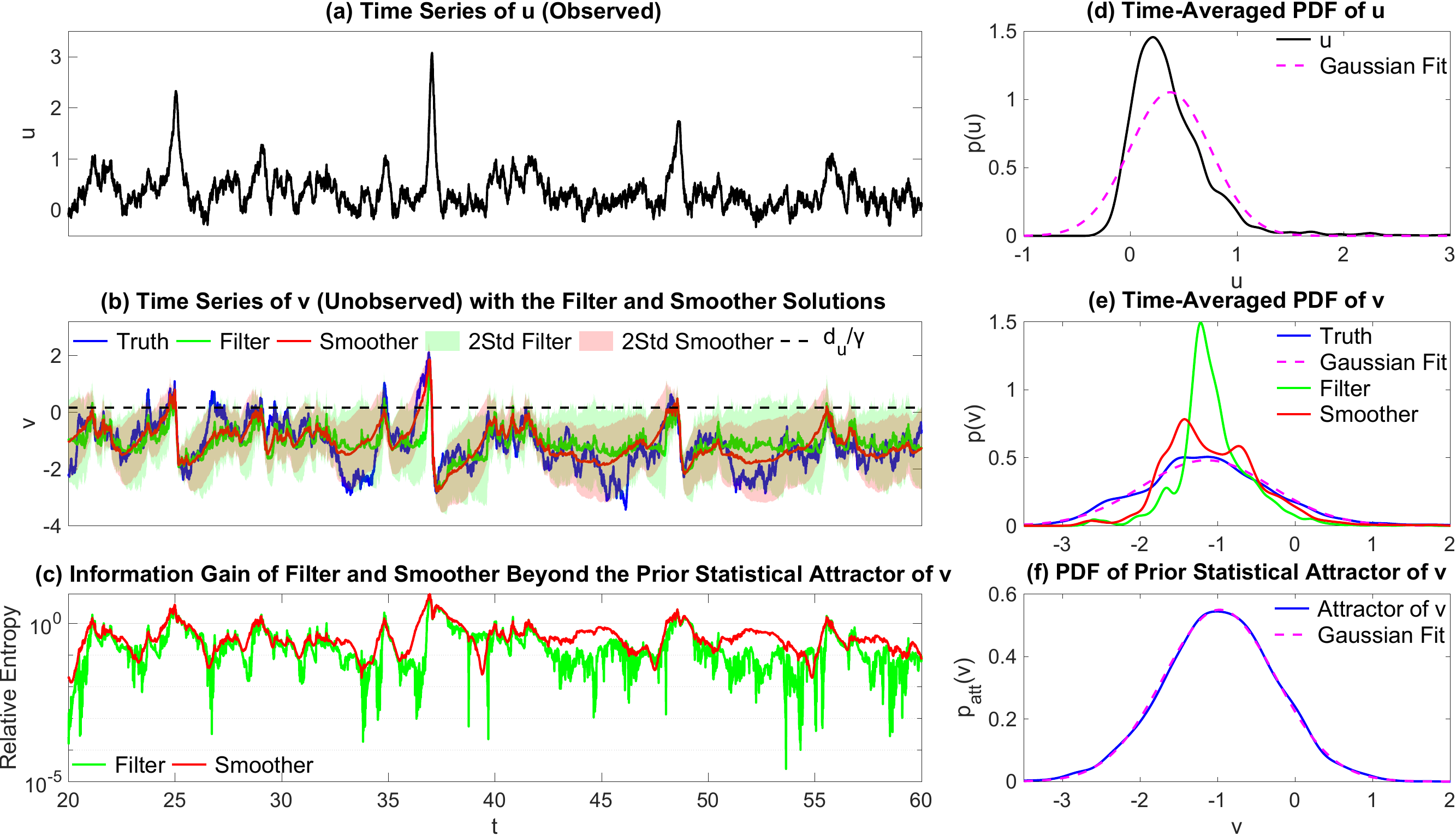}
    \caption{Panel (a): Trajectory of the observable variable $u$. Panel (b): True trajectory of the unobserved variable $v$ (in blue), alongside the filter (in green) and offline smoother (in red) posterior mean time series. Their respective first two standard deviations away from the mean state are also plotted through correspondingly colored shaded regions, with the standard deviation in this case being exactly equal to the square root of the respective posterior covariance. The threshold above which $v$ acts as anti-damping to $u$, i.e., the line $y=d_u/\gamma=1/6$, is also plotted in a dashed black line. Panel (c): Temporal evolution of the information gain of the filter and smoother posterior Gaussian statistics beyond the statistical attractor of the unobserved variable. This is calculated using the signal--dispersion decomposition of the relative entropy in \eqref{eq:signaldispersion}. A logarithmic scale is used for the y-axis. Panel (d): Time-averaged PDF of $u$ calculated using the observations over $t\in[20,60]$ (in black), alongside its Gaussian fit density defined by the mean and standard deviation of the signal in the same time period (in dashed magenta). Panel (e): Same as (d) but for the unobserved variable $v$, together with the PDFs corresponding to the filter and smoother posterior mean time series.  Panel (f): PDF of the unobserved variable's prior statistical attractor, $p_{\text{att}}(v)$.}
    \label{fig:Dyad_Interaction_Fig_1}
\end{figure}

Figure \ref{fig:Dyad_Interaction_Fig_2} focuses on a shorter intermittent period $\big(t\in[32.5,40]\big))$ to assess the performance of the adaptive-lag online smoother using the original sequence of information gains, $\left\{ \mathcal{P}\left(p^{j,n}_{\text{updated}}, p^{j,n}_{\text{lagged}}\right)\right\}_{R_n\leq j\leq n-1}$, with hyperparameters $b=600$ (which corresponds to $b\dt=3$ simulation time units) and $\delta=10^{-4}$. For this case study, similar results can be obtained when instead using the LSDev sequence, $\{\sigma^{j,n}\}_{R_n\leq j\leq n-1}$, as already discussed in Section \ref{sec:3.4.2}, by appropriately adjusting $\delta$ (see also Panel (g) of Figure \ref{fig:Dyad_Interaction_Fig_Appendix}). Panel (a) shows the signals of $u$ and $v$ over this period, while Panels (b)--(f) illustrate how the online smoother adjusts its state estimation as new observations arrive at five different time instants during this period. The online smoother is compared against the filter solution, which can also produce real-time estimated states with the sequential arrival of observations. We focus on Panels (c) and (d). When using the filter solution, the estimation of $v$ before $u$ reaches the peak of the extreme event contains significant errors. This is because the growth of $v$, once it exceeds the threshold of $d_u/\gamma=1/6$ to become anti-damping, triggers the extreme event in $u$ (the cause). However, the extreme event (the effect) does not appear instantly, thus introducing a delay in the onset of these instabilities. This is why at $t=36.50$ (Panel (c)), without seeing the future information from the development of the extreme event, the causal variable $v$ cannot be accurately recovered by the filter. By extension, due to the data from the extreme event in the observed signal simply not being available yet at the current time instant, this, in a sense, contaminates the online smoother state, which faithfully follows the filter one regardless of how long the adaptive lag is. Nevertheless, at $t=37.00$ (Panel (d)), with the observed trajectory currently undergoing the effects of the instability caused by $v$, this allows the adaptive online smoother, which utilizes future information for state estimation unlike the filter one, to correct its previous estimated states based on the value of the calculated adaptive lag, at most for $b\dt=3$ time units into the past. This is explicitly seen by the online smoother posterior mean, which now better approximates the true signal compared to its state half a time unit ago, thus demonstrating its property to systematically correct biases when future data become available. This attribute is crucial to help identify the current state, therefore, outperforming the filter solution. This crucial discrepancy between the filter and online smoother estimates reveals that information in the dynamics flows from the current state of the unobserved variable (cause) to the future state of the observed one (effect), highlighting a causal relationship between these states. Panel (g) presents the adaptive lag values generated during this period, from using the sequence of information gains or relative entropies as outlined in \eqref{eq:adaptlaginfodeforiginalseq} and \eqref{eq:adaptlaginfodef2}. The lag in the online smoother likewise reflects the aforementioned time delay in the causal relationship, with a peak in adaptive lag values at the onset of the extreme event, i.e., the causal period. These are essential to correct past errors and effectively recover the state of the cause $v$, which triggered the extreme event. In contrast, smaller lag values are needed during the effect period, specifically during the demise of this intermittent event, which is characterized by a high signal-to-noise ratio where both the filter and smoother estimates can effectively capture the relatively deterministic behavior. As such, lower adaptive lag values are required since the ``lagged" distribution in \eqref{eq:lackinfoupdate} can yield a quickly decaying information gain function as $j$ decreases from $n-1$ to $R_n$ (or to an already small sequence of information gains for these values of $n$, i.e., below the threshold parameter $\delta$ uniformly in $j$), thus leading to small values for $L_n$ in \eqref{eq:adaptlaginfodeforiginalseq}. After all, the lag essentially functions as an interpolation between the past state estimates and a ``full" (since $b$ limits how much into the past we can correct) backward-run smoother solution of the online smoother. As such, a close-to-zero lag value signifies that more trust is put in the statistics from the previously calculated online posterior distributions since the observations obtained during this period have tiny impact regions.

Recall that $\left\{ \mathcal{P}\left(p^{j,n}_{\text{updated}}, p^{j,n}_{\text{lagged}}\right)\right\}_{R_n\leq j\leq n-1}$, defined by \eqref{eq:lackinfoupdate}--\eqref{eq:covarianceratio}, is used to determine the adaptive lags depicted in Panel (g) (via \eqref{eq:adaptlaginfodeforiginalseq} and \eqref{eq:adaptlaginfodef2}). We define the standardization of this sequence in the following manner:
\begin{equation} \label{eq:infogainstandardized}
    \frac{\displaystyle\left|\mathcal{P}\left(p^{j,n}_{\text{updated}}, p^{j,n}_{\text{lagged}}\right)-\mathcal{P}\left(p^{0,n}_{\text{updated}}, p^{0,n}_{\text{lagged}}\right)\right|}{\displaystyle\max_{R_n\leq j\leq n-1}\left\{\left|\mathcal{P}\left(p^{j,n}_{\text{updated}}, p^{j,n}_{\text{lagged}}\right)-\mathcal{P}\left(p^{0,n}_{\text{updated}}, p^{0,n}_{\text{lagged}}\right)\right|\right\}}, \quad \forall n\in\mathbb{N}, \quad R_n\leq j\leq n-1.
\end{equation}
(Other forms of standardization yield similar results, i.e., by dividing by the standard deviation instead of the maximum over $j$.) The standardization of the sequence of its LSDevs, $\{\sigma^{j,n}\}_{R_n\leq j\leq n-1}$, is similarly defined. Henceforth, when we refer to the standardized variant of these sequences, we imply the sequence in $j$ and $n$ defined by \eqref{eq:infogainstandardized}.

In Panel (h) of Figure \ref{fig:Dyad_Interaction_Fig_2}, the standardized sequence of information gains or relative entropies (as defined in \eqref{eq:infogainstandardized}) is plotted, on a logarithmic scale, as a function of $n$ (for $n\dt\in[32.5,40]$) and of $n\dt-j\dt$ (for $j\dt\in[R_n\dt,n\dt]$). (As a reminder, for the $n$-th observation of $u$, the last index $j\in\{R_n,\ldots,n-1\}$ for which $\mathcal{P}\left(p^{j,n}_{\text{updated}}, p^{j,n}_{\text{lagged}}\right)$ remains below the tolerance value of $\delta=10^{-4}$, is $n-1-L_n$.) Of significance here is the observational interval of $n\dt\in[35.5,37.5]$, which covers the cause and effect of the extreme event. Up until around $t=37.00$, which is on the cusp of the extreme event's completion and $u$'s subsequent decline, the cause is fully captured with a persisting standardized information gain throughout the period which generated said extreme event. This explains the increasing and large adaptive lags leading up to the extreme event (i.e., during its generation), depicted in Panel (g). But, as soon as the effect in $u$, generated by $v$, ceases to exist after $t=37.00$, then a causal role reversal is initiated, with $u$ now being the significant factor in the system's dynamics and in driving $v$'s decrease or damping via the $-\gamma u^2$ term. As a result, the standardized information gain shows an extremely rapid decay during this regime-switch period, where the high signal-to-noise ratio leads to negligible impact regions for the observations of $u$ on $v$'s online smoother state estimation (into the past, during updates). This has the ulterior outcome of smaller adaptive lag values during this period, as shown in Panel (g), compared to the other periods. Importantly, this behavior persists over the following quiescent period, whenever $n\dt-j\dt$ bridges over and before $t=37.00$.

Finally, in Panel (i) of Figure \ref{fig:Dyad_Interaction_Fig_2}, the spectral radii of the update matrices $\mathbf{D}^{j,n-2}=D^{j,n-2}$ (i.e., their absolute value), are plotted in the same manner as the standardized information gains shown in Panel (h) (though, not standardized). First, it is important to note that for all $j$ and any $n$ the spectral radii of the update matrices stay below the threshold of $1$, which ensures the convergence of the online smoother as an iterative contraction mapping. Second, the spectral radii, or in this trivial case of a one-dimensional latent state space, the absolute value of the update operator $\mathbf{D}^{j,n-2}=D^{j,n-2}$, for every $n$, behaves overall as an exponentially decreasing function as $j$ decreases down from $n-1$ (or as $n-j$ increases from $0$), which indicates the exponential decrease of the impact region of each new observation by the turbulent dynamics, as discussed in Section \ref{sec:3.2.1}. As discussed in the prequel and observed in Panel (h), similarly here for $n$ values such that $R_n\dt\leq j\dt\leq (n-1)\dt$ covers an extreme event, the emergence of an instability in the observed signal induces a rapid regime switch and an additional steep exponential decrease in the spectral radius values, i.e., the exact influence region of this observation, for when $j$ passes through these intermittent periods. Finally, and most remarkably, is the fact that the information-based criterion we have developed in Section \ref{sec:3.4.2}, after standardization, showcases an impressively similar temporal behavior (both in observational (in $n$) and natural (in $j$) time), with slightly differing scales, as the exact metric that defines the observations' impact or influence region on online smoother state estimation, that being the spectral radii of the update operators (see also Figure \ref{fig:LDA_Ice_Floes_Fig_3}). This is true not just during extreme events, as already discussed, but also during quiescent periods. This provides strong numerical corroboration for our methodology of defining the adaptive lag as an extremely cheap (when compared to the exact method) yet effective approach.

\begin{figure}[!ht]
\centering
    \includegraphics[width=\textwidth]{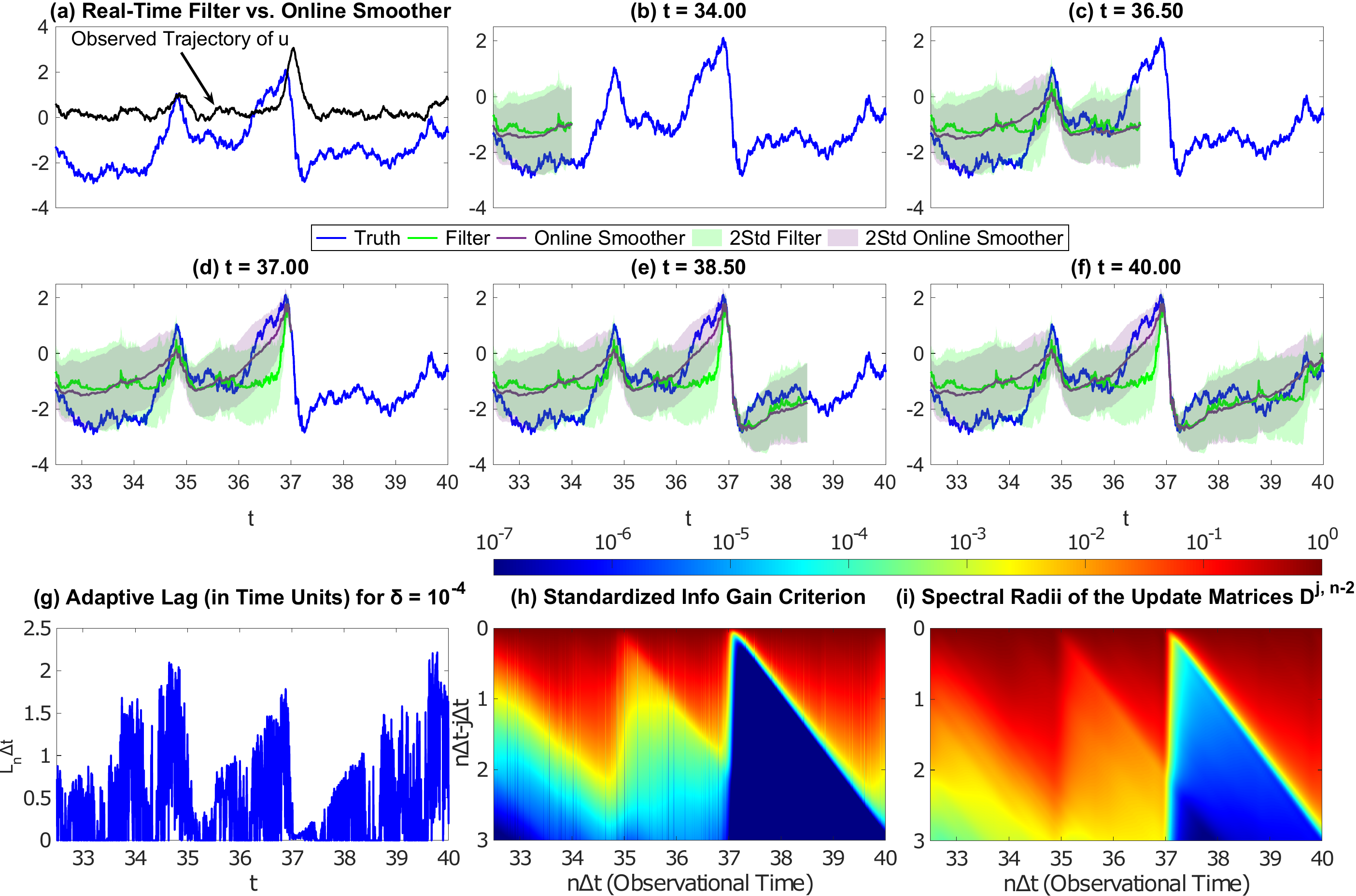}
    \caption{Panels (a)--(f): Real-time comparison between the filter posterior mean time series (in green) and the one generated from the adaptive-lag online smoother strategy by Theorem \ref{thm:onlinesmoother} (in purple). The observed trajectory of $u$ is plotted in Panel (a) (in black), while the true trajectory of $v$ is showed in all panels for reference (in blue). Panel (g): Adaptive lag values (measured in time units, i.e., $L_n\dt$) generated by the algorithm defined in Section \ref{sec:3.4.2}, specifically through \eqref{eq:adaptlaginfodeforiginalseq} and \eqref{eq:adaptlaginfodef2} (i.e., using the original sequence of relative entropies). Panel (h): Standardized information gain criterion in \eqref{eq:infogainstandardized} as a function of $n$ (for $n\dt\in[32.5,40]$) and of $n\dt-j\dt$ (for $j\dt\in[R_n\dt,n\dt]$). Plotted on a logarithmically scaled colormap. Panel (i): Same as Panel (h) but for the spectral radii of the update values $D^{j,n-2}$, i.e., $|D^{j,n-2}|$. Panels (h) and (i) share the same colorbar, and have their y-axis flipped.}
    \label{fig:Dyad_Interaction_Fig_2}
\end{figure}

For completeness, since a study into the trade-off between the computational and storage-wise complexity and lower-order pathwise error of the fixed- and adaptive-lag online discrete smoother solutions is provided for the high-dimensional Lagrangian data assimilation problem in Section \ref{sec:4.2} (see Figure \ref{fig:LDA_Ice_Floes_Fig_2}), a similar analysis for this dyad-interaction model is also carried out. It illustrates the computational storage advantages of the adaptive-lag online smoother compared to its fixed-lag variant, for a model defined by intermittent instabilities rather than high dimensionality. See Figure \ref{fig:Dyad_Interaction_Fig_Appendix} in Appendix \ref{sec:app3}.

\subsection{Lagrangian Data Assimilation of a High-Dimensional Flow Field} \label{sec:4.2}

Lagrangian data assimilation (LDA) is a specialized method that utilizes the trajectories of moving tracers (e.g., drifters or floaters) as observations used to recover the unknown flow field driving their motion \cite{chen2014information, chen2015noisy, apte2008bayesian, ide2002lagrangian, griffa2007lagrangian}. Unlike Eulerian observations, which are fixed at specific locations, Lagrangian drifters track the movement of fluid parcels over time \cite{castellari2001prediction, salman2008using, honnorat2009lagrangian}. This approach is particularly valuable for recovering ocean states in the mid-latitude using floats or in marginal ice zones using sea ice floes as Lagrangian tracers \cite{chen2022efficient, chen2021lagrangian, covington2022bridging}. However, due to the high dimensionality of the state space of the flow field, which is characterized by a large number of spectral modes or high-resolution mesh grids, running a forward-backward offline smoother for state estimation becomes computationally infeasible. Consequently, the adaptive online smoother is essential in practical applications.

The coupled tracer-flow model is given as follows,
\begin{align}
        \d \vx_{\ell}(t)&=\mathbf{v}_{\ell}(t)\d t+\ms^{\vx_{\ell}}\d \vw_{\vx_\ell}(t), \label{eq:lthfloelocation}\\
        \d \mathbf{v}_{\ell}(t)&=\beta\big(\mathbf{u}(t,\vx_{\ell})-\mathbf{v}_{\ell}(t)\big)\d t+\ms^{\mathbf{v}_{\ell}}\d \vw_{\mathbf{v}_\ell}(t), \label{eq:lthfloevelocity}\\
        \d \hat{u}_{\mathbf{k}}(t)&=\big(-d_{\mathbf{k}}\hat{u}_{\mathbf{k}}(t)+f_{\mathbf{k}}(t)\big)\d t+\sigma_{\mathbf{k}}\d W_{\mathbf{k}}(t),\quad \text{with~} \mathbf{u}(t,\vx_{\ell}(t))=\sum_{\mathbf{k}\in \mathscr{K}} \hat{u}_{\mathbf{k}}(t)e^{i\mathbf{k}\cdot\vx_{\ell}(t)}\mathbf{r}_{\mathbf{k}}, \label{eq:flowfieldspectral}
\end{align}
where $\vx_{\ell}$ and $\mathbf{v}_{\ell}$ denote the location and velocity of the $\ell$-th tracer, respectively, with $\ell=1,\ldots,L$, while $\mathbf{u}(t,\vx)$ is the ocean velocity field expressed through a spectral decomposition in its geostrophically balanced or potential vortical incompressible modes, using a Fourier basis in the doubly-periodic spatial domain $[-\pi,\pi]$ \cite{majda2003introduction}. The governing equations of the Fourier coefficients $\hat{u}_{\mathbf{k}}$ are expressed through a set of linear stochastic processes, where $\mathbf{k}\in\mathscr{K}$ denotes the Fourier wavenumber with $\mathscr{K}=[-K,K]^2\cap\mathbb{Z}^2$ for $K\geq 1$ denoting the two-dimensional discrete lattice collection of Fourier wavenumbers. Note that the stochastic noise mimics the turbulent effects of nonlinearity present in many practical systems, such as the quasi-geostrophic (QG) ocean model. This approximation is justified in many applications \cite{chen2014information, chen2015noisy}. Since the primary goal is to illustrate the effectiveness of the online smoother, this simplification is employed in the study here. By treating all the $\vx_{\ell}$ as the observational variables and all the $\mathbf{v}_{\ell}$ and $\hat{u}_{\mathbf{k}}$ as the variables for state estimation, the LDA system in \eqref{eq:lthfloelocation}--\eqref{eq:flowfieldspectral} belongs to the CGNS family. The parameter values in this simulation are the following. The wavenumber bound is $K=2$, such that there are in total 25 modes after excluding the zeroth mode, which corresponds to a time-varying
random background mean sweep, or background velocity field. The total simulation time is $T=5$ with the numerical integration time step (also the observational time frequency) being $\dt=0.005$. The total number of tracers is $L=18$. The model parameters are $f_{\mathbf{k}}(t)=0.15e^{\displaystyle 5\pi it}$, $\ms^{\vx_{\ell}}=0.005\pi\mathbf{I}_{2\times 2}$, $\beta=1$, $\ms^{\mathbf{v}_\ell}=0.1\mathbf{I}_{2\times 2}$. The damping and noise coefficients $d_{\mathbf{k}}=d_{-\mathbf{k}}$ and $\sigma_{\mathbf{k}}=\sigma_{-\mathbf{k}}$ are randomly drawn from $\mathcal{U}\big([0.5,1.5]\big)$ and $\mathcal{U}\big([0.15,0.25]\big)$, respectively, with the conjugacy condition establishing the reality of the underlying velocity field \cite{chen2014information, chen2015noisy}.

Figure \ref{fig:LDA_Ice_Floes_Fig_1} demonstrates the posterior state estimates using the filter (Theorem \ref{thm:filtering}) and smoother (Theorem \ref{thm:smoothing}) in recovering the underlying ocean flow field and the tracer velocity vectors. Note that the amplitudes for the ocean flow field (in blue) and tracer velocity (in the hot colormap) quiver plots are not shown on the same scale. Panels (a)–-(c) compare the spatial recovery of the ocean state and tracer velocity vectors at $t=2$. Both methods effectively reconstruct the tracer movement in direction and magnitude, with the smoother being slightly more effective. As for the flow velocity field, the smoother exhibits a much higher skill in its recovery compared to the filter solution. This is further illustrated in Panels (d)-–(e), where we compare the true time series with the filter and smoother posterior mean states over the observational period for the real part of the ocean Fourier mode $\mathbf{k}=(2,-1)^\tran$ of the ocean, and the zonal velocity of tracer \#1 (noted in Panels (a)--(c)). We can see that the smoother follows the truth more closely and with less uncertainty for both cases. This is observed uniformly among all Fourier modes and all tracer velocity vectors.

\begin{figure}[!ht]
\centering
    \includegraphics[width=\textwidth]{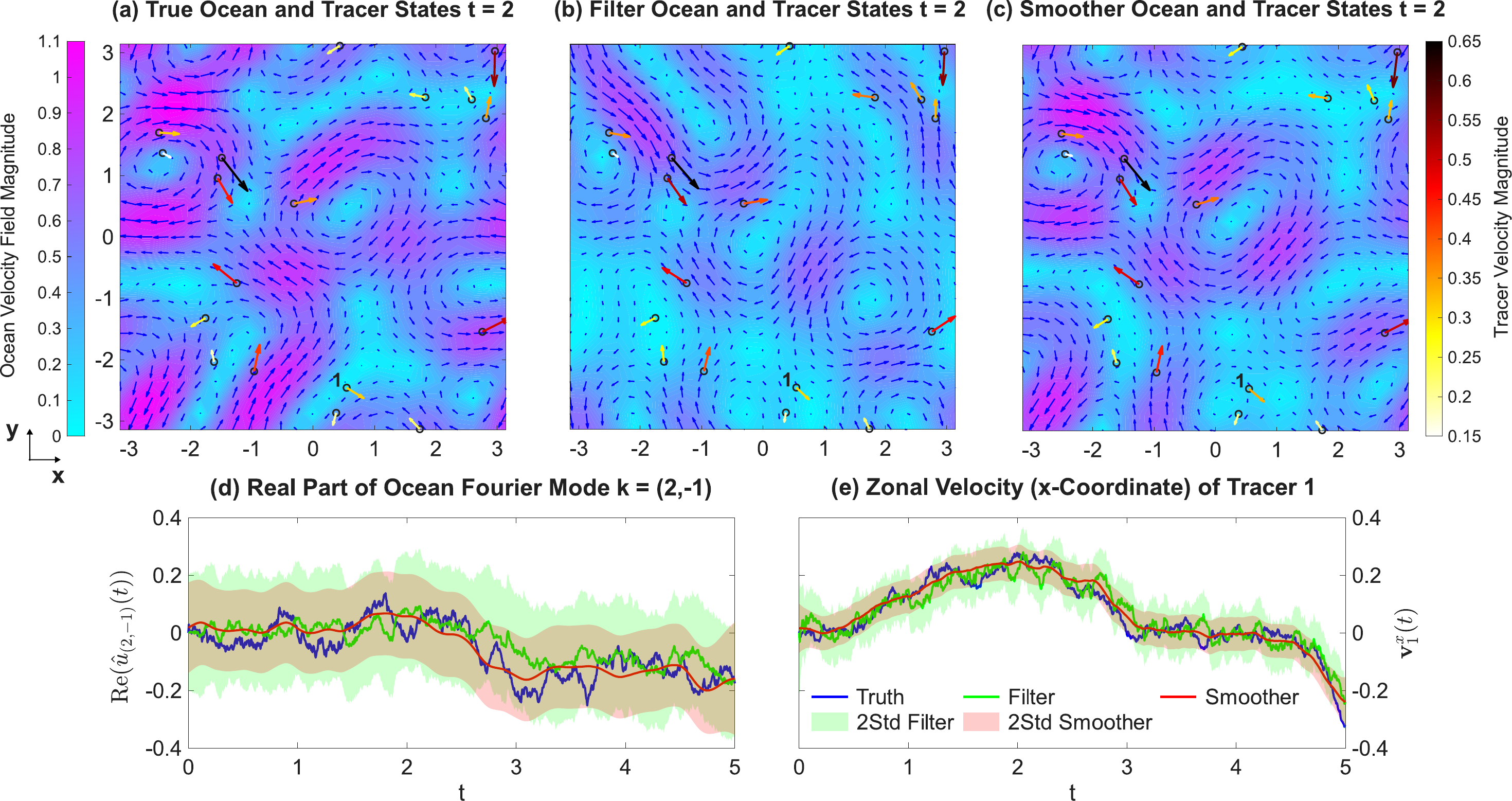}
    \caption{Panel (a): True state of the underlying flow field at $t=2$, where the colormap shows the amplitude and the quiver plot represents the velocity field. It is superimposed by the location of the tracers and their velocity vectors. The magnitudes of the quiver plot for the ocean's velocity field and velocity vectors for the tracers are on distinct scales. Panels (b)--(c): Similar to (a), but corresponding to the estimated state from the posterior mean, calculated through the filter and smoother, respectively. Panels (d)--(e): Comparison between the true time series (in blue) and the posterior mean time series of the filter (in green) and smoother (in red) solutions for the real part of the ocean Fourier mode $\mathbf{k}=(2,-1)^\tran$ and zonal velocity of tracer \#1 (labeled in Panels (a)--(c)).}
    \label{fig:LDA_Ice_Floes_Fig_1}
\end{figure}

Figure \ref{fig:LDA_Ice_Floes_Fig_2} compares the computational time (in seconds) and storage requirements (in gigabytes) for the fixed- and adaptive-lag strategies of the online smoother. (For details on how the storage and time values in Panels (a), (b), (d), and (e) are calculated, see Appendix \ref{sec:app4}.) Additionally, we present the normalized root-mean-square error (NRMSE) between the posterior mean time series and the truth as to assess the pathwise accuracy of these methods. The normalized RMSE is computed by dividing the RMSE (calculated in the temporal direction) by the standard deviation of the true signal \cite{lermusiaux1999data, hendon2009prospects}, such that the NRMSE indicates an error level comparable to the equilibrium variability. This metric is evaluated as a function of the fixed lag and tolerance parameter $\delta$ and showed on a logarithmically scaled x-axis, where for the latter a maximum lag bound of $b=300$ (corresponding to $b\dt=1.5$ simulation time units) is applied throughout, with the sequence of entropies used to define the adaptive lags (i.e., \eqref{eq:adaptlaginfodeforiginalseq} and \eqref{eq:adaptlaginfodef2}). Panels (a) and (b) display the fixed-lag results, while Panels (d) and (e) display the adaptive-lag ones. To provide robust results, computational times in Panels (b) and (e) are averaged over multiple runs to reduce external fluctuations. For the fixed-lag smoother, both storage and time increase algebraically (e.g., linearly) with lag, while the NRMSE converges exponentially, at a significant rate, from the filter solution (zero lag) to the one offline smoother one (maximum lag), indicating the sufficiency of using an overall short lag to significantly save on computational storage for this specific problem instance. On the other hand, the storage needs of the adaptive-lag smoother remain constant with respect to $\delta$, since a fixed $b$ is being used in all cases, while computational time increases algebraically to logarithmically as $\delta$ approaches zero. The former result showcases how the adaptive-lag online smoother implicitly minimizes storage requirements when compared to the fixed-lag variant, since for $b$ fixed $\delta$ can decrease freely to improve the recovery skill, with storage needs remaining unaffected. Note how by using a $b\dt$ value close to the uniform impact region of the observations in this LDA problem for a target pathwise recovery skill (e.g., $b\dt\approx0.35$ for a target NRMSE of around $0.35$, where we have diminishing returns with decreasing $\delta$; see Panel (a) in Figure \ref{fig:LDA_Ice_Floes_Fig_3}), there is a potential for significant storage savings to be had in this approach. Of course, the computational time spent unavoidably increases with decreasing $\delta$, since more linear-search checks in \eqref{eq:adaptlaginfodeforiginalseq} (or \eqref{eq:adaptlaginfodef1}) need to be carried out. In Panel (e) we also note at each data point on the time plot the percentage of time spent on the calculation of the adaptive lag (also averaged over many runs); in this calculation we include the time spent evaluating the update matrices $\mathbf{D}^{j,n}$ for $n$ and $R_n\leq j\leq n-1$, since they are required for the calculation of $L_n$ (see \eqref{eq:lackinfoupdate}--\eqref{eq:covarianceratio}). The NRMSE for the adaptive-lag method now decreases algebraically with $\delta$, demonstrating a sharper trade-off between computational resources and error compared to the fixed-lag approach. 

Panels (c) and (f) further illustrate the recovery skill of these approaches through the temporal average of the information gain for the fixed- and adaptive-lag online smoother strategies beyond the filter solution as functions of the fixed lag and tolerance parameter, respectively. Both are plotted on an x-axis that is scaled logarithmically and are calculated using the signal--dispersion formula of the relative entropy in \eqref{eq:signaldispersion}, since both distributions are Gaussian at each time instant. For reference, we show both the signal and dispersion parts, as well as the total information gain. For the fixed-lag smoother, since the zero-fixed-lag online smoother coincides with the filter solution, the information gain there is zero, but as the fixed lag increases, it rapidly and exponentially converges to the information gain value corresponding to the full-backward offline smoother beyond the filter estimates (see Section \ref{sec:3.3}). On the other hand, similarly to the results from Panels (d) and (e), we have that the temporally-averaged information gain beyond the filter solution is a logarithmic to algebraic function in $\delta$, where for large values of $\delta$ this is close to zero and approaches the information gain value of the full-backward offline smoother beyond the filter estimates as $\delta$ approaches zero. In both cases, these behaviors are demonstrated in both the signal and dispersion parts and, as such, in the total information gain as well.

\begin{figure}[!ht]
\centering
    \includegraphics[width=\textwidth]{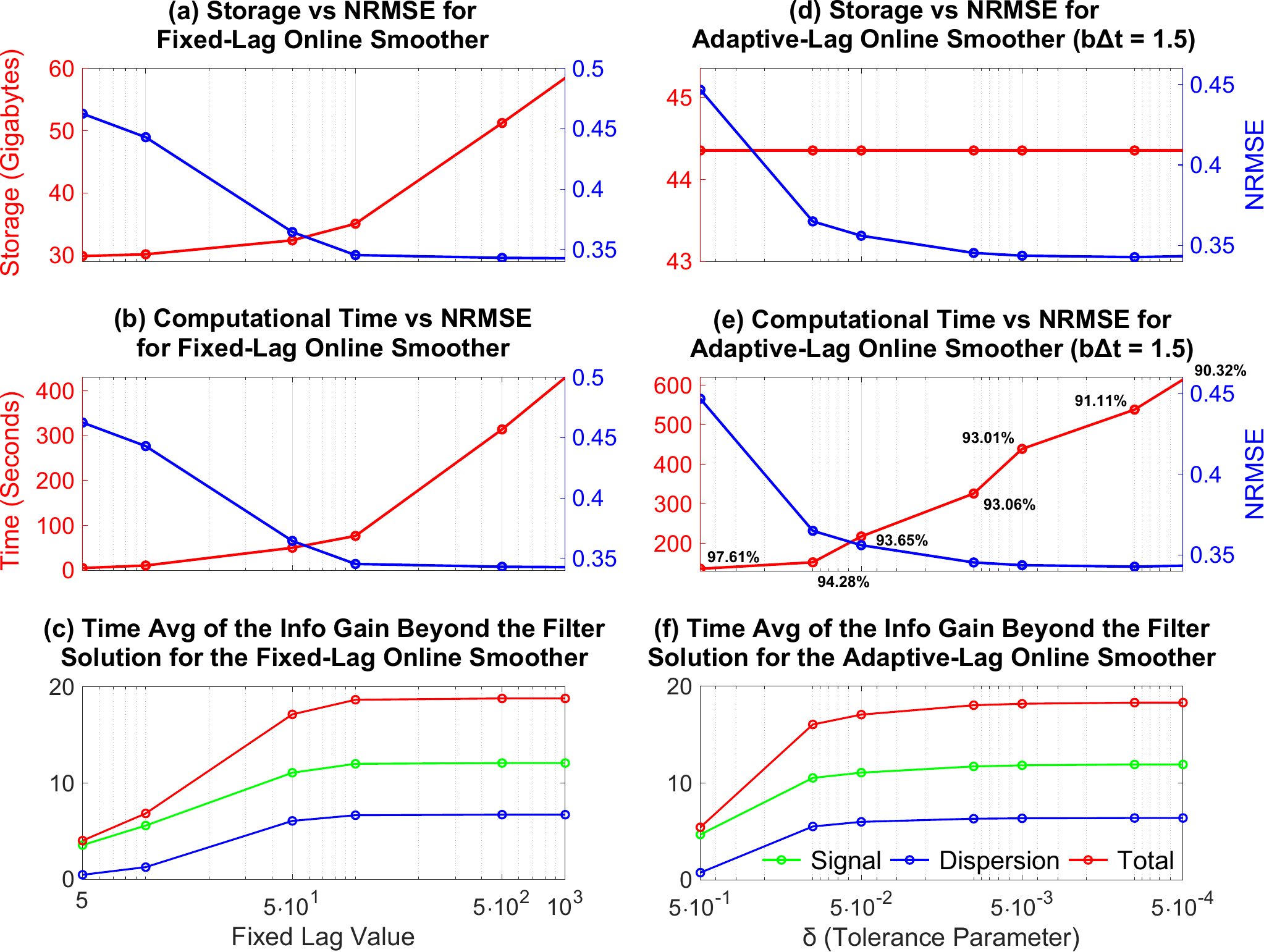}
    \caption{Panels (a)--(b) \& (d)--(e): Analysis of the trade-off between the computational complexity, measured in computational time (in seconds), and required storage capacity (in gigabytes), and lower-order pathwise error statistic of the NRMSE between the posterior mean time series and true signal, for the fixed-lag (Panels (a)--(b)) and adaptive-lag online smoother (Panels (d)--(e)). In Panel (e) we also show at each data point of the latter the percentage of time spent on the adaptive-lag calculation (including the calculation of the update matrices $\mathbf{D}^{j,n-2}$). Details on how the storage and time values are calculated are given in Appendix \ref{sec:app4}. Panels (c) \& (f): Temporal average of the information gain for the fixed- (Panel (c)) and adaptive-lag (Panel (f)) online smoother beyond the filter solution for various values of the fixed lag and tolerance parameter, respectively. Both the signal (in green) and dispersion (in blue) parts of the total information gain (in red) are shown. All panels are plotted on a semi-log x-axis, and the x-axes for the adaptive-lag online smoother  results are reversed for ease of comparison with the corresponding fixed-lag strategy plots.}
    \label{fig:LDA_Ice_Floes_Fig_2}
\end{figure}

Finally, in Figure \ref{fig:LDA_Ice_Floes_Fig_3}, we demonstrate the contrasting behavior of the adaptive-lag online smoother in the LDA problem when compared to the dyad-interaction model with intermittent instabilities from Section \ref{sec:4.1}. 
First, Panel (a), similar to Panel (g) in Figure \ref{fig:Dyad_Interaction_Fig_2}, illustrates the adaptive lag values (in time units) of the online smoother with hyperparameters $b=100$ (corresponding to $b\dt=0.5$ simulation time units), $\delta=0.05$ and calculated via the sequence of entropies through \eqref{eq:adaptlaginfodeforiginalseq} and \eqref{eq:adaptlaginfodef2}, as a function of time in $[2,5]$. Unlike the adaptive lag values generated from the dyad model in Section \ref{sec:4.1}, which highly depend on whether an extreme event is being observed or if we are currently at a period of dormancy, the adaptive lag value here meanders around a constant (about $0.3$ in simulation time units), via uniform, in amplitude, oscillations, and does so steadily in time which indicates the absence of intermittency in the system and the uniform impact each observation has on the online smoother solution. Similar to Panels (h) and (i) of Figure \ref{fig:Dyad_Interaction_Fig_2}, Panels (b) and (c) depict the standardized information gain sequence (see \eqref{eq:infogainstandardized}) and spectral radii of the update matrices $\mathbf{D}^{j,n-2}$ but for the LDA problem defined in this section, \eqref{eq:lthfloelocation}--\eqref{eq:flowfieldspectral}. Again, as with the results from the dyad-interaction model, the two metrics showcase significant temporal correlations, with their major discrepancies emerging in the differing scale and noisier behavior of the former. As before, this functions as numerical corroboration to the claim that our proposed information-theoretic approach to determining the observational impact region of the $n$-th observation onto the online discrete smoother update is dynamically- and statistically-consistent and effective in recovering the exact approach reflected by the spectrum of the update tensor $\mathbf{D}^{j,n-2}$, while being much more computationally efficient at the same time. But, it is also important to compare and contrast these results with those from Panels (g)--(i) of Figure \ref{fig:Dyad_Interaction_Fig_2}. Foremost, the spectral radii of the update tensors $\mathbf{D}^{j,n-2}$, which are high-dimensional matrices in this case, remain below $1$ throughout time and for all values of $n$. Also, in both case studies, these sequences illustrate an exponential behavior in $j$ for all values of $n$. However, for the LDA problem, we can see a uniform rate of exponential decay as $j$ decreases from $n-1$, across all observations (i.e., $n$). Recall that, in the dyad-interaction numerical experiment, this rate would change depending on whether a rare or extreme event was forming in the observable signal; it could even exhibit regime switches depending on the emergence of intermittent instabilities. Therefore, for the LDA problem where we recover the flow field's and tracers' velocities, the observations from the tracers' locations have a uniform impact on the online smoother estimate over time, which is a result of the stability that the model in \eqref{eq:lthfloelocation}--\eqref{eq:flowfieldspectral} has, as well as of the fact that the equilibrium distribution of the tracers is nearly uniform due to the incompressibility of the underlying ocean velocity field \cite{chen2014information}. This contrasts with the dyad-interaction model, which displays intermittent instabilities due to the latent dynamics acting as anti-damping in the observable ones on rare occasions. 

Lastly, in Panel (d) of Figure \ref{fig:LDA_Ice_Floes_Fig_3}, we showcase the standardized LSDev criterion, i.e., for $\sigma^{j,n}$ in \eqref{eq:infogainstandardized}, in the same manner as Panel (b). In this case study, where the absence of the intermittent instabilities allows for an almost consistent exponentially decaying behavior with a constant rate in $j$ for each $n$ for the adaptive-lag-defining criteria, Panels (b) and (d) numerically corroborate the note made in the second-to-last paragraph in Section \ref{sec:3.4.2}; we can yield analogous adaptive lag values with either sequence (of relative entropies $\left\{ \mathcal{P}\left(p^{j,n}_{\text{updated}}, p^{j,n}_{\text{lagged}}\right)\right\}_{R_n\leq j\leq n-1}$ or of their LSDevs $\{\sigma^{j,n}\}_{R_n\leq j\leq n-1}$) by a simple adjustment of the $\delta$ tolerance value.

\begin{figure}[!ht]
\centering
    \includegraphics[width=\textwidth]{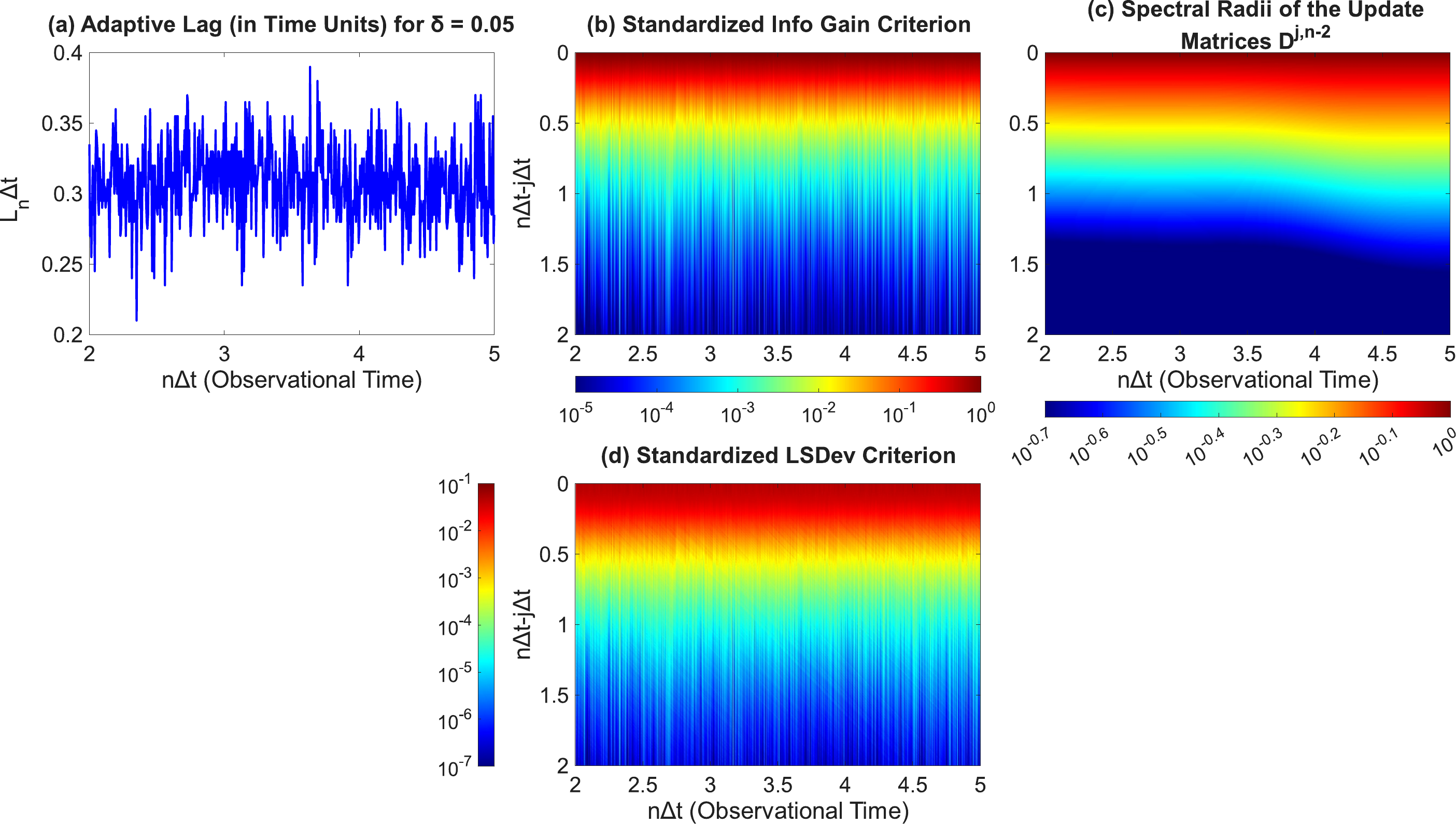}
    \caption{Panels (a)--(c): Similar to Panels (g)--(i) of Figure \ref{fig:Dyad_Interaction_Fig_2}, respectively. Panel (d): Same as Panel (b) but for the standardized LSDev sequence (i.e., for $\sigma^{j,n}$ in \eqref{eq:infogainstandardized}). All panels correspond to the LDA model in \eqref{eq:lthfloelocation}--\eqref{eq:flowfieldspectral}.}
    \label{fig:LDA_Ice_Floes_Fig_3}
\end{figure}

\subsection{Online Parameter Estimation of Nonlinear Systems with Only Partial Observations} \label{sec:4.3}

Parameter estimation and model identification for complex nonlinear dynamical systems are crucial for effective state estimation, DA, and forecasting. In partially observed nonlinear systems, parameter and state estimation typically occur simultaneously, requiring a solution that often necessitates complicated approximations and expensive numerical methods. Nevertheless, the closed analytic formulae of the online nonlinear smoother \eqref{eq:onlinerecursive1}--\eqref{eq:onlinerecursive2} facilitate an efficient online parameter estimation scheme for the CGNS \eqref{eq:condgauss1}--\eqref{eq:condgauss2} for when only the time series of $\mathbf{x}$ is observed. In this section, we incorporate the adaptive online smoother into an expectation-maximization (EM) algorithm \cite{sundberg1976iterative, ghahramani1996parameter, ghahramani1998learning, moon1996expectation} to alternately estimate parameters and hidden states, employing closed analytic formulae for both estimations.

\subsubsection{The Online EM Parameter Estimation Algorithm} \label{sec:4.3.1}

The basic EM algorithm is outlined in Appendix \ref{sec:app5}, and the mathematical details for the CGNS framework can be found in \cite{chen2020learning}. This reference serves as the foundation for developing an online version that utilizes the adaptive-lag online smoother for simultaneous state estimation and model identification, and also includes the tools necessary for incorporating into this algorithm's skeleton extensions like block decomposition for learning high-dimensional systems, sparse identification or regularization, and physical and other constraints on the parameters.

The parameter estimation aims at seeking an optimal estimation of the unknown parameters $\boldsymbol{\theta}$ by maximizing the log-likelihood function corresponding to the partially observed dynamics, also known as the marginal, partial, or conditional log-likelihood. Essentially, our goal is for model identification under the assumption that only a time series from the observable components $\vx$ can be obtained, with $\vy$ corresponding to the unobserved variables, which nevertheless interacts with $\vx$ in \eqref{eq:condgauss1}--\eqref{eq:condgauss2}. Since only the time series of $\vx$ is observed, the maximum likelihood estimator (MLE) for the marginal log-likelihood is given by
\begin{equation} \label{eq:mledef}
    \hat{\boldsymbol{\theta}}_{\text{MLE}}:=\underset{\boldsymbol{\theta}\in\Theta}{\mathrm{argmax}} \left\{\log\big(p(\vx;\boldsymbol{\theta})\big)\right\}=\underset{\boldsymbol{\theta}\in\Theta}{\mathrm{argmax}} \left\{\log\left(\int_{\vy}p\big(\vx,\vy;\boldsymbol{\theta}\big)\d\vy\right)\right\},
\end{equation}
where $\Theta$ denotes the underlying parameter space, which we assume to be a closed and convex subset of $\mathbb{C}^N$ for $N\in\mathbb{N}$, and the integration in \eqref{eq:mledef}, i.e., marginalization, is happening over the state space of $\vy$. The EM iteration alternates between performing the so-called expectation step (E-step), which estimates the hidden state of $\vy$ using the current estimate for the parameters, and a maximization step (M-step), which computes the parameters that maximize the expected marginal log-likelihood found in the E-step \cite{ghahramani1996parameter, dembo1986parameter, kokkala2014expectation}. Since $\vy$ is unobserved throughout time, the conditional expectation in the E-step needs to be taken for those measurable functions (measurable with respect to the current observable $\sigma$-algebra) containing $\vy$ when computing the log-likelihood estimate, which takes into account the uncertainty in the state estimation through the optimal a-posteriori statistics of $\vy$. As is shown in Appendix \ref{sec:app5}, the solution from the E-step is exactly the online smoother posterior distribution.

\subsubsection{A Numerical Example} \label{sec:4.3.2}

The model we study in this numerical simulation is as follows:
\begin{align}
    \d u(t)&=\big(-d_uu(t)+\gamma u(t)v(t)+F_u\big)\d t+\sigma_u\d W_u(t), \label{eq:emdyadmodel1}\\
    \d v(t)&=\big(-d_vv(t)-\gamma u(t)^2+F_v\big)\d t+\sigma_v\d W_v(t),\label{eq:emdyadmodel2}
\end{align}
where $d_u,d_v,\sigma_u,\sigma_v>0$ and $\gamma,F_u,F_v\in\rr$, with $W_u$ and $W_v$ being two mutually independent Wiener processes. The system \eqref{eq:emdyadmodel1}--\eqref{eq:emdyadmodel2} fits the CGNS framework for $\vx:=u$ and $\vy=v$. Similar to the model in \eqref{eq:dyad1}--\eqref{eq:dyad2} of Section \ref{sec:4.1}, the variable $v$ acts as a stochastic damping in the observable process, which leads to intermittent extreme events. The effect of these instabilities in the convergence of the algorithm is especially studied. We assume a time series of $u$ is continuously observed while there is no observation of $v$. To simplify the study in this section, we assume that the noise feedbacks $\sigma_u$ and $\sigma_v$ are known, although our framework does not require such a restriction \cite{chen2020learning}. Thus, the parameter vector for this model is as follows:
\begin{equation} \label{eq:parameterdyad}
    \boldsymbol{\theta}:=(d_u,\gamma,F_u,d_v,F_v)^\tran.
\end{equation}
The true parameter values are the following:
\begin{gather*}
    d_u = 1, \quad \gamma=3, \quad F_u = 1, \quad \sigma_u=0.5, \\
    d_v = 1, \quad F_v=0.2, \quad \sigma_v=1.
\end{gather*}

The numerical setup is as follows. The initial guess of the parameter values is $\boldsymbol{\theta}_0=(2,6,2,0.5,0.6)^\tran$. The numerical integration time step, which is also the observational frequency in such a continuous observational case, is $\dt=0.001$. The total length of observations is $T=200$ time units. Since implementing the online smoother requires a certain history of previous state estimation solutions, the observed time series of $u$ within a short initial window $[0,T_{\text{ini}}]$, with $T_{\text{ini}}=10$, is utilized to initially carry out the filtering solution and online smoother estimates, based on $\boldsymbol{\theta}_0$. This can be treated as a burn-in or learning period for the algorithm, and also aids in avoiding the possibility of short-time blowup. Afterwards, with the arrival of each new observation, we update the adaptive-lag online smoother state estimations via the procedure in \eqref{eq:adaptivelagmean}--\eqref{eq:adaptivelagcov} and then re-estimate the parameters using the MLE. The update happens at every observation to ensure stability. The hyperparameters in the adaptive-lag method are $b=1000$ (where $b\dt=1$ time unit) and $\delta=10^{-4}$, where we use the actual sequence of information gains to define the adaptive lag at each new observation, i.e., \eqref{eq:adaptlaginfodeforiginalseq}. For a brief discussion on the online EM parameter estimation algorithm's stability, sensitivity, and convergence skill with respect to the initial parameter values and length of the burn-in or learning period (for this specific model), through a purely empirical and simplified analysis, the interested reader is referred to Appendix \ref{sec:app6}.

Figure \ref{fig:EM_Parameter_Estimation_Fig} presents the results from using the online EM algorithm. Panels (a) and (b) display the true trajectories of the observed and unobserved variables, $u$ and $v$, respectively, displayed in blue. These trajectories are generated by running equations \eqref{eq:emdyadmodel1}--\eqref{eq:emdyadmodel2} forward in time. Alongside the true trajectories, the time series obtained by simulating the model with the estimated parameters at the final iteration of the online EM algorithm are included in red. Note that the same values for the Wiener processes are used for these two signals, allowing us to have a point-wise comparison. The recovery of the signal for both observable and unobserved variables is notably skillful. It is worth highlighting that, in Panel (b), the dashed horizontal lines indicate the true (in black) and online EM estimated (in green) anti-damping threshold lines, above which the unobserved variable effectively acts as anti-damping in the observable process, initiating intermittent instabilities in the signal. Their consistency is critical for retrieving the true dynamics using the estimated parameters. Panel (c) shows the adaptive lag values, measured in time units (i.e., $L_n\dt$) for the online smoother, illustrating that longer lags are typically required during the onset of the intermittent phases to effectively capture the dynamics. The efficacious recovery of the true state is also demonstrated through higher-order statistics, such as the PDFs and autocorrelation functions (ACFs) for both $u$ (in Panels (d)--(e)) and $v$ (in Panels (f)--(g)). The use of an upper lag bound of $b\Delta t=1$ time unit is enough to capture most of the information in the online smoother updates while not overspending on storage (for the chosen tolerance value of $\delta=10^{-4}$). This is showcased in the produced PDFs and ACFs, which indicate faithful model fidelity and memory recovery, respectively, where for the latter, this is true even for time-lags outside the interval of the produced adaptive lag values. Finally, Panels (h)--(m) present the trace plots produced by the online EM iterations for various parameters of interest.  The parameters in the observed process $u$ are nearly perfectly recovered, including the nonlinear feedback $\gamma$ and the deterministic constant forcing $F_v$ in the process of the unobserved variable. The only parameter that is not perfectly recovered is the damping $d_v$ in the process of the unobserved variable, with the algorithm converging to twice its true value. This is expected, as these parameters, along with the noise feedback, are typically the hardest to estimate due to observability issues, where their contribution to the observed process is relatively weak. Nevertheless, this inaccuracy does not hinder the model with the estimated parameters from reproducing key dynamical and statistical features, as shown in Panels (a)--(b) and (d)--(g).

An important observation from the trace plots is that extreme events in the observable signal enhance the convergence skill of the online EM algorithm within the CGNS framework. In Panel (a), four extreme events are marked as A, B, C, and D, characterized by a high signal-to-noise ratio. The trace plots in Panels (h)--(m) demonstrate that extreme events influence parameter convergence. The high signal-to-noise ratio mitigates the uncertainty in observations, resulting in strong observability of the system, which naturally facilitates parameter estimation \cite{chen2023causality}. This is particularly evident in the anti-damping threshold (Panel (m)), highlighting the role of these events in identifying the intermittent nature of the underlying dynamics.

\begin{figure}[!ht]
\centering
    \includegraphics[width=\textwidth]{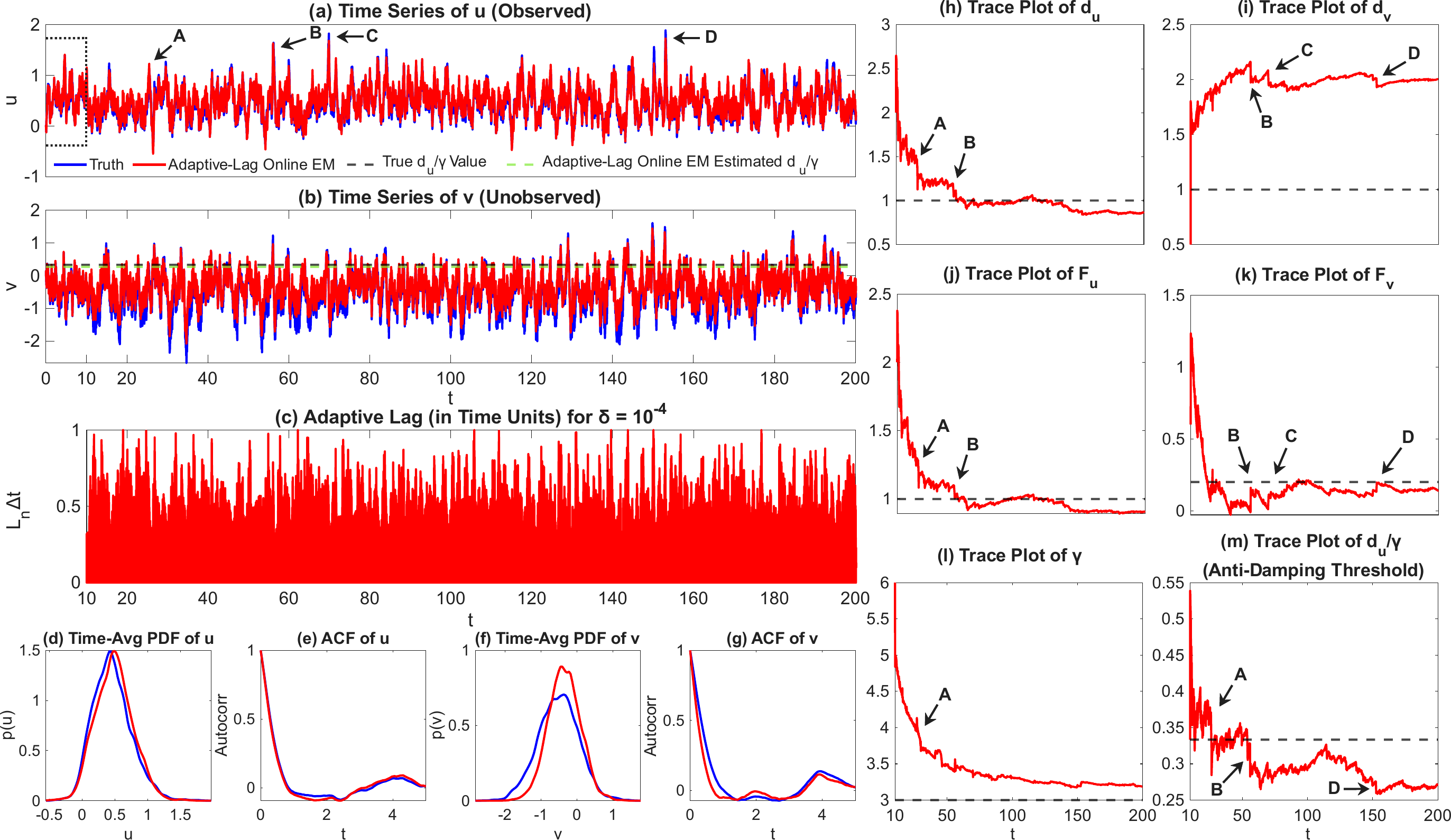}
    \caption{Panel (a): True trajectory of the observable variable $u$ over the simulation period (in blue), and the recovered trajectory from using the same white noise values used to generate the true signal (by utilizing the same seed number) but where the model in \eqref{eq:emdyadmodel1}--\eqref{eq:emdyadmodel2} is instead defined through the recovered parameter values at the last iteration of the online EM algorithm (in red). The dashed rectangular box denotes the period of running the filter and online smoother algorithms with $\boldsymbol\theta_0$, which is $[0,10]$. Four extreme events are marked in the observed time series (denoted by A, B, C, and D). Panel (b): Same as (a) but for the unobserved variable $v$. The true and online EM estimated anti-damping threshold $d_u/\gamma$ are marked by dashed lines, in black and green, respectively. Panel (c): Adaptive lag values generated using \eqref{eq:adaptlaginfodeforiginalseq} (measured in time units, i.e., $L_n\dt$), as a function of simulation time. Panels (d)--(g): PDFs of $u$ (Panel (d)) and $v$ (Panel (f)) and their ACFs (Panels (e) and (g), respectively), for both the true model parameters and the ones obtained at the last iteration of the online EM algorithm. Panels (h)--(m): Trace plots produced by the online EM algorithm.}
    \label{fig:EM_Parameter_Estimation_Fig}
\end{figure}

For the sake of completeness, in Figure \ref{fig:EM_Parameter_Estimation_Fig_Appendix}, we provide the results from instead using the fixed-lag online discrete smoother in the outlined online EM parameter estimation algorithm, when applied to this specific model identification problem. See Appendix \ref{sec:app7} for more info and a discussion comparing the fixed- and adaptive-lag online EM parameter estimation results.

\section{Conclusions} \label{sec:5}

In this paper, a computationally efficient algorithm for online adaptive-lag optimal smoother-based state estimation with partial observations is developed. Closed analytical formulae are available for this online smoother. It applies to CGNS, representing a rich class of CTNDSs with wide applications in neuroscience, ecology, atmospheric science, geophysics, and many other fields. Importantly, the adaptive-lag strategy allows the reduction of computational complexity and storage requirements, creating the potential for the algorithm to significantly outperform both offline smoothing and its fixed-lag counterpart operationally while achieving a comparable skill in state estimation. Notably, an information-theoretic criterion, exploiting uncertainty reduction across the entire posterior distribution, is developed to calculate the adaptive lag value, which helps in preserving dynamical and statistical consistency.

The adaptive-lag online smoother has been applied to a system with intermittent instabilities (Section \ref{sec:4.1}). It effectively recovers extreme events and discovers the causal dependence between different state variables. It is also applicable to high-dimensional systems, such as the LDA problem. The state estimation based on such an online smoother shows nearly the same accuracy as the standard forward-backward smoother while requiring significantly less computational storage under suitable choices of the algorithm parameters (Section \ref{sec:4.2}). In addition, the online adaptive-lag smoother facilitates real-time parameter estimation. The study highlights the importance of utilizing observed extreme events to accelerate the convergence of parameter estimation (Section \ref{sec:4.3}).

Potential future research topics in this direction are as follows. First, the closed-form expressions of the online smoother estimates can aid in assessing extreme events and their impact on the CGNS dynamics, as well as in identifying the most sensitive directions that exhibit rapid fluctuations during specific events using information-theoretic methods \cite{majda2010quantifying}. This is facilitated by the fact that the CGNS framework enjoys analytic formulae for such conditional distributions \cite{liptser2001statisticsII}. Second, the adaptive-lag online smoother can serve as a useful tool for causal inference and information flow analysis. On the one hand, one can exploit intermittent extreme events to improve state estimation, which helps reduce the model error in developing parsimonious surrogate models via causal inference. On the other hand, the conditional Gaussianity of the framework allows the application of transfer or causation entropy to explore these topics in future studies. Third, it is necessary to explicitly investigate how the adaptive-lag online smoother has the potential to outperform the fixed-lag alternative by uniformly reducing storage needs and computational complexity in CGNSs. Such a study will involve the derivation of  theoretical or empirical representations of these operational metrics, as well as of point-wise or distribution-based error measures during state estimation, as functions of the fixed lag and adaptive-lag tolerance parameter $\delta$ (for a fixed upper lag bound $b$). Finally, future work along this line also involves developing a rigorous investigation of the spectral properties of the update tensors $\mathbf{D}^{j,n-2}$ used in the online smoother, as well as of the matrices $\mathbf{E}^j$ that define them, potentially through suitable concentration inequalities. This includes studying the temporal behavior of the information-theoretic criterion for adaptive lags, $\mathcal{P}\left(p^{j,n}_{\text{updated}}, p^{j,n}_{\text{lagged}}\right)$ (and its LSDev). Numerical case studies have shown that these quantities showcase exponentially varying behavior over time, with the rate depending on whether a quiescent or intermittent signal is observed. Notably, the spectrum of $\mathbf{D}^{j,n-2}$ remains below the convergence threshold of $1$, which is reassuring. A rigorous analysis of these properties is a natural future topic.

\section*{Acknowledgments}

The authors thank the reviewers for their constructive feedback and valuable suggestions during the peer-review process.\\

\noindent N.C. is grateful to acknowledge the support of the Office of Naval Research (ONR) N00014-24-1-2244 and Army Research Office (ARO) W911NF-23-1-0118. M.A. is supported as a research assistant under these grants.

\begin{appendices} \label{sec:app}

\numberwithin{equation}{section}
\counterwithin{figure}{section}

\section{Sufficient Assumptions for the Results of the CGNS Framework} \label{sec:app1}

Let the following denote the standard elements of the vector- and matrix-valued functionals appearing in \eqref{eq:condgauss1}--\eqref{eq:condgauss2} as model parameters:
\begin{gather*}
        \vf^\vx(t,\vx):=\big(f^\vx_1(t,\vx),\ldots,f^\vx_k(t,\vx)\big)^\mathtt{T},\quad\vf^\vy(t,\vx):=\big(f^\vy_1(t,\vx),\ldots,f^\vy_l(t,\vx)\big)^\mathtt{T},\\
        \ml^\vx(t,\vx):=\big(\Lambda^\vx_{ij}(t,\vx)\big)_{k\times l},\quad \ms_1^\vx(t,\vx):=\big(\Sigma^{\vx,1}_{ij}(t,\vx)\big)_{k\times d},\quad \ms_2^\vx(t,\vx):=\big(\Sigma^{\vx,2}_{ij}(t,\vx)\big)_{k\times r},\\
    \ml^\vy(t,\vx):=\big(\Lambda^\vy_{ij}(t,\vx)\big)_{l\times l},\quad\ms_1^\vy(t,\vx):=\big(\Sigma^{\vy,1}_{ij}(t,\vx)\big)_{l\times d},\quad\ms_2^\vy(t,\vx):=\big(\Sigma^{\vy,2}_{ij}(t,\vx)\big)_{l\times r}.
\end{gather*}
To be able to obtain the main results which define the CGNS framework and its potency in DA through closed-form expressions for the posterior statistics, a set of sufficient regularity conditions needs to be assumed a-priori. In what follows, each respective pair of indices $i$ and $j$ takes all admissible values and $\vz$ is a $k$-dimensional function in $C^0\big([0,T];\mathbb{C}^k\big)$:
\begin{itemize}
    \item[\textbf{\blue{(1)}}] We assume:
    \begin{align*}
        \int_0^T \Big[&\big|f^\vy_i(t,\vz)\big|+\big|\Lambda^\vy_{ij}(t,\vz)\big|+\big| f^\vx_i(t,\vz)\big|^2+\big|\Lambda^\vx_{ij}(t,\vz)\big|^2 \\
        & +\big|\Sigma^{\vx,1}_{ij}(t,\vz)\big|^2+\big|\Sigma^{\vx,2}_{ij}(t,\vz)\big|^2+\big|\Sigma^{\vy,1}_{ij}(t,\vz)\big|^2+\big|\Sigma^{\vy,2}_{ij}(t,\vz)\big|^2\Big]\rmd t<+\infty.
    \end{align*}
    This ensures the existence of the integrals in \eqref{eq:condgauss1}--\eqref{eq:condgauss2} \cite{liptser2001statisticsI, liptser2001statisticsII}.

    \item[\textbf{\blue{(2)}}] $\displaystyle \big|\Lambda^\vx_{ij}(t,\vz)\big|,\big|\Lambda^\vy_{ij}(t,\vz)\big|\leq L_1$ for some (uniform) $L_1>0$, $\forall t\in[0,T]$.

    \item[\textbf{\blue{(3)}}] If $g(t,\vz)$ denotes any element of the multiplicative factors in the drift dynamics, $\ml^\vx$ and $\ml^\vy$, or of the noise feedback matrices, $\ms_m^\vx(t,\vz)$ and $\ms_m^\vy(t,\vz)$ for $m=1,2$, and $K(s)$ is a nondecreasing right-continuous function taking values in $[0,1]$, then there $\exists L_2,L_3,L_4,L_5>0$ such that for any $\vz$ and $\mathbf{w}$ $k$-dimensional functions in $C^0\big([0,T];\mathbb{C}^k\big)$ we have:
    \begin{gather*}
        |g(t,\vz)-g(t,\mathbf{w})|^2\leq L_2\int_0^t \left\|\vz(s)-\mathbf{w}(s)\right\|_2^2\rmd K(s)+L_3\left\|\vzt-\mathbf{w}(t)\right\|_2^2, \ \forall t\in[0,T],\\
        |g(t,\vz)|^2\leq L_4\int_0^t \big(1+\left\|\vz(s)\right\|_2^2\big)\rmd K(s)+L_5\big(1+\left\|\vzt\right\|_2^2\big), \ \forall t\in[0,T].
    \end{gather*}
    These integrals are to be understood in the Lebesgue--Stieltjes integration sense.

    \item[\textbf{\blue{(4)}}] $\mathbb{E}\big[\|\vx(0)\|_2^{2}+\|\vy(0)\|_2^{2}\big]<+\infty$.

    \item[\textbf{\blue{(5)}}] The sum of the Gramians (with respect to the rows) of the noise coefficient matrices in the observable process are uniformly nonsingular, i.e., the elements of the inverse of
    \begin{equation*}
        (\ms^\vx\circ\ms^\vx)(t,\vz)=\ms_1^\vx(t,\vz)\ms_1^\vx(t,\vz)^\dagger+\ms_2^\vx(t,\vz)\ms_2^\vx(t,\vz)^\dagger,
    \end{equation*}
    are uniformly bounded in $[0,T]$.

    \item[\textbf{\blue{(6)}}] $\displaystyle \int_0^T \ee{\big|\Lambda^\vx_{ij}(t,\vzt)y_j(t)\big|}\rmd t< +\infty$.

    \item[\textbf{\blue{(7)}}] $\displaystyle \mathbb{E}\Big[\big|y_j(t)\big|\Big]<+\infty, \ t\in[0,T]$.

    \item[\textbf{\blue{(8)}}] For $\mu_{\text{f},j}(t):=\mathbb{E}\big[y_j(t)\big|\cF_t^\vx\big]$, where $t\in[0,T]$ and $j=1,\ldots,l$, we assume that:
    \begin{equation*}
        \pp\left(\int_0^T\big|\Lambda^\vx_{ij}(t,\vzt)\mu_{\text{f},j}(t)\big|^2\rmd t<+\infty\right)=1,
    \end{equation*}
    for all $i=1,\ldots, k$.
    
    \item[\textbf{\blue{(9)}}] For $\vz\in C^0\big([0,T];\mathbb{C}^k\big)$, $\mathbf{u}\in C^0\big([0,T];\mathbb{C}^l\big)$, and $\mathbf{f}:=\begin{pmatrix}
            \mathbf{0}_{k\times k} & \ml^\vx \\
            \mathbf{0}_{l\times k} & \ml^\vy
        \end{pmatrix}\begin{pmatrix}
            \vx\\\vy
        \end{pmatrix}+\begin{pmatrix}
            \vf^\vx\\\vf^\vy
        \end{pmatrix}$, $\ms:=\begin{pmatrix}
            \ms_1^\vx & \ms_2^\vx\\
            \ms_1^\vy & \ms_2^\vy
        \end{pmatrix}$, there exists $L_6>0$ such that
        \begin{equation*}
            \left\|\mathbf{f}(t,\vz,\mathbf{u})-\mathbf{f}(s,\vz,\mathbf{u})\right\|_2+\left\|\ms(t,\vz)-\ms(s,\vz)\right\|_2\leq L_6\left(1+\big(\left\|\vz\right\|^2_2+\left\|\mathbf{u}\right\|^2_2\big)^{1/2}\right)|t-s|^{1/2}, \ \forall t,s\in[0,T].
        \end{equation*}

    \item[\textbf{\blue{(10)}}] Let $\big(\vx_{\dt}(t), \vy_{\dt}(t)\big)$ be the stochastic discrete-time approximation of $(\vx,\vy)$, after a continuous-time extension, that is provided by the numerical integration Euler--Maruyama scheme when applied on the CGNS in \eqref{eq:condgauss1}--\eqref{eq:condgauss2} (see \eqref{eq:discretecondgauss1}--\eqref{eq:discretecondgauss2}), over a temporal uniform partition of $[0,T]$ with $\dt$ being the uniform subinterval length. Then, as in Andreou \& Chen \cite{andreou2024martingale}, we assume that there exists $L_7>0$ such that:
    \begin{equation*}
        \ee{\big\|\vx(0)-\vx_{\dt}(0)\big\|^2_2+\big\|\vy(0)-\vy_{\dt}(0)\big\|^2_2}\leq L_7\dt.
    \end{equation*}
    In other words, we assume that the exact initial distributions are not being severely misidentified by the discrete-time ones (at least in the mean).

    \item[\textbf{\blue{(11)}}] $\displaystyle \int_0^T \ee{\big|\ml_{ij}^\vx(t,\vx)\big|^4+\big|\vf_i^\vy(t,\vx)\big|^4+\big|\Sigma^{\vy,1}_{ij}(t,\vx)\big|^4+\big|\Sigma^{\vy,2}_{ij}(t,\vx)\big|^4}\rmd t<+\infty$.

    \item[\textbf{\blue{(12)}}] $\displaystyle \ee{\left\|\vy(0)\right\|_2^4}<+\infty$.

    \item[\textbf{\blue{(13)}}] As with \textbf{\blue{(10)}}, by adopting a uniform temporal discretization of $[0,T]$, with $\dt=T/J$ being its norm and $t_j=j\dt$ for $j=0,1,\ldots,J$ being its nodes, where $J\in\mathbb{N}$, we assume that
    \begin{equation*}
        a_{\dt}(M):=\underset{\theta>0}{\inf}\Bigg\{\mathbb{E}\Bigg[e^{\theta\big(M\big(1+\underset{0\leq j \leq J}{\max}\big\{\big\|\ms_1^{\vy,j}\big\|_2^2+\big\|\ms_2^{\vy,j}\big\|_2^2\big\}\big)-1/\dt\big)}\Bigg]\Bigg\}=o(\dt),
    \end{equation*}
    as $\dt\to0^+$, for every $M=M(\omega)$ with $\pp(M<+\infty)=1$.
    
\end{itemize}
For more details about the regularity conditions in \textbf{\blue{(1)}}--\textbf{\blue{(13)}}, including possible relaxations, see Andreou \& Chen \cite{andreou2024martingale}. Most importantly, in the context of this work, these assumptions are sufficient for establishing the CGNS optimal online forward-in-time discrete smoother from Theorem \ref{thm:onlinesmoother}.

\section{Proof of the Optimal Online Forward-In-Time Discrete Smoother} \label{sec:app2}

In this appendix we outline the details of the proof to Theorem \ref{thm:onlinesmoother} which outlines the procedure for the online forward-in-time discrete smoother for optimal state estimation of the unobserved processes.

\begin{proof}[\textbf{\underline{Proof of Theorem \ref{thm:onlinesmoother}}}] From the derivations in Andreou \& Chen \cite{andreou2024martingale} concerning the discrete smoother, for $n\in\mathbb{N}$ and $j=0,1,\ldots,n-1$, we have that
\begin{align}
    \vm{\normalfont{s}}^{j,n}&=\vm{\nf}^j+\mathbf{E}^j\big(\vm{\normalfont{s}}^{j+1,n}-(\mathbf{I}_{l\times l}+\ml^{\vy,j}\dt)\vm{\nf}^j-\vf^{\vy,j}\dt\big)+\mathbf{F}^j\big(\vx^{j+1}-\vx^{j}-(\ml^{\vx,j}\vm{\nf}^j+\vf^{\vx,j})\dt\big)= \mathbf{E}^j\vm{\normalfont{s}}^{j+1,n}+\mathbf{b}^j, \label{eq:app2aux1} \\
    \mr{\normalfont{s}}^{j,n}&=\mr{\nf}^j-\mathbf{C}_{j+1}^j\mathbf{C}^j_{22}(\mathbf{C}_{j+1}^j)^\dagger+\mathbf{E}^j\mr{\normalfont{s}}^{j+1,n}(\mathbf{E}^j)^\dagger=\mathbf{E}^j\mr{\normalfont{s}}^{j+1,n}(\mathbf{E}^j)^\dagger+\mathbf{P}_{j+1}^j, \label{eq:app2aux2}
\end{align}
where
\begin{align*}
    \mathbf{b}^{j}&:=\vm{f}^{j}-\mathbf{E}^{j}\big((\mathbf{I}_{l\times l}+\ml^{\vy,j}\dt)\vm{\nf}^{j}+\vf^{\vy,j}\dt\big)+\mathbf{F}^{j}\big(\vx^{j+1}-\vx^{j}-(\ml^{\vx,j}\vm{\nf}^{j}+\vf^{\vx,j})\dt\big),\\
    \mathbf{P}^{j}_{j+1}&:= \mr{f}^{j}-\mathbf{E}^{j}(\mathbf{I}_{l\times l}+\ml^{\vy,j}\dt)\mr{f}^{j}-\mathbf{F}^{j}\ml^{\vx,j}\mr{f}^{j}\dt, \label{eq:onlineauxiliary4}
\end{align*}
which immediately shows that \eqref{eq:onlineauxiliary1} and \eqref{eq:onlineauxiliary2} hold by plugging in $j=n-1$ into these, since the smoother and filter estimates are the same at the end point, i.e., $\vm{s}^{n,n}=\vm{f}^n$ and $\mr{s}^{n,n}=\mr{f}^n$. As such, we just need to prove \eqref{eq:onlinerecursive1} and \eqref{eq:onlinerecursive2}.

We start by proving the following relation for $n\in\mathbb{N}$ and $j=0,1,\ldots,n-1$:
\begin{equation} \label{eq:app2aux3}
    \mathbf{D}^{j,n-1}=\mathbf{E}^j\mathbf{E}^{j+1}\cdots\mathbf{E}^{n-1}=:\overset{\mathlarger{\curvearrowright}}{\prod^{n-1}_{i=j}} \mathbf{E}^i=\mathbf{E}^j\mathbf{D}^{j+1,n-1}.
\end{equation}
The curved arrow pointing to the right above the product symbol in \eqref{eq:app2aux3} indicates the order or direction with which we expand said product, since we are dealing with matrices, thus making the product, in general, noncommutative. Also, note that \eqref{eq:app2aux3} is able to encompass both the second and third relations of \eqref{eq:updatematrix1}, but not the first one due to its irregular form. Nonetheless, it could be extraneously incorporated in \eqref{eq:app2aux3} by noting that if $j=n$, and so the product on the right-hand side is ``empty", then $\mathbf{D}^{n,n-1}$ is trivially defined as the identity matrix $\mathbf{I}_{l\times l}$.

For $n\in\mathbb{N}_{\geq 2}$ and $j=0,1,\ldots,n-1$, we have from iterative application of the appropriate relation in \eqref{eq:updatematrix1} (the second and third relations are appropriately applied in what follows, possibly after some reindexing, without loss of generality), that
\begin{align*}
    \mathbf{D}^{j,n-1}=\mathbf{D}^{j,n-2}\mathbf{E}^{n-1}&=\begin{cases}
        \mathbf{E}^{n-1}, & \text{ for } j=n-1\\
        \mathbf{E}^{n-2}\mathbf{E}^{n-1}, & \text{ for } j=n-2\\
        \mathbf{D}^{j,n-3}\mathbf{E}^{n-2}\mathbf{E}^{n-1}, & \text{ for } j<n-2=\begin{cases}
        \mathbf{E}^{n-3}\mathbf{E}^{n-2}\mathbf{E}^{n-1}, & \text{ for } j=n-3\\
        \mathbf{D}^{j,n-4}\mathbf{E}^{n-3}\mathbf{E}^{n-2}\mathbf{E}^{n-1}, & \text{ for } j<n-3
    \end{cases}
    \end{cases}\\
    &=\cdots=\mathbf{D}^{j,j}\mathbf{E}^{j+1}\cdots\mathbf{E}^{n-3}\mathbf{E}^{n-2}\mathbf{E}^{n-1}\equiv \overset{\mathlarger{\curvearrowright}}{\prod^{n-1}_{i=j}} \mathbf{E}^i,
\end{align*}
which proves the relation on the left-hand side of \eqref{eq:app2aux3}. Similarly, from the previous result and \eqref{eq:updatematrix1} (the first relation), we have
\begin{equation*}
    \mathbf{E}^{j}\mathbf{D}^{j+1,n-1}=\begin{cases}
        \mathbf{E}^{j},& \text{ for } j+1=n\\
        \mathbf{E}^{j}\overset{\mathlarger{\curvearrowright}}{\prod^{n-1}_{i=j+1}} \mathbf{E}^i,& \text{ for } j+1\leq n-1<n
    \end{cases}\equiv \overset{\mathlarger{\curvearrowright}}{\prod^{n-1}_{i=j}} \mathbf{E}^i,
\end{equation*}
which is the same as the right-hand side of \eqref{eq:app2aux3}. For $n=1$, and so $j=0$ necessarily, we have
\begin{equation*}
    \mathbf{D}^{0,0}=\mathbf{E}^{0}=\overset{\mathlarger{\curvearrowright}}{\prod^{0}_{i=0}}\mathbf{E}^i,
\end{equation*}
by the second relation in \eqref{eq:updatematrix1}, while by the first relation in \eqref{eq:updatematrix1} we have
\begin{equation*}
     \mathbf{E}^0\mathbf{D}^{1,0}=\mathbf{E}^0=\overset{\mathlarger{\curvearrowright}}{\prod^{0}_{i=0}}\mathbf{E}^i,
\end{equation*}
and so by combining these two results proves \eqref{eq:app2aux3} also for $n=1$ and $j=0$. Collecting all these results together, establishes \eqref{eq:app2aux3} for $n\in\mathbb{N}$ and $j=0,\ldots,n-1$.

We proceed now by proving that for $n\in\mathbb{N}$ and $0\leq j\leq n-1$, the following equations are valid:
\begin{align}
    \vm{s}^{j,n}&=\mathbf{D}^{j,n-1}\vm{f}^n+\mathbf{b}^j+\sum_{r=j+1}^{n-1} \mathbf{D}^{j,r-1}\mathbf{b}^r, \label{eq:app2aux4}\\
    \mr{s}^{j,n}&=\mathbf{D}^{j,n-1}\mr{f}^n(\mathbf{D}^{j,n-1})^\dagger+\mathbf{P}_{j+1}^j+\sum_{r=j+1}^{n-1} \mathbf{D}^{j,r-1}\mathbf{P}_{r+1}^r(\mathbf{D}^{j,r-1})^\dagger. \label{eq:app2aux5}
\end{align}
We start with the equation for the smoother mean. We have by an iterative application of \eqref{eq:app2aux1} for $n\in\mathbb{N}$ and $j=0,\ldots,n-1$ that
\begin{align*}
    \vm{s}^{j,n}&=\mathbf{E}^j\underbrace{\vm{s}^{j+1,n}}_{\displaystyle \mathclap{\hspace{8.9cm}=\mathbf{E}^{j+1}\vm{s}^{j+2,n}+\mathbf{b}^{j+1}, \text{ if } j+1<n, \text{ otherwise equals } \vm{f}^{n}}}+\mathbf{b}^j\\
    \\
    &=\begin{cases}
        \mathbf{E}^j\mathbf{E}^{j+1}\smash[t]{\overbrace{\vm{s}^{j+2,n}}^{\displaystyle \mathclap{\hspace{8.8cm}=\mathbf{E}^{j+2}\vm{s}^{j+3,n}+\mathbf{b}^{j+2}, \text{ if } j+2<n, \text{ otherwise equals } \vm{f}^{n}}}}+\mathbf{E}^j\mathbf{b}^{j+1}+\mathbf{b}^j, & \text{ for } j<n-1\\
        \mathbf{E}^j\vm{f}^{n}+\mathbf{b}^j, & \text{ for } j=n-1
    \end{cases}
    \vphantom{
  \begin{cases}
    \overbrace{\vm{s}^{j+2,n}}^{\displaystyle \mathclap{\hspace{8cm}=\mathbf{E}^{j+2}\vm{s}^{j+3,n}+\mathbf{b}^{j+2}, \text{ if } j+2<n-2 \text{ otherwise equals } \vm{f}^{n}}} \\
    .
  \end{cases}}\\
  &=\cdots \ \left(\text{Likewise where at the } k\text{-th step we have } \vm{s}^{j+k,n}=\begin{cases}
      \mathbf{E}^{j+k}\vm{s}^{j+k+1,n}+\mathbf{b}^{j+k}, & \text{ for }j+k<n\\
      \vm{f}^{n,n}, & \text{ for } j+k=n
  \end{cases} \right) \\
  &\hspace{0.2cm}\vdots \hspace{0.8cm}\text{ (Inductively...)}\\
  &= \mathbf{E}^j\mathbf{E}^{j+1}\cdots \mathbf{E}^{n-1}\vm{f}^{n}+\mathbf{b}^j+\Bigg(\mathbf{E}^j\mathbf{b}^{j+1}+\mathbf{E}^j\mathbf{E}^{j+1}\mathbf{b}^{j+2}+\cdots+\Bigg(\overset{\mathlarger{\curvearrowright}}{\prod^{n-2}_{i=j}} \mathbf{E}^i\Bigg)\mathbf{b}^{n-1}\Bigg).
\end{align*}
Using now \eqref{eq:app2aux3} (possibly after some reindexing, without loss of generality), then this expression reduces down to
\begin{equation*}
    \vm{s}^{j,n}=\mathbf{D}^{j,n-1}\vm{f}^n+\mathbf{b}^j+\sum_{r=j+1}^{n-1} \mathbf{D}^{j,r-1}\mathbf{b}^r,
\end{equation*}
which is exactly what we wanted to prove for the smoother mean in \eqref{eq:app2aux4}. We now turn our attention to the smoother covariance matrix. In a very similar fashion, now by iteratively using \eqref{eq:app2aux2} for $n\in\mathbb{N}$ and $j=0,\ldots,n-1$, we get
\begin{align*}
    \mr{s}^{j,n}&=\mathbf{E}^j\underbrace{\mr{s}^{j+1,n}}_{\displaystyle \mathclap{\hspace{10.4cm}=\mathbf{E}^{j+1}\mr{s}^{j+2,n}(\mathbf{E}^{j+1})^\dagger+\mathbf{P}^{j+1}_{j+2}, \text{ if } j+1<n, \text{ otherwise equals } \mr{f}^{n}}}(\mathbf{E}^j)^\dagger+\mathbf{P}^j_{j+1}\\
    \\
    &=\begin{cases}
        \mathbf{E}^j\mathbf{E}^{j+1}\smash[t]{\overbrace{\mr{s}^{j+2,n}}^{\displaystyle \mathclap{\hspace{10.3cm}=\mathbf{E}^{j+2}\mr{s}^{j+3,n}(\mathbf{E}^{j+2})^\dagger+\mathbf{P}^{j+2}_{j+3}, \text{ if } j+2<n, \text{ otherwise equals } \mr{f}^{n}}}}(\mathbf{E}^{j+1})^\dagger(\mathbf{E}^{j})^\dagger+\mathbf{E}^j\mathbf{P}^{j+1}_{j+2}(\mathbf{E}^j)^\dagger+\mathbf{P}^j_{j+1}, & \text{ for } j<n-1\\
        \mathbf{E}^j\mr{f}^{n}(\mathbf{E}^j)^\dagger+\mathbf{P}^j_{j+1}, & \text{ for } j=n-1
    \end{cases}
    \vphantom{
  \begin{cases}
    \overbrace{\vm{s}^{j+2,n}}^{\displaystyle \mathclap{\hspace{8cm}=\mathbf{E}^{j+2}\vm{s}^{j+3,n}+\mathbf{b}^{j+2}, \text{ if } j+2<n-2 \text{ otherwise equals } \vm{f}^{n}}} \\
    .
  \end{cases}}\\
  &=\cdots \ \left(\text{Likewise where at the } k\text{-th step we have } \mathbf{R}_{\text{s}}^{j+k,n}=\begin{cases}
      \mathbf{E}^{j+k}\mathbf{R}_{\text{s}}^{j+k+1,n}(\mathbf{E}^{j+k})^\dagger+\mathbf{P}_{j+k+1}^{j+k}, & \text{ for }j+k<n\\
      \mathbf{R}_{\text{f}}^{n,n}, & \text{ for } j+k=n
  \end{cases} \right)\\
  &\hspace{0.2cm}\vdots \hspace{0.8cm}\text{ (Inductively...)}\\
  &= \mathbf{E}^j\mathbf{E}^{j+1}\cdots \mathbf{E}^{n-1}\mr{f}^{n,n}(\mathbf{E}^{n-1})^
  \dagger\cdots (\mathbf{E}^{j+1})^\dagger(\mathbf{E}^{j})^\dagger+\mathbf{P}_{j+1}^j\\
  &\hspace{2cm}+\Bigg(\mathbf{E}^j\mathbf{P}^{j+1}_{j+2}(\mathbf{E}^j)^\dagger+\mathbf{E}^j\mathbf{E}^{j+1}\mathbf{P}^{j+2}_{j+3}(\mathbf{E}^{j+1})^\dagger(\mathbf{E}^{j})^\dagger+\cdots+\Bigg(\overset{\mathlarger{\curvearrowright}}{\prod^{n-2}_{i=j}} \mathbf{E}^i\Bigg)\mathbf{P}^{n-1}_n\Bigg(\overset{\mathlarger{\curvearrowleft}}{\prod^{n-2}_{i=j}} (\mathbf{E}^i)^\dagger\Bigg)\Bigg).
\end{align*}
Like we did before, by using \eqref{eq:app2aux3} (possibly after some reindexing, without loss of generality), and by noticing that
\begin{equation*}
        (\mathbf{D}^{j,r-1})^\dagger=\Bigg(\overset{\mathlarger{\curvearrowright}}{\prod^{r-1}_{i=j}} \mathbf{E}^i\Bigg)^\dagger=\overset{\mathlarger{\curvearrowleft}}{\prod^{r-1}_{i=j}} (\mathbf{E}^i)^\dagger,
\end{equation*}
for $r=n,n-1,\ldots,j+1$, then the previous expression is equivalent to,
\begin{equation*}
    \mr{s}^{j,n}=\mathbf{D}^{j,n-1}\mr{f}^n(\mathbf{D}^{j,n-1})^\dagger+\mathbf{P}_{j+1}^j+\sum_{r=j+1}^{n-1} \mathbf{D}^{j,r-1}\mathbf{P}_{r+1}^r(\mathbf{D}^{j,r-1})^\dagger,
\end{equation*}
which is exactly what we wanted to prove for the smoother covariance in \eqref{eq:app2aux5}.

Now to verify the recursive formulae given in \eqref{eq:onlinerecursive1} and \eqref{eq:onlinerecursive2}, we use \eqref{eq:app2aux4} and \eqref{eq:app2aux5} by looking at the differences $\vm{s}^{j,n}-\vm{s}^{j,n-1}$ and $\mr{s}^{j,n}-\mr{s}^{j,n-1}$, respectively. First, observe that \eqref{eq:onlinerecursive1} and \eqref{eq:onlinerecursive2} trivially hold in the case where $n=1$, and as such we necessarily have $j=0,\ldots,n-1\Leftrightarrow j=0$. This is immediate by the first relation in \eqref{eq:updatematrix1} and the fact that $\vm{s}^{0,0}=\vm{f}^0$ and $\mr{s}^{0,0}=\mr{f}^0$. As such, we only consider the case where $n\in\mathbb{N}_{\geq 2}$ and $j=0,\ldots,n-1$. From \eqref{eq:app2aux4} we have,
\begin{equation*}
    \vm{s}^{j,n}-\vm{s}^{j,n-1}=\mathbf{D}^{j,n-1}\vm{f}^{n}+\mathbf{D}^{j,n-2}\mathbf{b}^{n-1}-\mathbf{D}^{j,n-2}\vm{f}^{n-1},
\end{equation*}
and by \eqref{eq:app2aux5} we get,
\begin{equation*}
    \mr{s}^{j,n}-\mr{s}^{j,n-1}=\mathbf{D}^{j,n-1}\mr{f}^n(\mathbf{D}^{j,n-1})^\dagger+\mathbf{D}^{j,n-2}\mathbf{P}_{n}^{n-1}(\mathbf{D}^{j,n-2})^\dagger-\mathbf{D}^{j,n-2}\mr{f}^{n-1}(\mathbf{D}^{j,n-2})^\dagger.
\end{equation*}
Replacing now the coefficient $\mathbf{D}^{j,n-1}$ by $\mathbf{D}^{j,n-2}\mathbf{E}^{n-1}$ via the third relation in \eqref{eq:updatematrix1}, since it holds for $n\in\mathbb{N}_{\geq 2}$, $j=0,\ldots,n-1$, for the smoother mean we have,
\begin{align*}
    \vm{s}^{j,n}&=\vm{s}^{j,n-1}+\mathbf{D}^{j,n-2}\Big(\underbrace{\mathbf{E}^{n-1}\vm{f}^{n}+\mathbf{b}^{n-1}}_{\displaystyle\mathclap{\hspace{4.2cm}=\vm{s}^{n-1,n}, \text{ because of } \eqref{eq:onlineauxiliary1}}}-\vm{f}^{n-1}\Big)\\
    &=\vm{s}^{j,n-1}+\mathbf{D}^{j,n-2}\left(\vm{s}^{n-1,n}-\vm{f}^{n-1}\right),
\end{align*}
and for the smoother covariance,
\begin{align*}
    \mr{s}^{j,n}&=\mr{s}^{j,n-1}+\mathbf{D}^{j,n-2}\Big(\underbrace{\mathbf{E}^{n-1}\mr{f}^{n}(\mathbf{E}^{n-1})^
    \dagger+\mathbf{P}_n^{n-1}}_{\displaystyle\mathclap{\hspace{4.2cm}=\mr{s}^{n-1,n}, \text{ because of } \eqref{eq:onlineauxiliary2}}}-\mr{f}^{n-1}\Big)(\mathbf{D}^{j,n-2})^\dagger\\
    &=\mathbf{R}_{\text{s}}^{j,n-1}+\mathbf{D}^{j,n-2}\left(\mathbf{R}_{\text{s}}^{n-1,n}-\mathbf{R}_{\text{f}}^{n-1}\right)(\mathbf{D}^{j,n-2})^\dagger,
\end{align*}
where by \eqref{eq:app2aux3} we have (after the trivial reindexing $n-1\leadsto n-2$, without loss of generality),
\begin{equation*}
    \mathbf{D}^{j,n-2}=\overset{\mathlarger{\curvearrowright}}{\prod^{n-2}_{i=j}} \mathbf{E}^i.
\end{equation*}
With these we have proved \eqref{eq:onlinerecursive1} and \eqref{eq:onlinerecursive2} for $n\in\mathbb{N}$ and $j=0,1,\ldots,n-1$. For $j=n$, by the fundamental filter-smoother relation, we have $\vm{s}^{n,n}=\vm{f}^n$, $\mr{s}^{n,n}=\mr{f}^n$. This finishes the proof of Theorem \ref{thm:onlinesmoother}.
\end{proof}

\section{Trade-off Analysis Between the Computational and Storage Complexity and Posterior Solution Skill for the Dyad Model in \eqref{eq:dyad1}--\eqref{eq:dyad2}} \label{sec:app3}

Figure \ref{fig:Dyad_Interaction_Fig_Appendix}, similar to Figure \ref{fig:LDA_Ice_Floes_Fig_2}, compares the computational time (in seconds) and storage requirements (in gigabytes) of the fixed- and adaptive-lag online smoothers against the NRMSE between their posterior mean time series and the true unobserved signal. As with Figure \ref{fig:LDA_Ice_Floes_Fig_2}, computational times in Panels (b), (e), and (g) are averaged over multiple runs to reduce external fluctuations, with corresponding panels being plotted in a similar manner. Details on how the storage and time values are calculated are given in Appendix \ref{sec:app4}. Comparing these with the computational behavior of the fixed- and adaptive-lag online smoothers in the LDA application (see Figure \ref{fig:LDA_Ice_Floes_Fig_2}), which is a model that lacks the intermittent instabilities that emerge in the dyad model, \eqref{eq:dyad1}--\eqref{eq:dyad2}, more intricate results arise.

Starting with the similarities between the two models, both the storage use (Panel (a), in red) and computational time (Panel (b), in red) of the fixed-lag online smoother show an algebraically increasing behavior with the fixed lag value. Furthermore, as is to be expected, since a constant upper lag bound parameter is used for all simulations of the adaptive-lag online smoother ($b\dt=3$ simulation time units), the storage use does not vary with decreasing $\delta$ (Panel (d), in red).

In contrast, the NRMSE of the fixed-lag method in the dyad model (Panels (a)--(b), in blue) roughly indicates an algebraic decrease with increasing fixed lag, instead of the exponential one in LDA (see Panels (a) and (b) of Figure \ref{fig:LDA_Ice_Floes_Fig_2}). With regards to the adaptive-lag online smoother's computational performance in time (Panel (e), in red), as well as its prediction skill formulated via the NRMSE (Panels (d)--(e), in blue), both are nontrivial functions of the tolerance parameter $\delta$. Specifically, with decreasing $\delta$, after a certain tolerance value, there is a regime switch from, roughly speaking, a sub-logarithmic increase for the former and algebraic to exponential decrease for the latter, to a logarithmic to algebraic increase for the former and algebraic decrease for the latter. In Panel (e), like the one in Figure \ref{fig:LDA_Ice_Floes_Fig_2}, we also show at each data point on the computational time curve the percentage of time spent on the adaptive-lag calculation (including the calculation of the update values $D^{j,n-2}$). Due to the one-dimensional observable and unobservable state spaces, the time consumed on this operation is much smaller than that in the LDA case study, but, because of this, the storage need reductions are only marginal in nature. These observations collectively indicate how, in such settings, the adaptive-lag algorithm has the potential to provide a significant reduction in computational complexity beyond its fixed-lag counterpart, while maintaining a skillful state estimation in pathwise error statistics (even when using a relatively inflated upper lag bound of $b\dt=3$).

The aforementioned results also emerge in the recovery skill of these algorithms, formulated through the temporally-averaged information gain beyond the filter solution, depicted in Panels (c) for fixed lag and (f) for adaptive lag. In these, the associated signal and dispersion parts are also plotted. (Note how in the dyad model, due to the intermittent extreme events, the dispersion is comparatively much smaller than the signal part.) In the former, we observe how the time-averaged information gain algebraically increases with respect to the fixed lag value, while in the latter we see a somewhat sub-logarithmic behavior with decreasing $\delta$ (after the aforementioned tolerance parameter threshold is met).

Finally, in Panel (g), the results from using the LSDev sequence to define the adaptive lag (i.e., \eqref{eq:adaptlaginfodef1}--\eqref{eq:adaptlaginfodef2}) are plotted, similar to Panel (e). Since, as observed from Panel (h) in Figure \ref{fig:Dyad_Interaction_Fig_2}, in this case study the extreme events do not allow for the standardized sequence of relative entropies or information gains in \eqref{eq:infogainstandardized} to behave in an overall exponential manner in $j$ uniformly over $n$, as in the LDA numerical experiment (see Panel (b) and (d) in Figure \ref{fig:LDA_Ice_Floes_Fig_3}), the LSDev approach showcases a slightly distinct behavior with respect to the tolerance parameter $\delta$ when compared to the approach which uses the original sequence (see Panel (e)). Specifically, there is a trade-off to be made; for larger NRMSE the use of the original sequence of relative entropies outperforms in terms of computational time, while for smaller ones, particularly for the apparent lower-bound limit at $0.6315$, the LSDev approach is able to terminate faster.

\begin{figure}[!ht]
\centering
    \includegraphics[width=\textwidth]{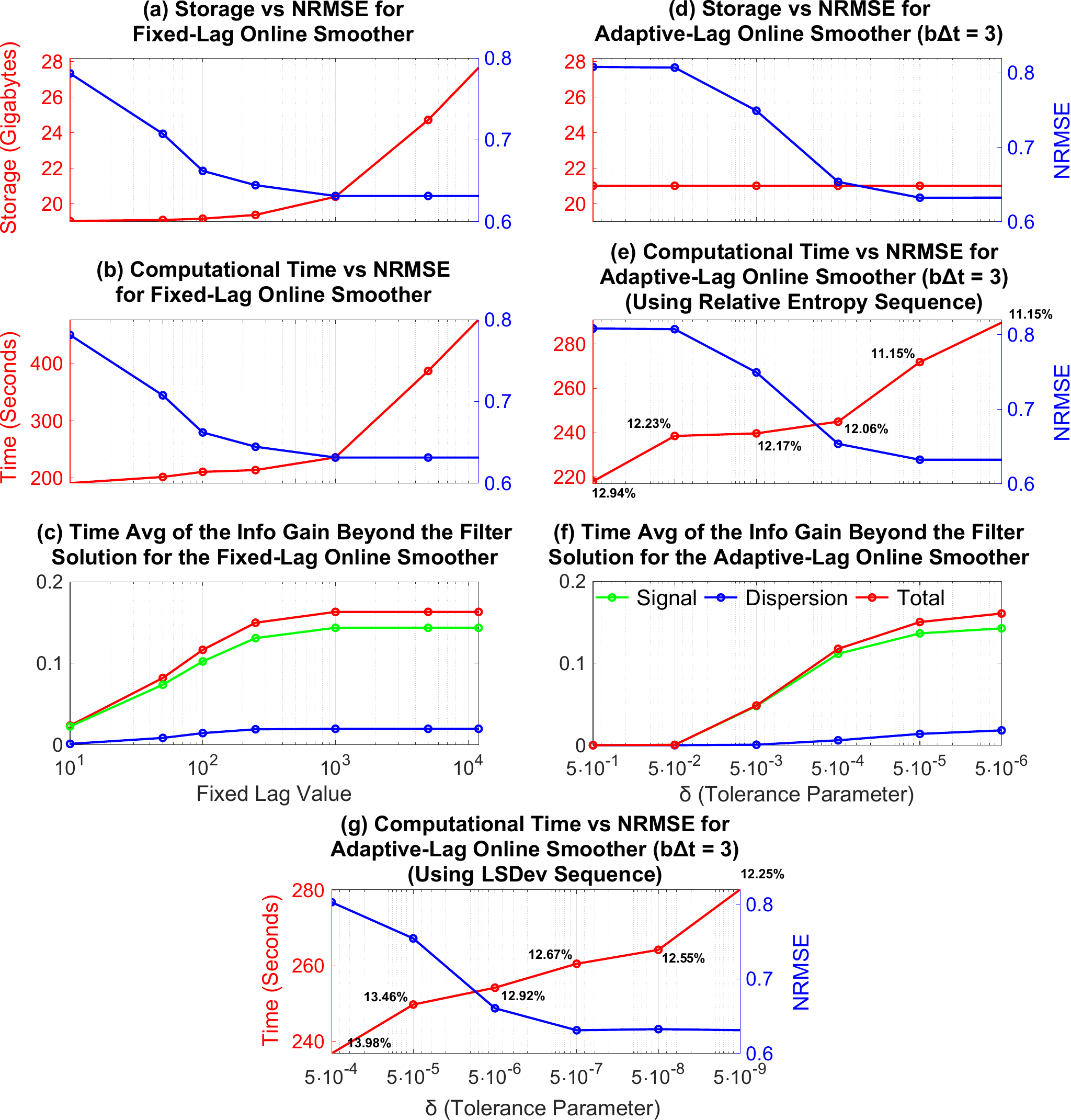}
    \caption{Panels (a)--(f): Same as Panels (a)--(f) of Figure \ref{fig:LDA_Ice_Floes_Fig_2}, but for the dyad-interaction model, \eqref{eq:dyad1}--\eqref{eq:dyad2}, from the case study in Section \ref{sec:4.1}. Panel (g): Same as Panel (e) but using the LSDev sequence to define the adaptive lags (per \eqref{eq:adaptlaginfodef1}--\eqref{eq:adaptlaginfodef2}).}
    \label{fig:Dyad_Interaction_Fig_Appendix}
\end{figure}

\section{Details of the Storage Value and Computational Time Calculation in Figures \ref{fig:LDA_Ice_Floes_Fig_2} and \ref{fig:Dyad_Interaction_Fig_Appendix}} \label{sec:app4}

The storage values (in gigabytes) depicted in Panels (a) and (d) of Figure \ref{fig:LDA_Ice_Floes_Fig_2} (and by extension in Panels (a) and (d) of Figure \ref{fig:Dyad_Interaction_Fig_Appendix}) are calculated by adding the storage required for storing each of the MATLAB nested cell array structures that are used to hold the online smoother mean vectors $\boldsymbol{\mu}_{\text{s}}^{j,n}$, online smoother covariance matrices $\mathbf{R}_{\text{s}}^{j,n}$, and the online smoother update matrices $\mathbf{D}^{j,n-2}$ (one nested cell array for each), for that run of the online smoother algorithm (Algorithm \ref{algo:onlinesmoother}). The outer cell array (of each nested cell array) is $1\times 1000$ in size, where $1000$ is the total number of observations in this numerical experiment. Each cell in the outer cell array holds another ($1 \times n$)-sized cell array, where $n$ is the ordinal number of the current observation, which itself, in each of its cells, holds the associated multidimensional quantity mentioned prior (vector or matrix, respectively). Finally, all of these storage values also include the storage requirements for the 2D matrix holding the filter mean vectors and 3D matrices holding the filter covariance matrices and $\mathbf{E}^j$ and $\mathbf{F}^j$ auxiliary matrices (one for each). For the adaptive-lag online smoother we also include the 2D array holding the adopted adaptive-lag-defining criterion, $\left\{ \mathcal{P}\left(p^{j,n}_{\text{updated}}, p^{j,n}_{\text{lagged}}\right)\right\}_{R_n\leq j\leq n-1}$, or $\{\sigma^{j,n}\}_{R_n\leq j\leq n-1}$.\medskip

We utilize nested cell arrays for these case studies to simulate jagged or ragged 3D- and 4D-arrays, since the numerical experiments are not being carried out on an HPC with GPU clusters (and since MATLAB does not natively support jagged multidimensional arrays). Depending on the fixed lag value $L$ for the fixed-lag online smoother and the upper lag bound $b$ for the adaptive-lag online smoother, only the necessary components that are required to carry out the associated online smoother updates are calculated and stored, and therefore added to the aforementioned storage values. While this adopted approach leads to an overcalculation of the computational time required, as MATLAB can struggle with accessing elements in a (nested) cell array (due to how they occupy place in memory), this method allows us to run these algorithms efficiently even for high-dimensional systems without the need for an HPC with GPU clusters.

As for the computational time values, both for the fixed- and adaptive-lag algorithms, the simulation time does not include the time required for the calculation of the filter statistics, of $\mathbf{E}^j$ and $\mathbf{F}^j$ (see \eqref{eq:discretesmootherauxmat1}--\eqref{eq:discretesmootherauxmat2}), and of the model parameters at each newly obtained observation. Those are assumed to already be calculated and stored in memory for quick access. Therefore, the computational times depicted in these figures only cumulate the time needed at each $n$ (i.e., each observation) to access said memory and to carry out the the calculations in Algorithm \ref{algo:onlinesmoother} for the $j$'s defined either by the fixed or adaptive lag value, where for the latter we also include the time required to determine it (i.e., calculate the information gain sequence (and possibly its LSDev) and then use \eqref{eq:adaptlaginfodeforiginalseq} (or \eqref{eq:adaptlaginfodef1})). Importantly, it does not include the time required to create and delete the aforementioned nested cell arrays that hold the online smoother Gaussian statistics and update matrices $\mathbf{D}^{j,n-2}$.

\section{Details of the Online Expectation-Maximization Algorithm for Parameter Estimation} \label{sec:app5}

In this appendix, we outline the basic learning EM algorithm, which sets the baseline structure for the online model identification algorithm developed in Section \ref{sec:4.3.1}. As usual, it is assumed that the coupled CGNS in \eqref{eq:condgauss1}--\eqref{eq:condgauss2} is only partially observed, and based on the time discretization setup adopted thus far in this work (see \eqref{eq:discretecondgauss1}--\eqref{eq:discretecondgauss2} in Section \ref{sec:3.1}), we denote the set of values of the state variables, when evaluated on the discrete-time steps $t_j$, as $\mathcal{X}=\{\vx^0,\ldots,\vx^j,\ldots,\vx^n\}$ and $\mathcal{Y}=\{\vy^0,\ldots,\vy^j,\ldots,\vy^n\}$, for the observable and unobservable components, respectively.

Given an ansatz of the model, the goal here is to maximize an objective function, specifically the log-likelihood of the observational process, with respect to the model parameters,
\begin{equation*}
\ell(\boldsymbol{\theta}):=\log\big(p(\mathcal{X};\boldsymbol{\theta})\big)=\log\left(\int_{\mathcal{Y}}p(\mathcal{X},\mathcal{Y};\boldsymbol{\theta})\d\mathcal{Y}\right),
\end{equation*}
where $\boldsymbol{\theta}$ is the collection of model parameters. This is also known as the marginal log-likelihood of the observed data, with the full log-likelihood being given by $\log\big(p(\mathcal{X},\mathcal{Y};\boldsymbol{\theta})\big)$. Using any distribution over the discretized hidden variables, with density $q(\mathcal{Y})$, a lower bound on the marginal log-likelihood can be obtained in a rather trivial manner in the following way \cite{ghahramani1998learning}:
\begin{align*}
    \ell(\boldsymbol{\theta})=\log\left(\int_{\mathcal{Y}}p(\mathcal{X},\mathcal{Y};\boldsymbol{\theta})\d\mathcal{Y}\right)&=\log\left(\int_{\mathcal{Y}}q(\mathcal{Y})\frac{p(\mathcal{X},\mathcal{Y};\boldsymbol{\theta})}{q(\mathcal{Y})}\d\mathcal{Y}\right)\\
    &\geq\int_{\mathcal{Y}}q(\mathcal{Y})\log\left(\frac{p(\mathcal{X},\mathcal{Y};\boldsymbol{\theta})}{q(\mathcal{Y})}\right)\d\mathcal{Y} \ (\text{by Jensen's integral inequality}) \\
    &= \int_{\mathcal{Y}}q(\mathcal{Y})\log\big(p(\mathcal{X},\mathcal{Y};\boldsymbol{\theta})\big)\d\mathcal{Y}-\int_{\mathcal{Y}}q(\mathcal{Y})\log\big(q(\mathcal{Y})\big)\d\mathcal{Y}\\
    &:= \mathcal{F}(q,\boldsymbol{\theta}),
\end{align*}
where the negative value of $\int_{\mathcal{Y}}q(\mathcal{Y})\log\big(p(\mathcal{X},\mathcal{Y};\boldsymbol{\theta})\big)\d\mathcal{Y}$ is the so-called free energy, while $\mathcal{S}(q)=-\int_{\mathcal{Y}}q(\mathcal{Y})\log\big(q(\mathcal{Y})\big)\d\mathcal{Y}$ is the Shannon differential entropy. Therefore, based on the fact that $\mathcal{F}(q,\boldsymbol{\theta})\leq \ell(\boldsymbol{\theta})$ for all distributions $q$ over $\mathcal{Y}$, it is clear that maximizing the marginal log-likelihood is equivalent to maximizing $\mathcal{F}$ alternatively with respect to the distribution $q$ and the parameters $\boldsymbol{\theta}$ \cite{neal1998view, hastie2001elements}, thus creating a bridge between the usual MLE algorithm and EM procedure. This alternate maximization of $\mathcal{F}$ is the essence of the EM algorithm, where the EM designation stems from the following procedure:
\begin{itemize}[leftmargin=3cm]
    \item[E-Step: ] $q_{n+1}\leftarrow \underset{q}{\mathrm{argmax}}\big\{\mathcal{F}(q,\boldsymbol{\theta}_n)\big\}$

    \item[M-Step: ] $\boldsymbol{\theta}_{n+1}\leftarrow \underset{\boldsymbol{\theta}}{\mathrm{argmax}}\big\{\mathcal{F}(q_{n+1},\boldsymbol{\theta})\big\}$
\end{itemize}

An essential observation in this process is that the maximization in the E-Step is exactly reached when $q$ is the posterior conditional distribution of $\mathcal{Y}$, that is
\begin{equation*}
    q_{n+1}(\mathcal{Y})=p\big(\mathcal{Y}\big|\mathcal{X};\boldsymbol{\theta}_n\big),
\end{equation*}
which in the online setting of sequential arrival of observations corresponds to the online smoother posterior distribution. This is attributed to the equality condition for the Jensen inequality, which we used to prove the bound between the marginal log-likelihood and the objective function $\mathcal{F}$. In Jensen's integral inequality, $\phi\big(\ee{X}\big)\leq \mathbb{E}\big[\phi(X)\big]$ for $\phi$ being a convex operator and $X$ an arbitrary-valued random variable, it is known that equality holds if and only if $\phi$ is an affine function or if $X$ is constant almost surely. As such, since the objective function $\mathcal{F}$ can be rewritten in the following much more intuitive form:
\begin{equation*}
    \mathcal{F}(q,\boldsymbol{\theta})=-\mathcal{P}\big(q,p(\cdot|\mathcal{X};\boldsymbol{\theta})\big)+\ell(\boldsymbol{\theta}),
\end{equation*}
where $\mathcal{P}$ is the relative entropy from Section \ref{sec:3.4.1}, then equality holds when $\log\left(\frac{p(\mathcal{X},\mathcal{Y};\boldsymbol{\theta})}{q(\mathcal{Y})}\right)$ is an affine function, which is exactly achieved when $q(\mathcal{Y})=p(\mathcal{Y}|\mathcal{X};\boldsymbol{\theta})\propto p(\mathcal{X},\mathcal{Y};\boldsymbol{\theta})$, with the proportionality being due to Bayes'. Of course, this fact can also be seen through the Gibbs' inequality, since for $q(\mathcal{Y})=p(\mathcal{Y}|\mathcal{X};\boldsymbol{\theta})$ the relative entropy term on the right-hand side vanishes \cite{little2019statistical}. In such a situation, the bound $\mathcal{F}(q,\boldsymbol{\theta})\leq \ell(\boldsymbol{\theta})$ becomes an equality, $\mathcal{F}(q,\boldsymbol{\theta}) = \ell(\boldsymbol{\theta})$. Note that the conditional distribution in the E-Step is very difficult to solve for general CTNDSs. Various numerical methods and approximations are often used \cite{ghahramani1998learning, ghahramani1996parameter}, which usually suffer from both approximation errors and the curse of dimensionality. Nevertheless, for the CGNS framework, the distribution $p(\mathcal{Y}|\mathcal{X};\boldsymbol{\theta}_k)$ is given in an optimal and unbiased manner by the closed analytic formulae of the online discrete smoother from Theorem \ref{thm:onlinesmoother}, which greatly facilitates the application of the online EM parameter estimation algorithm to many nonlinear models. On the other hand, since the differential entropy does not depend on $\boldsymbol{\theta}$, the maximum in the M-Step is obtained by maximizing the negative of the free energy,
\begin{equation*}
    \boldsymbol{\theta}_{n+1}\leftarrow \underset{\boldsymbol{\theta}}{\mathrm{argmax}}\left\{\int_{\mathcal{Y}}p\big(\mathcal{Y}\big|\mathcal{X};\boldsymbol{\theta}_n\big)\log\big(p(\mathcal{X},\mathcal{Y};\boldsymbol{\theta})\big)\d\mathcal{Y}\right\}.
\end{equation*}
This interpretation of the basic learning EM algorithm is known as the ``maximization-maximization procedure" formulation, where the EM algorithm is viewed as two alternating maximization steps, which is a specific application of coordinate descent. We finally note that the EM algorithm enjoys a monotonicity property, which states that improving the negative of the free energy will at least not make the marginal log-likelihood worse with regards to the optimization goal \cite{chen2010demystified}, and that a convergence analysis for the EM method (including results for distributions outside the exponential family), can be found in the fundamental theoretical work of Wu \cite{wu1983convergence}. As for the explicit expression of $\boldsymbol{\theta}_{n+1}$ in the case of the CGNS, as well as of the noise feedbacks in the case of unknown uncertainty matrices, both with respect to the posterior statistics, see the work of Chen \cite{chen2020learning}.

Finally, here we provide a simple numerical trick for accelerating the convergence of the parameters to their true values in the online EM algorithm. This process is based on momentum-based methods from classical convex optimization, which utilize momentum terms in the update to aid the learning process (e.g., Nesterov’s accelerated gradient descent \cite{nesterov1983method}). We denote by $\boldsymbol{\theta}_{n}$ and $\boldsymbol{\theta}_{n+1}$ the learned parameters (or part of them) in the previous ($n$) and current ($n+1$) iteration, respectively. An acceleration of the parameter value at the current step can be achieved by
\begin{equation*}
    \boldsymbol{\theta}_{n+1}^{\text{new}}=\boldsymbol{\theta}_{n+1}+\alpha\left(\boldsymbol{\theta}_{n+1}-\boldsymbol{\theta}_{n}\right),
\end{equation*}
where $\alpha\in[0,1]$ is a hyperparameter. If $\alpha=0$, then there is no acceleration; the acceleration rate depends on the amplitude of $\alpha$. Such a trick can be applied in the first few iterations of the online EM algorithm, especially for those in the unobserved process when the observability of the system is weak.

\section{Stability, Sensitivity, and Convergence of the Adaptive-Lag Online EM Parameter Estimation Algorithm for the Dyad Model \eqref{eq:emdyadmodel1}--\eqref{eq:emdyadmodel2} in Initial Parameter Values and Learning Duration} \label{sec:app6}

In this appendix we briefly discuss the stability, sensitivity, and convergence of the (adaptive-lag) online EM parameter estimation algorithm, when applied to the dyad model in \eqref{eq:emdyadmodel1}--\eqref{eq:emdyadmodel2}, with respect to the algorithm's burn-in or learning duration and initial parameter values. 

The results noted and discussed in what follows have been derived on a purely empirical basis, through trial and error, due to the tremendous computational costs that a full parameter space analysis entails. Specifically, this empirical analysis is conducted in a ``ceteris paribus" fashion, by adjusting only a single parameter's initial guess or the algorithm's learning period, while all other components of the algorithm remain unchanged from those in the simulation presented in Section \ref{sec:4.3.2}. Furthermore, these adjustments are made in both directions, i.e., by both decreasing or increasing the component of interest, and done so to severe degrees as to confirm the algorithm's stability or convergence skill is not localized under this regime.

For the dyad model in \eqref{eq:emdyadmodel1}--\eqref{eq:emdyadmodel2}, the adaptive-lag online EM parameter estimation algorithm carried out to derive the results in Figure \ref{fig:EM_Parameter_Estimation_Fig} is extremely stable, both in terms of the initial parameter values as well as in the length of the burn-in or learning period. Here, by stable, we mean that the algorithm does not destructively deteriorate and blow up; the algorithm might still diverge and approach a different subset of the parameter space as it evolves, but at least does so in a stable manner.

In terms of sensitivity and convergence skill with respect to the initial parameter values and burn-in or learning period, when we apply the adaptive-lag online EM parameter estimation algorithm on the dyad model in \eqref{eq:emdyadmodel1}--\eqref{eq:emdyadmodel2}, we have:


\begin{itemize}[leftmargin=2cm]

    \item[$d_u$:] Initial overestimation (i.e., a large positive value) affects $d_u$'s convergence, but it also affects that of $F_v$ (by overestimating it), as is to be expected; overdamping $u$ initially requires a larger positive $F_v$ to counteract this through the $-\gamma u^2$ term. The convergence of the other parameters remains unaffected. Relatively large values of initial $d_u$ are needed to effectively diverge $d_u$ and $F_v$ away from their true values.
    
    If $d_u$ is set equal to zero initially, it remains close to its neighborhood and oscillates there even after extreme event observance. This is compensated by estimating a negative $F_v$ as the algorithm evolves. Other parameters' convergence is again not severely affected by this.

    Initial underestimation (i.e., a large, in absolute value, negative initial guess) is not corrected with respect to the sign as the algorithm evolves, even after observing extreme events, with $d_u$ remaining at a significant negative value throughout. $F_v$ diverges as well, with large negative values. The convergence of the other parameters remains unaffected.
    
    \item[$\gamma$:] Overestimating $\gamma$ initially severely skews the convergence of itself, $d_v$, and $F_v$, but for $\gamma$ and $F_v$ this is corrected with the observance of extreme events (especially after extreme event A; see Panel (a) in Figure \ref{fig:EM_Parameter_Estimation_Fig}). For $d_v$, it converges towards an extremely damped regime, where this is to be expected due to the $-\gamma u^2$ term.
    
    By starting from zero, $\gamma$ remains close to it for the whole run, with $d_u$ being overestimated and the other parameters being underestimated.

    For negative initial $\gamma$, the algorithm's estimation for it remains negative and converges towards $-\gamma$, due to the symmetry of the energy conservation condition, but for large negative values it compensates by significantly overestimating $d_v$ and having a negative significant $F_v$. $d_u$ and $F_u$ are unaffected under this regime. 

    Severe positive or negative $\gamma$ does slow down its convergence, but the true value is well approximated eventually after a sufficiently long run.

    \item[$F_u$:] Significantly overestimating $F_u$ initially leads to the algorithm to search for an antidamping regime for $u$ over the parameter space (in $d_u$) and to severely underestimate $\gamma$, due to the quadratic coupling effects. It also completely misidentifies the $F_v$ value, due to the significant joint contribution of $ F_u$ and $ F_v$ in achieving system observability. 
    
    Same for underestimating, with very large initial $F_v$ negative values, but in this case $\gamma$'s convergence is fine, with problems only emerging in the other parameters.

    \item[$d_v$:] By initially severely overestimating $d_v$, the convergence can be severely skewed, leading to misidentification, because $v$ is the driving force behind the generation of the intermittent extreme events of the system (in $u$); after all, observing these instabilities accelerates and corrects the algorithm's current estimation of the parameters. 

    Same for underestimating, with very large negative values, which, although leading to an incorrect estimation for $d_v$, the algorithm at least immediately corrects the sign to a positive one (unlike in $d_u$'s case). The convergence of all other parameters is effectively unaffected.
    
    Starting at zero for $d_v$ initially has similar results as for negative initial values, i.e., they get quickly corrected, especially at the observance of extreme events.

    Finally, we note that even if $d_v$ is misidentified, as in the case of the results shown in Figure \ref{fig:EM_Parameter_Estimation_Fig} of Section \ref{sec:4.3.2}, leading to lower-order or time-series-based errors, nevertheless the memory and fidelity of the dyad model, represented by the higher-order metrics like the PDFs and ACFs of $u$ and $v$, are still effectively recovered.

    \item[$F_v$:] Only for significant initial overestimation, it only slightly slows down and skews $d_v$'s convergence, while it marginally aids in the convergence of $d_u$ and $d_v$ when assumed to be negative and large.

    \item[$\substack{\displaystyle\text{Learning}\\ \displaystyle\text{\quad Period:}}$] For the initial guess of $\boldsymbol{\theta}_0=(2,6,2,0.5,0.6)^\tran$, we have the following effects from adjusting the length of the learning period from the $10$ time units used in Section \ref{sec:4.3.2} (in $[0,200]$).
    \begin{itemize}[leftmargin=3cm]
        \item[Increase:] Convergence skill, i.e., reaching the same values as the ones depicted in Figure \ref{fig:EM_Parameter_Estimation_Fig} of Section \ref{sec:4.3.2}, is not hindered by increasing the learning period. But, we do note that a really long burn-in period, combined with severely wrong initial parameter guesses (which are the ones utilized during it), can be rather destructive and lead to incorrect values and even severe instability and divergence. Also, a significantly long learning period can severely slow down convergence even for stable initial parameter values; e.g., by using a 30 time units-long learning period in this case, we miss out on doing state estimation at the observance of the highly influential extreme event A (see Panel (a) in Figure \ref{fig:EM_Parameter_Estimation_Fig}).

        \item[Decrease:] Still stable, but trace plots are highly variable initially if chosen to be really brief. These variations are smoothed out relatively quickly, and satisfactory convergence is achieved at the observance of extreme events (but $d_v$ still converges to double the true value). Notably though, a training period that is at least a single time unit long is necessary for this case study; too small and $\gamma$ is thought to be zero, with similar results to as if we started with an initial guess of $\gamma=0$, as the learning period did not have enough data to showcase the presence of the energy-conserving quadratic coupling between $u$ and $v$.
    \end{itemize}
\end{itemize}

We do note that it is possible to make this study into the sensitivity, stability and variability of the online EM parameter estimation algorithm much more rigorous, for this problem or model instance, by formulating an appropriate statistical response problem \cite{majda2010quantifying}: we identify the Fisher information matrix corresponding to the parameter estimation problem at hand as a function of a perturbation induced on the initial parameter values or learning period length, under sufficient regularity conditions, and then search for its most sensitive directions which are defined by its maximal eigenvectors (i.e., those corresponding to its largest eigenvalues).

\section{Results of the Fixed-Lag Online EM Parameter Estimation Algorithm on \eqref{eq:emdyadmodel1}--\eqref{eq:emdyadmodel2}} \label{sec:app7}

In Figure \ref{fig:EM_Parameter_Estimation_Fig_Appendix}, we showcase the trace plots of the parameter value absolute errors that are produced by the adaptive-lag and fixed-lag online EM parameter estimation algorithms (for the parameters of interest). By noting that the average adaptive lag, over $[10,200]$, of the adaptive-lag online EM algorithm is about $0.1686$ time units (see Panel (c) of Figure \ref{fig:EM_Parameter_Estimation_Fig}), for the fixed-lag online EM parameter estimation algorithm we use a constant lag of $L_n\dt=0.25$ time units for all $n$, i.e., for each newly acquired observation of the observed variable $u$. This amounts to about a $50\%$ increase from the mean adaptive lag.

As already observed in the adaptive-lag online EM algorithm implementation from Section \ref{sec:4.3.2}, the observed extreme events help accelerate and correct the algorithm's divergence, regardless of whether a fixed or adaptive lag is used. Additionally, as expected, these two variations of the online EM parameter estimation method do not exhibit significant discrepancies for the forcing parameters, $F_u$ and $F_v$, and the damping parameter of the unobserved variable, $d_v$. This is especially true after the algorithm is let run for long lead times. 

However, it quickly becomes evident that the adaptive-lag variation is superior in correctly identifying the values of the parameters which are the driving force behind the significant dynamics of the model. Specifically, the adaptive-lag method is better at estimating $d_u$ and $\gamma$, and more importantly, the ratio $d_u/\gamma$; recall that whenever $v>d_u/\gamma$, this induces an instability in $u$, via anti-damping, which forms the system's intermittent extreme events. This is especially apparent during the observation of extreme events, particularly after the highly influential events A and B. As such, despite the use of a fixed lag that is $50\%$ longer than the average one found in the adaptive-lag online EM algorithm, the fact that the latter adaptively chooses longer lags during the emergence of extreme events, thus capturing the significant contributions which influence the posterior smoother state estimation, aids in its efficacious model identification. Furthermore, by using negligible lags during periods of large signal-to-noise ratios (e.g., at the demise of the extreme events in $u$), instead of a uniformly large fixed lag value throughout the run, the adaptive-lag strategy can also comparatively show significant storage savings, depending of course on the choice of its upper lag bound parameter ($b$) and the fixed lag length (as is observed from the results in Figures \ref{fig:LDA_Ice_Floes_Fig_2} and \ref{fig:Dyad_Interaction_Fig_Appendix}). 

\begin{figure}[!ht]
\centering
    \includegraphics[width=\textwidth]{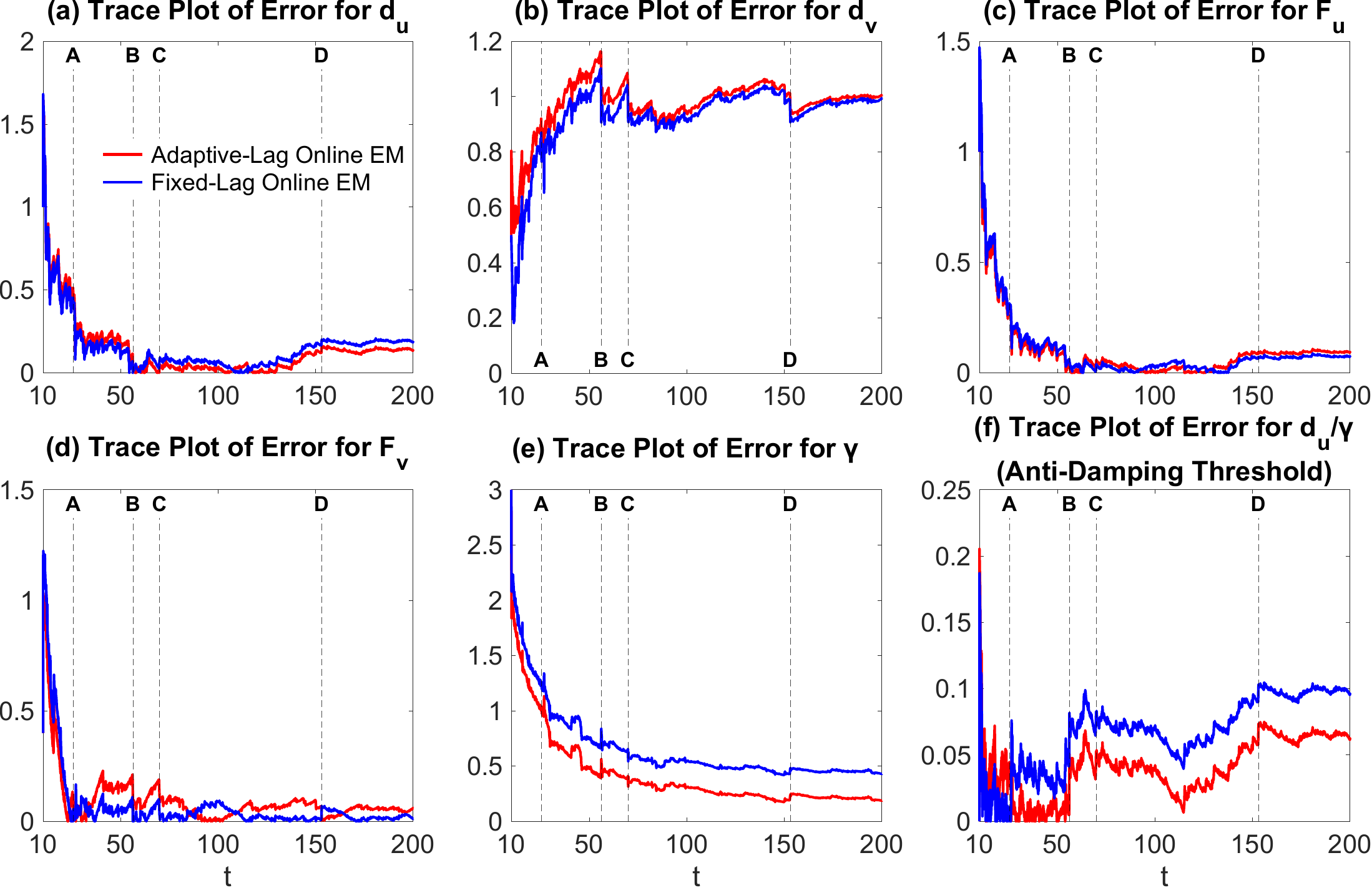}
    \caption{Panels (a)--(f): Trace plots of the parameter value absolute errors produced by the adaptive-lag (in red) and fixed-lag (in blue) online EM parameter estimation algorithms.}
    \label{fig:EM_Parameter_Estimation_Fig_Appendix}
\end{figure}

\end{appendices}

\clearpage

\bibliographystyle{unsrt}
\bibliography{references}

\end{document}